%% file: main.tex
\documentclass[11pt, letterpaper]{article}
\usepackage{graphicx}
\usepackage[
letterpaper,
top=1in,
bottom=1in,
left=1in,
right=1in]{geometry}

\input{preamble.tex}

\hypersetup{
    pdfauthor = {Haitong Liu, Deepak Narayanan Sridharan, David Steurer, Manuel Wiedmer},
}

\title{Fast algorithms for learning a Gaussian under halfspace truncation with optimal sample complexity}

\author{
  \begin{minipage}[t]{0.3\textwidth}
    \centering
    Haitong Liu \\
    \footnotesize \href{mailto:liuhai@student.ethz.ch}{liuhai@student.ethz.ch}\\
    \small ETH Zurich
  \end{minipage}
  \hfill
  \begin{minipage}[t]{0.3\textwidth}
    \centering
    Deepak Narayanan {Sridharan}\\
    \footnotesize \href{mailto:dsridharan@inf.ethz.ch}{dsridharan@inf.ethz.ch}\\
    \small ETH Zurich
  \end{minipage}
  \\\\
  \begin{minipage}[t]{0.3\textwidth}
    \centering
    David Steurer\\
    \footnotesize \href{mailto:david.steurer@inf.ethz.ch}{david.steurer@inf.ethz.ch}\\
    \small ETH Zurich
  \end{minipage}
  \hfill
  \begin{minipage}[t]{0.3\textwidth}
    \centering
    Manuel Wiedmer\\
    \footnotesize \href{mailto:manuel.wiedmer@inf.ethz.ch}{manuel.wiedmer@inf.ethz.ch}\\
    \small ETH Zurich
  \end{minipage}
}

\date{June 25, 2026}

\begin{document}

\maketitle
\thispagestyle{empty}

\begin{abstract}
We study the fundamental problem of learning a high-dimensional Gaussian truncated to an unknown halfspace. 
Lee, Mehrotra and Zampetakis (FOCS'24) recently obtained the first polynomial time algorithm for this problem, but their resulting sample and time complexity bounds are not optimal. 
Under non-trivial truncation, for any target accuracy $\varepsilon > 0$ and dimension $d$ we give an efficient algorithm that uses $n = \tilde{O}(d^2/\varepsilon^2)$ samples and learns the underlying Gaussian to error~$\varepsilon$ in total variation distance. Our algorithm is also fast: its runtime is dominated by the cost of computing the empirical covariance matrix. Both our sample and time complexity are optimal in terms of $d$ and $\varepsilon$ even \emph{without} truncation: in this regard, we can learn a Gaussian under halfspace truncation for free.

The key ingredient behind our result is a novel reinterpretation of the low-degree moments of the truncated Gaussian in terms of a relative truncation parameter. This relative truncation parameter uniquely determines the parameters of the untruncated Gaussian and enables direct parameter recovery. This reinterpretation allows us to circumvent the time intensive projected stochastic gradient descent procedure that is widely used in learning under truncation.
\end{abstract}

\setcounter{tocdepth}{2}
\newpage
\thispagestyle{empty}
\tableofcontents
\newpage

\setcounter{page}{1}

\input{content/introduction}

\input{content/techniques}

\input{content/preliminaries}

\input{content/preparation}

\input{content/algorithms}

\input{content/robustness}

\input{content/extensions}

\section*{Acknowledgments}
\addcontentsline{toc}{section}{Acknowledgments}
This research was supported by the Swiss National Science Foundation (SNSF), grant no. 10004947. We thank the anonymous reviewers for their helpful comments.

\appendix
\crefalias{section}{appendix}
\input{content/appendix}

\bibliographystyle{amsalpha}
\bibliography{refs}

\end{document}

%% file: preamble.tex
\usepackage[utf8]{inputenc}

\usepackage{amsmath, amsthm, amssymb}
\usepackage{mathtools}

\usepackage[dvipsnames]{xcolor}
\usepackage[T1]{fontenc}
\usepackage[full]{textcomp}
\usepackage[american]{babel}

\usepackage{caption}
\usepackage{booktabs}
\usepackage{siunitx}

\usepackage{newpxtext}
\usepackage{textcomp}
\usepackage[varg,bigdelims]{newpxmath}
\usepackage[scr=rsfso]{mathalfa}
\usepackage{bm}

\usepackage[pagebackref,colorlinks=true,urlcolor=blue,linkcolor=blue,citecolor=OliveGreen]{hyperref}

\usepackage{aliascnt}
\usepackage{enumitem}  
\usepackage{listings}
\usepackage{array}
\usepackage[ruled]{algorithm2e}
\usepackage{float}

\DeclareSymbolFont{bbold}{U}{bbold}{m}{n}
\DeclareSymbolFontAlphabet{\mathbbb}{bbold}

\usepackage[capitalize,nameinlink]{cleveref}

\newtheorem{theorem}{Theorem}[section]
\newtheorem*{theorem*}{Theorem}

\newaliascnt{lemma}{theorem}
\newtheorem{lemma}[lemma]{Lemma}
\aliascntresetthe{lemma}
\newtheorem*{lemma*}{Lemma}

\newaliascnt{fact}{theorem}
\newtheorem{fact}[fact]{Fact}
\aliascntresetthe{fact}
\newtheorem*{fact*}{Fact}

\newaliascnt{proposition}{theorem}

\aliascntresetthe{proposition}
\newtheorem*{proposition*}{Proposition}

\newaliascnt{corollary}{theorem}
\newtheorem{corollary}[corollary]{Corollary}
\aliascntresetthe{corollary}
\newtheorem*{corollary*}{Corollary}

\newaliascnt{claim}{theorem}
\newtheorem{claim}[claim]{Claim}
\aliascntresetthe{claim}
\newtheorem*{claim*}{Claim}

\newaliascnt{remark}{theorem}
\newtheorem{remark}[remark]{Remark}
\aliascntresetthe{remark}
\newtheorem*{remark*}{Remark}

\theoremstyle{definition}

\newaliascnt{definition}{theorem}
\newtheorem{definition}[definition]{Definition}
\aliascntresetthe{definition}
\newtheorem*{definition*}{Definition}

\crefname{claim}{Claim}{Claims}
\Crefname{claim}{Claim}{Claims}

\usepackage{boxedminipage}
\newcommand{\paren}[1]{(#1)}
\newcommand{\Paren}[1]{\left(#1\right)}

\newcommand{\brac}[1]{[#1]}
\newcommand{\Brac}[1]{\left[#1\right]}

\newcommand{\abs}[1]{\lvert#1\rvert}
\newcommand{\Abs}[1]{\left\lvert#1\right\rvert}

\newcommand{\Set}[1]{\left\{#1\right\}}

\newcommand{\norm}[1]{\lVert#1\rVert}
\newcommand{\Norm}[1]{\left\lVert#1\right\rVert}

\newcommand{\iprod}[1]{\langle#1\rangle}
\newcommand{\Iprod}[1]{\left\langle#1\right\rangle}

\newcommand{\Esymb}{\mathbb{E}}
\newcommand{\Psymb}{\mathbb{P}}
\newcommand{\Rsymb}{\mathbb{R}}

\DeclareMathOperator*{\E}{\Esymb}

\DeclareMathOperator*{\ProbOp}{\Psymb}
\renewcommand{\Pr}{\ProbOp}

\newcommand\bdot\bullet

\newcommand{\ind}[1]{\mathbbb{1}\Brac{#1}}

\DeclareMathOperator{\poly}{poly}

\DeclareMathOperator{\polylog}{polylog}

\newcommand{\iid}{i.i.d.\xspace}

\newcommand{\Z}{\mathbb Z}
\newcommand{\N}{\mathbb N}
\newcommand{\R}{\mathbb R}

\newcommand{\x}{X}
\newcommand{\y}{Y}

\newcommand{\cD}{\mathcal D}

\newcommand{\cN}{\mathcal N}

\newcommand{\cS}{\mathcal S}

\newcommand{\bbE}{\mathbb E}

\renewcommand{\leq}{\leqslant}
\renewcommand{\le}{\leqslant}
\renewcommand{\geq}{\geqslant}
\renewcommand{\ge}{\geqslant}
\let\epsilon=\varepsilon
\numberwithin{equation}{section}

\newcommand{\eps}{\epsilon}

\allowdisplaybreaks
\newcommand*{\Lowner}{L\"owner\xspace}

\renewcommand{\epsilon}{\ensuremath\varepsilon}

\renewcommand{\phi}{\ensuremath{\varphi}}

\newcommand{\vnorm}[1]{\left\lVert#1\right\rVert}

\newcommand{\normal}[1]{\mathcal{N}\left(#1\right)}

\newcommand{\sk}{\textup{skew}}
\newcommand{\tr}{\mathrm{tr}}

\newcommand{\truncatedGauss}[3]{\cN(#1,#2)|_{#3}}
\newcommand{\TVdist}{d_\mathrm{TV}}

\newcommand{\truncparaword}{relative truncation parameter}
\newcommand{\Truncparaword}{Relative truncation parameter}
\newcommand{\truncpara}{\gamma}

\newcommand\tp[2][-6]{{#2}^{\mkern#1mu\top}}

\newcommand{\ddt}{\frac{\mathrm{d}}{\mathrm{d}t}}
\newcommand{\ddx}{\frac{\mathrm{d}}{\mathrm{d}x}}

\definecolor{codegreen}{rgb}{0,0.6,0}
\definecolor{codegray}{rgb}{0.5,0.5,0.5}
\definecolor{codepurple}{rgb}{0.58,0,0.82}
\definecolor{backcolour}{rgb}{0.95,0.95,0.92}
\usepackage{algpseudocode}

\lstdefinestyle{mystyle}{
    backgroundcolor=\color{backcolour},   
    commentstyle=\color{codegreen},
    keywordstyle=\color{magenta},
    numberstyle=\tiny\color{codegray},
    stringstyle=\color{codepurple},
    basicstyle=\ttfamily\footnotesize,
    breakatwhitespace=false,         
    breaklines=true,                 
    captionpos=b,                    
    keepspaces=true,                 
    numbers=left,                    
    numbersep=5pt,                  
    showspaces=false,                
    showstringspaces=false,
    showtabs=false,                  
    tabsize=2,
    language=Mathematica
}

%% file: content/introduction.tex
\newcommand{\package}{\emph}

\section{Introduction}
\label{sec:introduction}
We study statistical estimation under truncation. In this setting, data is only observed when it falls in some survival set.
There is a long line of work on learning based on truncated or censored samples; dating back to at least the work of Daniel Bernoulli \cite{Bernoulli:smallpoxvaccination}.
Bernoulli studied the effectiveness of the smallpox vaccination and encountered the problem of missing data from people who died of smallpox; thus he only had access to `truncated' data.
Since then, statistical estimation from truncated or censored samples has been studied by numerous scientists, including \cite{Galton:horses}, \cite{Pearson:truncated}, \cite{PearsonLee:truncated}, \cite{Lee:truncated}, \cite{Fisher:truncated}, and \cite{BlissStevens:censored}; we refer the reader to the book \cite{Cohen:book} for a more comprehensive treatment.
Importantly however, all these works either study the one-dimensional setting or utilize methods that are not computationally efficient in higher dimensions.

In the last decade, the following setting has received renewed attention:
Given samples from a multivariate Gaussian with unknown parameters truncated to some set $S$, learn the Gaussian in total variation (TV) distance.
We note that for this, it is enough to learn the mean and covariance up to small error.
\cite{DGTZ18:efficientforknownset} showed that, given oracle access to the truncation set, the parameters can be learned in TV-distance up to error $\varepsilon$ in polynomial time and in sample complexity $\tilde{O}(d^2/\varepsilon^2)$.
This sample complexity is optimal up to logarithmic factors since already for learning the covariance of a Gaussian, even without truncation, $n = \Omega(d^2/\varepsilon^2)$ samples are information-theoretically necessary (cf. \cite{cai2010optimal}).

If the truncation set is unknown and {\em arbitrary}, \cite{DGTZ18:efficientforknownset} also showed that learning up to arbitrarily small error in TV-distance is impossible.
\cite{kontonis2019} studied the more tractable setting where the truncation set is unknown but belongs to some well-behaved class of possible truncation sets.
They are able to learn the parameters of a Gaussian with \emph{diagonal} covariance matrix given a bound on the Gaussian surface area of the truncation set.
If the Gaussian surface area is at most $\Gamma$, then they can learn the Gaussian up to TV-error $\varepsilon$ with time and sample complexity $d^{\poly(1/\varepsilon) \cdot \Gamma^2}$.
Without further assumptions on the truncation set, \cite{diakonikolas2024} showed that the exponential dependence on $1/\eps$ is necessary in the SQ model.

More recently, \cite{lee2024} extended the result of \cite{kontonis2019} to general Gaussians.
Furthermore, they showed that for two simple classes of truncation sets, namely, axis-aligned rectangles and halfspaces, one can even learn in time and sample complexity $\poly(d/\varepsilon)$.
This is the first algorithm that achieves fully polynomial runtime for learning a Gaussian under truncation to some unknown survival set.

In this work, we continue the investigation of learning from truncated samples for `simple' survival sets. We focus on the case of halfspace truncation.
For halfspaces, the result by \cite{lee2024} needs at least $\Omega(d^3)$ samples and might not be practical.
This is because an important step of their algorithm is based on projected stochastic gradient descent, for which they use the ellipsoid method for the projection step.
It is important to note that the projection is \emph{not} to the survival set, but instead to some set in the parameter space.
Therefore, even for simple classes it is unclear whether this step can be made fast.
Concretely, a direct application of the ellipsoid method (without using any specifics about the problem) yields a lower bound of $\Omega(d^{10})$ on the runtime of their algorithm (\cite[Algorithm 3]{lee2024}).\footnote{Algorithm 3 in \cite{lee2024} runs the ellipsoid method on a problem in $m = \Theta(d^2)$ dimensions. The ellipsoid method needs $\Omega(m^2)$ iterations to terminate and each iteration takes $\Omega(m^2)$ operations (as we need to update the $m \times m$ matrix describing the ellipsoid). Thus, one run of the ellipsoid method needs $\Omega((d^2)^4) = \Omega(d^8)$ time. Finally, they run the ellipsoid method $\Omega(d^2)$ times, giving a runtime lower bound of $\Omega(d^{10})$.}
Thus, their paper leaves open the following natural question for halfspace truncation:
\begin{center}
    \textit{Are there fast algorithms with optimal sample complexity, matching the untruncated setting?}
\end{center}
Naturally, the best one could hope for is to match the results of the untruncated setting --- which has been studied for decades and is well-understood.
There, the best known algorithm is to compute the sample mean and the sample covariance; see e.g. \cite{cai2010optimal}.
To achieve error $\varepsilon$, one needs $N = \Theta(d^2/\varepsilon^2)$ many samples.
The time needed to compute this estimator is asymptotically the same as the time needed to multiply a $d \times N$ matrix with its transpose, which we denote by $T(N,d)$.
Currently, the best-known runtime for this is $O(d^{3.250035}/\varepsilon^2)$ as a result of recent advances in fast matrix multiplication due to \cite{AlmanDuanWilliamsXuXuZhou:fastmatrixmultplication}.

\subsection{Main contributions}
In this paper, we answer the above question in the affirmative by designing an algorithm that acheives the following guarantees.
\begin{theorem}[Informal, see \Cref{thm:main}]\label{THM:INTRO:main}
    There is an algorithm with the following guarantee:
    The algorithm takes as input~$N$ samples from a $d$-dimensional Gaussian $\normal{\mu, \Sigma}$ truncated to some halfspace of mass $\alpha$ as well as a lower bound $\alpha_0$ on $\alpha$.
    If ${N = d^2 \cdot \poly(1/\varepsilon, \log(1/\alpha_0))}$, the algorithm outputs $\hat{\mu}$ and $\widehat{\Sigma}$ such that, with probability at least~0.99, we have:
    \[
        \TVdist\left(\normal{\mu, \Sigma}, \normal{\hat{\mu}, \widehat{\Sigma}}\right) \leq \varepsilon.
    \]
    The time complexity of the algorithm is $O(T(N,d)) \leq O(d^{3.250035}) \cdot \poly(1/\varepsilon, \log(1/\alpha_0))$.
\end{theorem}

\begin{remark*}
    If $\alpha \leq 0.99$, then our algorithm needs only $N = d^2/\varepsilon^2 \cdot \polylog(1/\alpha_0)$ samples; otherwise it needs $N = O(d^2/\varepsilon^4)$ many samples. 
    The algorithm is the same for both cases.

    We also note that the `$0.99$' both in the condition on $\alpha$ and in the success probability can be replaced by arbitrary constants $c < 1$.
\end{remark*}
Our algorithm is the first algorithm that achieves optimal sample complexity dependence on~$d$.
The runtime of our algorithm matches the runtime of computing the sample covariance in the untruncated setting and hence is also optimal in terms of $d$.
Furthermore, the dependence of the sample and time complexity on $\varepsilon$ are optimal whenever there is non-trivial truncation (meaning, $\alpha \leq 0.99$).

Regarding our dependence of the sample complexity on $\alpha_0$, we remark that all prior works, including inefficient algorithms, also have a dependence on $\alpha_0$.
In fact, under the assumption that the set is unknown, all previous \emph{efficient} algorithms incur either an exponential dependence on $1/\alpha_0$ or work under the assumption that $\alpha_0 \geq \Omega(1)$ (information-theoretically, \cite{kontonis2019} showed that one can achieve polynomial dependence).
Interestingly, in the case where the set is known, more recently, \cite{KaratapanisKontonisTzamos:polylogalphaknown} showed that one can achieve $\polylog(1/\alpha_0)$ dependence on $\alpha_0$ for the sample complexity (their runtime dependence on $1/\alpha_0$ is polynomial).
In the case of halfspaces, our algorithm matches this guarantee even when the set is \emph{unknown}.
This is the first efficient algorithm to get better than exponential dependence on $1/\alpha_0$ when the set is unknown.
An overview of the known algorithms for learning a Gaussian under halfspace truncation can be found in \Cref{tab:overview}.
We note that all of these cited results are more general than halfspaces; in order to be able to compare them to our result, the results mentioned in the table are specialized to the case of halfspace truncation.

\begin{table}[H]
\centering
\caption{Overview of efficient algorithms for learning a Gaussian under halfspace truncation}
\label{tab:overview}
\vspace{1mm}
\small
\setlength{\tabcolsep}{3pt}
\setlength{\extrarowheight}{2pt}
\begin{tabular}{|l|c|c|c|c|}
\hline
Paper & \begin{tabular}[c]{@{}c@{}}Truncation set\\(known/unknown)\end{tabular} & \begin{tabular}[c]{@{}c@{}}Underlying\\distribution\end{tabular} & \begin{tabular}[c]{@{}c@{}}Sample complexity\\(in $d$ and $\varepsilon$)\end{tabular}& \begin{tabular}[c]{@{}c@{}}Dependence\\on $\alpha_0$\end{tabular} \\\hline

\cite{DGTZ18:efficientforknownset} & Known & $\normal{\mu,\Sigma}$ & $\tilde{O}(d^2/\varepsilon^2)$ & \begin{tabular}[c]{@{}c@{}}Assumption:\\$\alpha_0 \geq \Omega(1)$\end{tabular} \\ \hline

\cite{kontonis2019} & Unknown & \begin{tabular}[c]{@{}c@{}}$\normal{\mu,\Sigma}$\\for $\Sigma$ diagonal\end{tabular} & $d^{\poly(1/\varepsilon)}$ & $\exp(1/\alpha_0)$ \\ \hline

\cite{lee2024} & Unknown & $\normal{\mu,\Sigma}$ & $\mathrm{poly}(d,1/\varepsilon)$ & $\exp(1/\alpha_0)$ \\ \hline

\cite{KaratapanisKontonisTzamos:polylogalphaknown} & Known & $\normal{\mu, \Sigma}$ & $O(d^2/\varepsilon^2)$ & $\poly\log(1/\alpha_0)$ \\ \hline

{\hypersetup{linkcolor=red}
\hyperref[THM:INTRO:main]{This work}} & Unknown & $\normal{\mu,\Sigma}$ & \begin{tabular}{c l}$O(d^2/\varepsilon^2)$ &if $\alpha \leq 0.99$\\$O(d^2/\varepsilon^4)$ &otherwise\end{tabular} & $\poly\log(1/\alpha_0)$ \\ \hline

\end{tabular}
\end{table}

\paragraph{High-level overview of our techniques and comparison to prior work.} 

At a high level, our algorithm is based on the method of moments. 
We use the first three moments of the distribution $\truncatedGauss{\mu}{\Sigma}{\iprod{w,x}\leq\tau}$\footnote{Here, $\truncatedGauss{\mu}{\Sigma}{\iprod{w,x}\leq\tau}$ denotes the Gaussian distribution~$\cN(\mu, \Sigma)$ truncated to the halfspace $\{x : \langle w, x \rangle \leq \tau\}$ for some unit vector $w$.} to estimate a \emph{\truncparaword}.
Our key insight is that this \truncparaword{} \emph{fully} characterizes how truncation changes the moments compared to the untruncated distribution.
This enables us to directly estimate the parameters $\mu$ and $\Sigma$ of the Gaussian from the moments themselves.

We now compare our approach to \cite{lee2024}, who gave the first efficient algorithm for learning the parameters of a Gaussian under halfspace truncation.
They use the method of moments to first learn the \emph{truncation set}. Then, they\ apply projected stochastic gradient descent to estimate the parameters.
At a high level, their proof can be divided into the following four steps:
\begin{enumerate}[noitemsep]
    \item First, they use the third central moment to find a vector close to $u = \frac{\Sigma w}{\Norm{\Sigma w}}$.\label{LMZ:step1}
    \item Then, they compute a vector close to the vector $w$ of the halfspace.\footnote{In their paper, \cite{lee2024} claim that the third central moment $M_3$ is proportional to \[M_3 \coloneqq \mathbb{E}[\Paren{X-\mathbb{E}[X]}^{\otimes 3}] \propto w^{\otimes 3}.\] They use this to determine $w$. However, there is small gap in this argument: As we show later, the third central moment is in fact proportional to $(\Sigma w)^{\otimes 3}$ (and not $w^{\otimes 3}$); see \Cref{lem:first3moments}. Since our goal is to estimate $\Sigma$, the third central moment therefore \emph{cannot} be directly used to estimate $w$. Instead, one can compute a vector $u$ proportional to $\Sigma w$ using the third central moment. Then, one can solve the following linear system involving the second central moment $M_2$ to recover $w$. Namely, it holds that \[M_2^{-1}u \propto w.\] We include a proof of this for completeness in \Cref{APP:howtogetw}.}\label{LMZ:step2}
    \item Next, they find a value close to $\tau$ by picking the smallest threshold such that the resulting set contains all the samples.\label{LMZ:step3}
    \item Finally, given the estimate for the truncation set, they use their general projected stochastic gradient descent procedure to learn the parameters.\label{LMZ:step4}
\end{enumerate}
While our approach is also based on the method of moments, we use the moments for a different purpose: rather than estimating the truncation set as in \cite{lee2024}, we use the moments to \emph{directly} estimate the parameters of the \mbox{underlying} Gaussian. By using the method of moments for direct parameter recovery, we are able to get a fast algorithm: crucially, we can avoid the time-\hspace{0pt}consuming projected stochastic gradient step (step \ref{LMZ:step4}) and achieve optimal runtime in terms of~$d$.

As a first step in our algorithm, we also estimate a vector $\hat{u}$ as \cite{lee2024} do in step~\ref{LMZ:step1}.
However, we then do not use it to estimate the vector $w$.
Instead, the key insight in our algorithm is that the first non-central moment and the second and third central moments all share an important structural property:
They are equal to the untruncated moments plus an additional term involving \emph{only} the following vector $\Sigma w/\Norm{\Sigma^{1/2} w}$ and the \emph{same} \truncparaword{} $\truncpara$.
Importantly, we show that the first three moments \emph{uniquely} determine $\truncpara$ and $\|\Sigma w\|/\norm{\Sigma^{1/2}w}$.
This allows us to use the moments to get an estimate for $\truncpara$ and $\|\Sigma w\|/\norm{\Sigma^{1/2}w}$.
Combining this with our first step enables us to directly compute estimates for the parameters.
Together, these ideas yield a simple and fast algorithm for estimating the parameters of a halfspace truncated Gaussian. 

Regarding the sample complexity, \cite{lee2024} require a high-accuracy estimate of the third moment in their step~\ref{LMZ:step2}.
This step requires at least $\Omega(d^3)$ samples.
In contrast, for our algorithm, when computing a vector close to $u = \Sigma w/\norm{\Sigma w}$, we use a random contraction of the third central moment and then use a spectral decomposition of the resulting matrix, which allows us to get optimal sample complexity.

\paragraph{Robustly learning the parameters of a Gaussian under halfspace truncation.}
In contrast to the algorithm by~\cite{lee2024}, we utilize the moments of the truncated distribution to directly estimate the parameters. 
Motivated by this, we introduce the problem of \emph{robust estimation of a truncated Gaussian}.
We consider the following corruption model: First, samples from a halfspace truncated Gaussian $\truncatedGauss{\mu}{\Sigma}{\iprod{w,x}\leq\tau}$ are drawn. Then, an adversary is allowed to modify an \mbox{$\eta$-fraction} of the samples arbitrarily; in particular the samples \emph{do not} need to lie in the set $\{x \in \R^d \mid \iprod{w,x}\leq\tau\}$ anymore. The adversary is computationally unbounded, knows the truncation set as well as the parameters and can adaptively change the samples after they have been drawn.
We believe that this corruption model, which is analogous to the strong contamination model in robust statistics \cite{diakonikolaskane2023algorithmic}, might be of independent interest for learning under truncation.

Given that our algorithm only uses the first three moments, we can robustify it by using robust estimates for the moments instead of relying on the sample moments.
Our robust algorithm uses a black-box robust estimator due to \cite{KothariSteinhardtSteurer:robustmomentestiamtion} and the fact that the truncated distribution is sub-Gaussian.
We prove the following result.
To simplify the exposition, we state our result for constant $\alpha$ and for $\Sigma \succeq \Omega(1) \cdot I_d$.
The full result can be found in \Cref{THM:robust} in \Cref{sec:robustness}.
\begin{theorem}[Informal, see \Cref{THM:robust}]\label{THM:INTRO:robust}
    Consider the truncated distribution $\truncatedGauss{\mu}{\Sigma}{\iprod{w,x}\leq\tau}$, where the mass $\alpha$ of the survival set satisifies $\Omega(1) \leq \alpha \leq 0.99$.
    Assume $\Sigma \succeq \Omega(1) \cdot I_d$.
    Let $0 < \delta < 1/2$ be an arbitrary constant. 
    Let $\eta > 0$ be the corruption parameter and assume that $\eta \leq 1/d^{1-2\delta}$.
    For $n = \poly(d)$, given an $\eta$-corruption $X_1, \ldots, X_n$ of $\truncatedGauss{\mu}{\Sigma}{\{w^\top x \leq \tau\}}$, we can compute, in polynomial time in $d$ and with probability at least $0.99$, estimators $\hat{\mu}$ and $\widehat{\Sigma}$ such that
    \[
        \Norm{\Sigma^{-1/2}(\hat{\mu} - \mu)} \leq O_\delta\Paren{\sqrt{d} \cdot \eta^{1/2-\delta}}
    \]
    and
    \[
        \norm{\Sigma^{-1/2}\widehat{\Sigma}\Sigma^{-1/2}-I_d}_F
        \leq O_\delta\Paren{\sqrt{d} \cdot \eta^{1/2-\delta}}.
    \]
\end{theorem}
While our algorithm can only tolerate a very small amount of corruption, we believe that already this is non-trivial because we are \emph{not} trying to estimate the moments of the distribution to which we have (sample) access to (which would be the typical goal in robust statistics) but instead of the underlying untruncated distribution.
In particular, the adversary can place samples outside the survival set. Thus, it is unclear if one can even estimate the truncation set.
We remark that previous work in truncated statistics cannot (easily) be made robust. For example, even if one could learn the truncation set under adversarial corruption, it is unclear whether the algorithm of \cite{lee2024} can be made to succeed, since it is not clear how to robustify their final step involving projected stochastic gradient descent.

That being said, we believe that it is likely that the amount of corruption and the error rate can be improved.
We leave it as an open question for future work whether using a more specialized robust estimator for the moments can improve \Cref{THM:INTRO:robust}.
We believe that our result can act as a starting point for further investigation of robust estimation given only truncated samples.
In particular, our result also leaves open what happens in weaker corruption models (e.g., a model where only oblivious corruptions are allowed).

\paragraph{Extensions to other truncation sets.}
Finally, we show that our techniques can be used beyond halfspace truncation.
In \Cref{thm:intersec-of2orthogonal-halfspaces}, we show that using similar ideas we can also learn the mean of an identity-covariance Gaussian $\normal{\mu, I_d}$ truncated to an \emph{intersection of two orthogonal halfspaces}.
Similar to our main result, we can achieve sample complexity ${O\Paren{\frac{d^2}{\varepsilon^2}} \cdot \polylog(1/\alpha_0) + O\Paren{\frac{d^2}{\varepsilon^4}}}$.
Our algorithm for this case also works using only the first three moments of the truncated distribution.
At a very high level, instead of recovering one \truncparaword{}, for the case of truncation by an intersection of two orthogonal halfspaces, we recover two \truncparaword s; one for each truncation direction. 
For details, we refer to \Cref{sec:extensions}.

While we only prove the case of truncation by an intersection of 2 orthogonal halfspaces, we believe that our proof can be extended to an intersection of $k$ (pairwise) orthogonal halfspaces, yielding a sample complexity of $O\Paren{\frac{d^2 \cdot \poly(k)}{\varepsilon^2}} \cdot \polylog(1/\alpha_0) + O\Paren{\frac{d^2 \cdot \poly(k)}{\varepsilon^4}}$; see \Cref{rem:furtherextensions}.

\subsection{Related work}
\paragraph{Method of Moments.}
Solving problems in statistical estimation using the Method of Moments (MoM) has a long history. It can be dated back to the famous statistician Karl Pearson~\cite{pearson1894contributions}. This method has led to several recent advances in high-dimensional parameter estimation tasks, including those involving learning mixture models~\cite{moitra2010settling, liu2022clustering} and independent component analysis (ICA)~\cite{belkin2013blind, bhaskara2014uniqueness}. In this work we adopt the MoM for learning the parameters of a halfspace truncated Gaussian.

\paragraph{Truncated Learning.}
Beyond the results described earlier, estimation under truncation has seen applications to \emph{differentially private} statistical estimation~\cite{zampetakis2025private}. There has also been recent work on mean testing under  truncation~\cite{canonne2025gaussian} and on detecting truncation~\cite{de2023testing, de2024detecting}. A recent work~\cite{lee2025learning} studies learning under positive and imperfectly labeled samples and extends the general learning algorithm of~\cite{lee2024} for learning under general truncation sets.

\paragraph{Robust Estimation.} Robust parameter estimation has a long history in statistics~\cite{huber1992robust, huber2004estimation}. Recent algorithmic breakthroughs by~\cite{lai2016agnostic, diakonikolas2019robust} have led to the new subfield of algorithmic robust statistics, which primarily deals with parameter estimation in the presence of adversarial outliers (see the textbook~\cite{diakonikolaskane2023algorithmic} for a comprehensive treatment). A very useful method that has emerged for robust statistical estimation is the powerful Sum-of-Squares hierarchy of semi-definite programs~\cite{barak2016proofs} that often efficiently achieves the best possible statistical guarantees for a large range of robust estimation tasks~(see e.g. \cite{KothariSteinhardtSteurer:robustmomentestiamtion}). In this work, we use Sum-of-Squares based robust moment estimation algorithms as a black-box to solve the robust version of the learning under truncation problem that we introduce.

%% file: content/techniques.tex
\section{Technical overview}\label{sec:techniques}
In order to prove our main result \Cref{THM:INTRO:main}, we want to use the moments of the truncated distribution to estimate the parameters.
Consider the distribution $\truncatedGauss{\mu}{\Sigma}{\{w^\top x \leq \tau\}}$, where $w$ is a unit vector.
Recall that our goal is to find estimates $\hat{\mu}$ and $\widehat{\Sigma}$ such that
\[
    \TVdist\left(\normal{\mu, \Sigma}, \normal{\hat{\mu}, \widehat{\Sigma}}\right) \leq \varepsilon.
\]
To show closeness in total variation distance as above, it is enough to show
\[
    \norm{\Sigma^{-1/2}(\hat{\mu} - \mu)} \leq \varepsilon \quad \text{and} \quad \norm{\Sigma^{-1/2} \widehat{\Sigma} \Sigma^{-1/2} - I_d}_F \leq \varepsilon,
\]
see e.g. \cite{diakonikolas2019robust}.
To achieve this, on a high level, we first determine the moments of the truncated distribution.
Then, we observe that the moments are determined by $\mu$ and $\Sigma$ as well as a \emph{\truncparaword}~$\truncpara$ and a certain vector.
We show that if we can determine this parameter and vector, then based on the first non-central and the second central moment of the truncated distribution, we would be able to determine the $\mu$ and $\Sigma$.
Finally, we show how to determine $\truncpara$ and this vector based on the second and third central moment.

\paragraph{The moments of the truncated distribution.}
As a warm up, we start by considering the truncated distribution~${Y \sim \truncatedGauss{0}{I_d}{\{e_1^\top x \leq \truncpara\}}}$.
Intuitively, the moments of this distribution should be the ones of the standard Gaussian, but shifted in the direction of $e_1^{\otimes k}$ since $e_1$ is the direction in which the distribution is truncated.\footnote{We remark that this is only true since the coordinates are independent. As we will see later, for $\truncatedGauss{\mu}{\Sigma}{\{e_1^\top x \leq \truncpara\}}$, the moments would be shifted in the direction of $(\Sigma e_1)^{\otimes k}$ and not $e_1^{\otimes k}$.}
For example, $\bbE[Y] = c_1 e_1$ since in all but the first coordinate, $Y$ is a standard Gaussian.
More precisely, $c_1 = \kappa_1(\truncpara)$ is the first cumulant\footnote{The first and second cumulants are just the mean and the covariance. The third cumulant is the third central moment. For more details, see \Cref{DEF:cumulant}.} of a one-dimensional standard Gaussian truncated to $(-\infty, \truncpara]$.
More generally, higher moments can be related to the higher order cumulants of this one-dimensional distribution.

To compute the mean (and more generally the moments) of a general truncated Gaussian ${X \sim \truncatedGauss{\mu}{\Sigma}{\{w^\top x \leq \tau\}}}$, our goal is to relate it to the mean (and the corresponding moments) of~$Y$.
To this end, consider the linear transformation
\[
    y \mapsto \psi(y) = \mu + \Sigma^{1/2} U y,
\]
where $U$ is a rotation matrix. For any $U$, this map transforms a standard Gaussian $\normal{0, I_d}$ to $\normal{\mu, \Sigma}$.
We want to pick $U$ and $\gamma$ in a way such that the truncation set of $\psi(Y)$ is equal to the truncation set of $X$.
This will show that $\psi(Y)$ and $X$ have the same distribution, thus allowing us to express the moments of $X$ in terms of the moments of $Y$.

First, we show how to pick $U$ such that the direction of truncation of $\psi(Y)$ is equal to $w$.
Since $Y$ is truncated to $\{x \mid e_1^\top x \leq \truncpara\}$, $\psi(Y)$ is truncated to the set
\begin{align}
    \{\psi(y) \mid y \text{ such that } e_1^\top y \leq \truncpara\} &= \{x \mid e_1^\top \psi^{-1}(x) \leq \truncpara\}\nonumber\\
    &= \{x \mid e_1^\top U^{-1}\Sigma^{-1/2}(x-\mu) \leq \truncpara\}\nonumber\\
    &= \{x \mid e_1^\top U^\top \Sigma^{-1/2}x \leq \truncpara + e_1^\top U^\top \Sigma^{-1/2} \mu\}\label{EQ:TO:truncationsetaftertransform}
\end{align}
Thus, we want to pick $U$ such that
\[
    \Sigma^{-1/2}Ue_1 \propto w \: \Longleftrightarrow \: Ue_1 \propto \Sigma^{1/2}w.
\]
As a rotation matrix, $U$ preserves the norm and hence
\begin{equation}\label{EQ:TO:definitionofU}
    Ue_1 = \frac{\Sigma^{1/2}w}{\Norm{\Sigma^{1/2}w}}.
\end{equation}
In order to determine the value of $\truncpara$ needed such that $\psi(Y)$ has the same distribution as $X$, we want to write \eqref{EQ:TO:truncationsetaftertransform} in the form $\{x \mid w^\top x \leq \tau\}$.
By \eqref{EQ:TO:definitionofU}, we can write the truncation set for $Y$ as
\[
    \{\psi(y) \mid y \text{ such that } e_1^\top y \leq \truncpara\} = \{x \mid w^\top x \leq \norm{\Sigma^{1/2}w}\truncpara + w^\top \mu\}.
\]
Hence, we need that
\begin{equation}\label{EQ:TO:defoftruncpara}
    \truncpara = \frac{\tau - w^\top \mu}{\norm{\Sigma^{1/2} w}}.
\end{equation}
For $U$ as in \eqref{EQ:TO:definitionofU} and $\truncpara$ as in \eqref{EQ:TO:defoftruncpara}, $\psi(Y) \sim \truncatedGauss{\mu}{\Sigma}{\{w^\top x \leq \tau\}}$.
In particular, we can now compute $\bbE[X]$ using this transformation $\psi$ and $\bbE[Y]$:
\[
    \bbE[X] = \mu + \Sigma^{1/2}U \bbE[Y] = \mu + \kappa_1(\truncpara) \cdot \Sigma^{1/2}Ue_1 = \mu + \kappa_1(\truncpara) \cdot \frac{\Sigma w}{\norm{\Sigma^{1/2}w}}.
\]
Analogously we can also compute the higher moments of $Y$ and use them, together with the the transformation $\psi$ to compute the moments of $X$.
Namely, we have that
\begin{align}
    \bbE[X] &= \mu + \kappa_1(\truncpara) \cdot \frac{\Sigma w}{\vnorm{\Sigma^{1/2} w}}\label{EQ:TO:firstmoment}\\
    \bbE[(X-\bbE[X])^{\otimes 2}] &= \Sigma - (1- \kappa_2(\truncpara)) \cdot \Paren{\frac{\Sigma w}{\norm{\Sigma^{1/2} w}}}^{\otimes 2}\label{EQ:TO:secondmoment}\\
    \bbE[(X-\bbE[X])^{\otimes 3}] &= \kappa_3(\truncpara) \cdot \Paren{\frac{\Sigma w}{\norm{\Sigma^{1/2} w}}}^{\otimes 3}\label{EQ:TO:thirdmoment}
\end{align}
For more details on this computation, see \Cref{lem:first3moments}.
In particular, these moments are determined by the parameters of the untruncated Gaussian, the vector $\Sigma w/\norm{\Sigma^{1/2} w}$ and $\truncpara$ defined as in \eqref{EQ:TO:defoftruncpara}.
We will refer to $\truncpara = \truncpara(\mu, \Sigma, w,\tau)$ as the \emph{\truncparaword} of the distribution $\truncatedGauss{\mu}{\Sigma}{\{w^\top x \leq \tau\}}$.
In particular, since the third moment is rank 1, if we would know the \truncparaword{}, then we could use \eqref{EQ:TO:thirdmoment} to estimate $\Sigma w/\norm{\Sigma^{1/2} w}$.
Once we determined this vector, we could then use \eqref{EQ:TO:firstmoment} and \eqref{EQ:TO:secondmoment} respectively to estimate $\mu$ and $\Sigma$.
To summarize, since we can approximate the moments of $\truncatedGauss{\mu}{\Sigma}{\{w^\top x \leq \tau\}}$ by the sample moments, the main question is how to estimate the \truncparaword{} defined in \eqref{EQ:TO:defoftruncpara}.

\paragraph{First three moments determine the \truncparaword{}.}
Before showing how we estimate the \truncparaword{}~$\truncpara$, we want to establish that the information in the first three moments is in fact enough to uniquely determine $\truncpara$.
For this, consider two truncated Gaussians $\truncatedGauss{\mu_1}{\Sigma_1}{\{w_1^\top x \leq \tau_1\}}$ and $\truncatedGauss{\mu_2}{\Sigma_2}{\{w_2^\top x \leq \tau_2\}}$.
We want to show that if the first non-central moment as well as the second and third central moments of these two distributions are identical (cf. \eqref{EQ:TO:firstmoment}, \eqref{EQ:TO:secondmoment} and \eqref{EQ:TO:thirdmoment}), then we need to have 
$\gamma_1 = \gamma_2$.
This will show that it is information-theoretically possible to determine the \truncparaword{} based only on these low-degree moments.

To formally show that the first three moments of the two distributions $\truncatedGauss{\mu_1}{\Sigma_1}{\{w_1^\top x \leq \tau_1\}}$ and $\truncatedGauss{\mu_2}{\Sigma_2}{\{w_2^\top x \leq \tau_2\}}$ being equal implies that $\truncpara_1 = \truncpara_2$, consider first the second moment.
By assumption, we have that
 \[
  \Sigma_1 - (1- \kappa_2(\truncpara_1)) \cdot \Paren{\frac{\Sigma_1 w_1}{\vnorm{\Sigma_1^{1/2} w_1}}}^{\otimes 2} = \Sigma_2 - (1- \kappa_2(\truncpara_2)) \cdot \Paren{\frac{\Sigma_2 w_2}{\vnorm{\Sigma_2^{1/2} w_2}}}^{\otimes 2}
 \]
By multiplying from the left by $w_1^\top$ and from the right by $w_2$ as well as identifying $a^{\otimes 2}$ with $aa^\top$, we get that
\[
    w_1^\top \Sigma_1 w_2 - \frac{(1- \kappa_2(\truncpara_1))}{\vnorm{\Sigma_1^{1/2}w_1}^2} \cdot w_1^\top \Sigma_1 w_1 \cdot  w_1^\top \Sigma_1 w_2 =  w_1^\top \Sigma_2 w_2 - \frac{(1- \kappa_2(\truncpara_2))}{\vnorm{\Sigma_2^{1/2}w_2}^2} \cdot w_1^\top \Sigma_2 w_2 \cdot w_2^\top \Sigma_2 w_2.
\]
Using that $w_i^\top \Sigma_i w_i = \vnorm{\Sigma_i^{1/2} w_i}^2$, the above simplifies to
\begin{equation}\label{eqn:secmomentiden}
    \kappa_2(\truncpara_1) \cdot w_1^\top \Sigma_1 w_2 = \kappa_2(\truncpara_2) \cdot w_1^\top \Sigma_2 w_2.
\end{equation}
Our goal now is to relate the two numbers $w_1^\top \Sigma_1 w_2$ and $w_1^\top \Sigma_2 w_2$. We do so using the third central moment (cf. \eqref{EQ:TO:thirdmoment}).
Since we assumed that the third moments also agrees, we have that
\[
    \kappa_3(\truncpara_1) \cdot \Paren{\frac{\Sigma_1 w_1}{\norm{\Sigma_1^{1/2} w_1}}}^{\otimes 3} = \kappa_3(\truncpara_2) \cdot \Paren{\frac{\Sigma_2 w_2}{\norm{\Sigma_2^{1/2} w_2}}}^{\otimes 3}
\]
This is a rank-$1$ tensor, so by considering the diagonal elements, we get that also the corresponding vectors need to agree, i.e., that
\[
    \frac{\Sigma_1 w_1}{\vnorm{\Sigma_1^{1/2}w_1}} = \left(\frac{\kappa_3 (\truncpara_2) }{\kappa_3(\truncpara_1) }  \right)^{1/3} \cdot \frac{\Sigma_2 w_2}{\vnorm{\Sigma_2^{1/2} w_2}}.
\]
By multiplying by $w_2^\top$ from the left, we can get an expression for $w_1^\top \Sigma_1 w_2 = w_2^\top \Sigma_1 w_1$.
By multiplying instead by $w_1^\top$ from the left, we can get an expression for $w_1^\top \Sigma_2 w_2$.
Using again that $w_i^\top \Sigma_i w_i = \vnorm{\Sigma_i^{1/2}w_i}^2$ and combining these we get that
\[
    \left(\frac{\kappa_3 (\truncpara_1) }{\kappa_3(\truncpara_2)}\right)^{1/3}w_2^\top\Sigma_1w_1 = \vnorm{\Sigma_1^{1/2}w_1} \cdot \vnorm{\Sigma_1^{1/2}w_2} = \left(\frac{\kappa_3 (\truncpara_2) }{\kappa_3(\truncpara_1)}\right)^{1/3}w_1^\top\Sigma_2w_2.
\]
We can now use this equality in \eqref{eqn:secmomentiden} to get
\[
    \frac{\kappa_2(\truncpara_1)}{\kappa_3(\truncpara_1)^{2/3}} = \frac{\kappa_2(\truncpara_2)}{\kappa_3(\truncpara_2)^{2/3}} \quad \Longleftrightarrow \quad \frac{\kappa_3(\truncpara_1)}{\kappa_2(\truncpara_1)^{3/2}} = \frac{\kappa_3(\truncpara_2)}{\kappa_2(\truncpara_2)^{3/2}}.
\]
We now observe that the function $\kappa_3(\truncpara)/(\kappa_2(\truncpara))^{3/2}$ is the \emph{skewness} of the one-dimensional standard Gaussian truncated to $(-\infty, \truncpara]$. We denote this function by $\sk(\truncpara)$. Thus, we have that if the first three moments of $\truncatedGauss{\mu_1}{\Sigma_1}{\{w_1^\top x \leq \tau_1\}}$ and $\truncatedGauss{\mu_2}{\Sigma_2}{\{w_2^\top x \leq \tau_2\}}$ agree, then we need to have
\[
    \sk(\truncpara_1) = \sk(\truncpara_2).
\]
The key argument in this proof now is that for our specific case this function $\sk$ is \emph{monotonic} (cf. \Cref{lemma:monotonicity-of-skewness}) and thus the above implies that also $\truncpara_1 = \truncpara_2$.

\paragraph{Estimating the \truncparaword{} algorithmically.}
We now want to show how we can use the proof above to determine the parameters $\mu$ and $\Sigma$ given samples from the distribution $\truncatedGauss{\mu}{\Sigma}{\{w^\top x \leq \tau\}}$.
We showed above that the first three moments uniquely determine $\truncpara = \frac{\tau - w^\top \mu}{\vnorm{\Sigma^{1/2}w}}$.
The above also lays down a clear path: first, estimate $\truncpara$ and then determine the parameters $\mu$ and~$\Sigma$.
To do this, we first compute $\sk(\truncpara)$ based on the moments.
Then, using the fact that $\sk$ is monotonic, we can invert this \emph{one-dimensional} function and recover the \truncparaword{}.
This then allows us to use \eqref{EQ:TO:firstmoment} and \eqref{EQ:TO:secondmoment} to recover the parameters.

It thus remains to show how to determine $\truncpara$ algorithmically based on the first three moments.
Define the vector
\[
    u^* = \kappa_3(\truncpara)^{1/3} \cdot \frac{\Sigma w}{\vnorm{\Sigma^{1/2}w}}.
\]
The third central moment equals ${u^*}^{\otimes 3}$ and hence, we can use it to obtain a unit vector $\tilde{u}$ in the direction of $u^*$, that is, $\tilde{u} = (\Sigma w)/\vnorm{\Sigma w}$.
By computing the inner product of the third central moment with the tensor $\tilde{u}^{\otimes 3}$, we can even get the norm of $u^*$ since
\begin{equation}\label{EQ:TO:expressionkappa3}
    \Iprod{\tilde{u}^{\otimes 3}, \bbE[(X-\bbE[X])^{\otimes 3}]} = \Iprod{\tilde{u},u^*}^3 = \kappa_3(\truncpara) \cdot \frac{\|\Sigma w\|^3}{\|\Sigma^{1/2} w\|^3}.
\end{equation}
We now want to use the second central moment to compute a similar expression involving the same norms but $\kappa_2(\truncpara)$ instead of $\kappa_3(\truncpara)$.
Combining this with the above then allows us to compute $\sk(\truncpara)$.
It turns out that we can do so by computing $\tilde{u}^\top M_2^{-1} \tilde{u}$, where $M_2$ is the second central moment.
Since $M_2$ differs from $\Sigma$ only by a rank-$1$ term, we can use the Sherman-Morrison formula to compute its inverse
\begin{align*}
    M_2^{-1} &= \Sigma^{-1} + (1-\kappa_2(\truncpara)) \cdot \frac{\Sigma^{-1} \frac{\Sigma w}{\vnorm{\Sigma^{1/2} w}} \cdot \frac{w^\top \Sigma}{\vnorm{\Sigma^{1/2} w}} \Sigma^{-1}}{1 - (1-\kappa_2(\truncpara)) \cdot \frac{w^\top \Sigma}{\vnorm{\Sigma^{1/2} w}} \Sigma^{-1} \frac{\Sigma w}{\vnorm{\Sigma^{1/2} w}}}\\
    &= \Sigma^{-1} + (1-\kappa_2(\truncpara)) \cdot \frac{w w^\top}{\vnorm{\Sigma^{1/2} w}^2 \cdot \kappa_2(\truncpara)}.
\end{align*}
Now, left-multiplying by $\tilde{u}^\top$ and right-multiplying by $\tilde{u}$, we get that
\begin{align}
    \tilde{u}^\top M_2^{-1} \tilde{u} &= \frac{w^\top \Sigma \Sigma^{-1} \Sigma w}{\vnorm{\Sigma w}^2} + (1-\kappa_2(\truncpara)) \cdot \frac{w^\top \Sigma w w^\top \Sigma w}{\vnorm{\Sigma w}^2 \cdot \vnorm{\Sigma^{1/2}w}^2 \cdot \kappa_2(\truncpara)}\nonumber\\
    &= \frac{\vnorm{\Sigma^{1/2}w}^2}{\vnorm{\Sigma w}^2} + (1-\kappa_2(\truncpara)) \cdot \frac{\vnorm{\Sigma^{1/2}w}^2}{\vnorm{\Sigma w}^2 \cdot \kappa_2(\truncpara)}\nonumber\\
    &= \frac{1}{\kappa_2(\truncpara)} \cdot \frac{\vnorm{\Sigma^{1/2}w}^2}{\vnorm{\Sigma w}^2} \label{EQ:TO:expressionkappa2}.
\end{align}
Now, combining \eqref{EQ:TO:expressionkappa2} with \eqref{EQ:TO:expressionkappa3}, we get that
\[
    \Iprod{\tilde{u}^{\otimes 3}, \bbE[(X-\bbE[X])^{\otimes 3}]} \cdot \Paren{\tilde{u}^\top M_2^{-1} \tilde{u}}^{3/2} = \frac{\kappa_3(\truncpara)}{\kappa_2(\truncpara)^{3/2}} = \sk(\truncpara).
\]
We remark that we can compute both quantities on the left-hand side of the above equation using~$\tilde{u}$ (that we computed using the third moment) and the moments.
Hence, we can compute $\sk(\gamma)$.
We now crucially use the fact that $\sk$ is monotonic, which allows us to recover $\truncpara$.
As already discussed before, this allows us to recover the parameters: we can use the third moment to recover the vector $(\Sigma w)/\vnorm{\Sigma^{1/2} w}$ (which we do implicitly in our algorithm) and then use the first and second moment to recover $\mu$ and $\Sigma$.
Our algorithm is summarized in \Cref{algo:intro:main}.
\begin{algorithm}[!ht]
\caption{Parameter estimation of a halfspace-truncated Gaussian based on first three moments}
\label{algo:intro:main}
\KwIn{First non-central moment $\mu_\mathrm{t}$, second central moment $M_2$ and third central moment $M_3$ of the distribution $\truncatedGauss{\mu}{\Sigma}{\{w^\top x \leq \tau\}}$}
\KwOut{Mean $\mu$ and covariance $\Sigma$.}

1. Compute $\tilde{u}$ from the third central moment $M_3$.

2. Compute the relative truncation parameter: $\gamma=\sk^{-1}\Paren{\Paren{\tilde{u}^\top M_2^{-1}\tilde{u}}^{3/2}\cdot \Iprod{\tilde{u}^{\otimes 3},M_3}}$.

\BlankLine

3. Compute the mean of the Gaussian: $\mu=\mu_\mathrm{t}-\frac{\kappa_1(\gamma)}{\sqrt{\kappa_2(\gamma)}}\Paren{\tilde{u}^\top M_2^{-1}\tilde{u}}^{-1/2}\tilde{u}$.

\BlankLine
4. Compute the covariance of the Gaussian: $\Sigma=M_2+\Paren{\frac{1}{\kappa_2(\gamma)}-1}\cdot\Paren{\tilde{u}^\top M_2^{-1}\tilde{u}}^{-1}\tilde{u}\tilde{u}^\top$.
\end{algorithm}

We emphasize here that our particular choice of expressing both the second and third moment in \eqref{EQ:TO:secondmoment} and \eqref{EQ:TO:thirdmoment} in terms of the \emph{same} \truncparaword{}~$\truncpara$ is precisely what enables \emph{direct} parameter recovery without having to learn the truncation set --- deviating from existing approaches in the truncated statistics literature ~\cite{DGTZ18:efficientforknownset,lee2024}.
An important upshot is that this also enables us to design \emph{simple} and \emph{fast} algorithms, making learning under truncation potentially practical.

\paragraph{Sample complexity.}
In \Cref{sec:fullproof}, we will show how to implement the above strategy when only having access to the sample moments.
We now give a brief overview how the errors in estimating the sample moments influence the sample complexity of our algorithm.

The main bottlenecks for the sample complexity are the approximation $\hat{u}$ to $\tilde{u}$ and the approximation $\widehat{M}_2$ to $M_2$.
We need to ensure that the expressions in \eqref{EQ:TO:expressionkappa3} and \eqref{EQ:TO:expressionkappa2} are close to their true values.
If we have sufficiently good approximations to these, then with $O_\alpha(d^2/\varepsilon^2)$ extra samples we can get good approximations to the mean and covariance of the truncated distribution and use the strategy described above to get estimators for $\mu$ and $\Sigma$ with error $\varepsilon$.
We show in \Cref{sec:fullproof} that it is sufficient to get estimates
\begin{equation}\label{EQ:TO:necessaryestimatesforsamplecomplexity}
    \Norm{\hat{u}-\frac{u^{\ast}}{\norm{u^{\ast}}}}_2\leq O_\alpha(\varepsilon) \:\:\: \text{and} \:\:\: \norm{M_2^{-1/2}\widehat{M}_2 M_2^{-1/2}-I_d}_F \leq O_\alpha(\varepsilon).
\end{equation}
We note that our algorithm contains a preconditioning step that allows us to reduce our problem to a well-conditioned one.
This allows us to invert the matrix $\widehat{M}_2$ without incurring an error that depends on the condition number of $\Sigma$.

Achieving the estimate of $\widehat{M}_2$ as in \eqref{EQ:TO:necessaryestimatesforsamplecomplexity} is possible with $O_\alpha(d^2/\varepsilon^2)$ samples by standard concentration bounds.
For the estimate of $\tilde{u}$, we could estimate it by estimating the diagonal elements of the third moment tensor.
However, the diagonal elements of the third moment tensor are $(u^*_i)^3$ (and not $u^*_i$).
Using this strategy would lead to suboptimal sample complexity in $d$ since we would need a very good approximation to each entry $(u^*_i)^3$.
This is because $(u^*_i)^3$ can be small and thus taking the cube-root can amplify the error.

To improve the sample complexity, we instead compute a random contraction of the third moment tensor.\footnote{Random contractions are typically used in decomposition of higher-order tensors and have seen applications in learning Gaussian Mixture Models, Independent Component Analysis and Dictionary Learning. The specific type of random contraction we utilize is attributed to Jennrich~\cite{harshman1970foundations, leurgans1993decomposition}.}
Then, on the resulting matrix we can compute the largest eigenvalue (in absolute value) and compute the corresponding vector. This vector is close to $\tilde{u}$ since in expectation (over the samples) this random contraction will yield a rank one matrix proportional to $\tilde{u} \tilde{u}^\top$.
We show in \Cref{lem:jenrich,lem:tensor-spectral-concentration} that this strategy allows us to get an estimate $\hat{u}$ satisfying the guarantee in \eqref{EQ:TO:necessaryestimatesforsamplecomplexity} using $O_\alpha(d^{2}/\varepsilon^2)$ samples.
This allows us to get optimal sample complexity in $d$.

Putting it together, we show that using $\frac{d^2}{\varepsilon^2} \cdot \poly(1/\kappa_2(\truncpara), 1/\kappa_3(\truncpara))$ samples, we can get estimates $\hat{\mu}$ and $\widehat{\Sigma}$ such that
\[
    \Norm{\Sigma^{-1/2}(\hat{\mu}-\mu)} \leq \varepsilon \quad \text{and} \quad \Norm{I_d - \Sigma^{-1/2}\widehat{\Sigma}\Sigma^{-1/2}}_F \leq \varepsilon.
\]
We show in \Cref{sec:millsratio} that for $\alpha \leq 0.99$ we can bound $1/\kappa_2(\truncpara), 1/\kappa_3(\truncpara) \leq \poly\log(1/\alpha_0)$.
For $\alpha \to 1$, $\kappa_3(\truncpara) \to 0$, which leads to an increase in the sample complexity.
If $\alpha$ is very close to $1$ (meaning $\alpha \geq 1 - \tilde{O}(\varepsilon)$), then the sample mean and sample covariance of the truncated distribution are good enough estimates and we show that our estimators are close to these.
For the in-between case (i.e., $0.99 \leq \alpha < 1 - \tilde{O}(\varepsilon)$), we can bound $\kappa_3(\truncpara) \geq \Omega(\varepsilon)$, which gives the sample complexity $O(d^2/\varepsilon^4)$.
This shows the sample complexity as described in \Cref{THM:INTRO:main}.

The fact that we need more samples for the case $\alpha > 0.99$ seems somewhat artificial.
The problem is that, when $\alpha$ is close to $1$, the norm of the third central moment tensor $M_3$ goes to $0$, which makes it harder to estimate $\tilde{u}$.
On the other hand it seems that the problem of estimating the parameters of a truncated Gaussian should become easier if there is less truncation.
However, it is unclear to us whether and how the sample complexity could be improved with our approach.

\paragraph{Runtime.}
We now briefly discuss the runtime of our algorithm.
As a first step, we need to compute the vector $\hat{u}$.
For this, we need to compute a random contraction of the third central sample moment tensor.
If we first perform the contraction and then average over the samples, we need to add~$N$ rank-$1$ matrices of size $d^2$, which can be done in time $O(T(N,d))$ using fast rectangular matrix multiplication.
Next, we need to compute the spectral decomposition of the resulting matrix, which can be done in time $O(d^3)$.
For computing the estimate $\widehat{M}_2$ we need to compute the sample covariance, which again can be done in time $O(T(N,d))$.
We then need to invert the matrix $\widehat{M}_2$, which can be done in time $O(d^3)$.
Next, we compute the expressions in \eqref{EQ:TO:expressionkappa3} and \eqref{EQ:TO:expressionkappa2}.
If we again first compute the inner product and then average over the samples (for \eqref{EQ:TO:expressionkappa3}), these can be done in time $O(N \cdot d)$ and $O(d^3)$ respectively.
Since $T(N,d) \geq \Omega(N \cdot d) \geq \Omega(d^3)$, the terms involving $T(N,d)$ dominate and we can thus compute an estimate for $\sk(\truncpara)$ in time $O(T(N,d))$.

Once we know an estimate of $\sk(\truncpara)$, we use the inverse of $\sk$ to compute $\hat{\truncpara}$.
We note that since this is a one-dimensional monotonic function, we can compute an approximate inverse using binary search.
The computational complexity of this step is negligible relative to the primary calculations (see \Cref{rem:approximateinverse}).
The final step requires adding a single term to the first and second moments, which is achievable in time $O(d^2)$.

Hence, the dominating term in the overall runtime is the computation of the sample covariance, which, as argued in \Cref{sec:introduction}, is $O(T(N,d)) \leq O(d^{3.250035}) \cdot \poly(1/\varepsilon, \log(1/\alpha_0))$.

\paragraph{Organization.}
In \Cref{sec:preliminaries}, we cover the necessary preliminaries for our paper and in particular prove that the function $\sk$ is monotonic.
In \Cref{sec:setup}, we compute the moments of the truncated distribution and discuss our preconditioning.
Then, in \Cref{sec:fullproof}, we give a complete proof of our main result.
In \Cref{sec:robustness}, we give our result on robustly estimating the parameters of a Gaussian under halfspace truncation.
In \Cref{sec:extensions}, we prove the extension of our result to an intersection of two orthogonal halfspaces.
In \Cref{APP:facts,APP:computations}, we prove some facts and computations needed for our results.
Finally, in \Cref{APP:howtogetw}, we show how one could get $w$ based on the moments.
This is not needed in our proof, but we include it for completeness since this is necessary for the proof of \cite{lee2024}.

%% file: content/preliminaries.tex
\section{Preliminaries}\label{sec:preliminaries}
\subsection{Notation}
We use the following convention:   
We denote the unit sphere in $\R^d$ by $\cS^{d-1}$. For $n \in \N$, $[n]$ denotes the set~$\{1, 2, \dots, n\}$.
Unless explicitly stated, the base of the logarithm is $e$. For two vectors $a,b\in\R^d$, we denote their inner product by $\iprod{a,b}$ or $a^\top b$. Unless otherwise specified, all vector norms $(\Vert \cdot \Vert$ or $\Vert \cdot \Vert_2)$ are the Euclidean norm.
For matrices, $\Vert\cdot\Vert$ denotes the spectral norm and $\Vert\cdot\Vert_F$ denotes the Frobenius norm, $\Vert\cdot\Vert_{\ast}$ denotes the nuclear norm. We use the notation $O(\cdot), \Theta(\cdot), \Omega(\cdot), \lesssim, \gtrsim $ to hide absolute constants. We use $\widetilde{O}(\cdot)$ to hide logarithmic factors.
We use $\mathbbb{1}[\cdot]$ for the indicator variable.
For positive semi-definite (PSD) matrices $A, B \in \R^{d \times d}$, we use $A \preccurlyeq B$ to mean that $A$ is at most $B$ in the \Lowner order (meaning $B-A$ is PSD).
We use $\lambda_{\min}(A)$ to denote the minimum eigenvalue of a square matrix $A$.
We denote the total variation distance by $\TVdist$. We denote the density function of the standard Gaussian by $\phi(\cdot)$ and its cumulative density function by $\Phi(\cdot)$. 

\paragraph{Tensor norm.} For two order-3 tensors $T, S\in \Rsymb^{d\times d\times d}$, we define $\iprod{T,S}=\sum_{ijk}T_{i,j,k} S_{i,j,k}$. Furthermore, we define the injective tensor norm of $T$ as $\norm{T}=\sup_{\norm{x}=\norm{y}=\norm{z}=1} \Set{\iprod{x\otimes y\otimes z, T}}$.
Note that if $T$ is a symmetric tensor, as a consequence of Banach's Theorem (see e.g. \cite[Theorem 1]{friedland2012} or \cite[Corollary 4.2]{ChenHeLiZhang:symmetrictensors}), it holds that \[\Norm{T}=\sup_{\norm{x}=\norm{y}=\norm{z}=1} \Set{\iprod{x\otimes y\otimes z, T}}=\underset{\norm{x}=1}{\sup}\Set{\Iprod{x^{\otimes 3},T}}.\]

\paragraph{Truncated Gaussian.} For a measurable set $S\subseteq\R^d$, let $\alpha>0$ to be the probability mass of the set $S$ under the true Gaussian distribution $\cN(\mu,\Sigma)$, i.e., $\alpha=\Esymb_{x\sim \normal{\mu, \Sigma}}\mathbbb{1}[x\in S]$.
We say that a random vector $\x$ follows a Gaussian $\cN(\mu,\Sigma)$ truncated to $S$, if it follows the distribution $\cN(\mu,\Sigma)$ conditioned on the event $x\in S$. For a halfspace $H=\Set{x:\iprod{w,x}\leq \tau}$, we denote a Gaussian $\cN(\mu,\Sigma)$ truncated to a halfspace $H$ by \[
\truncatedGauss{\mu}{\Sigma}{\iprod{w,x}\leq\tau}
\]
\subsection{Inverse Mill's ratio}\label{sec:millsratio}
For $\gamma \in \R$, consider the standard Gaussian random variable $\normal{0,1}$ conditioned on $(-\infty,\gamma]$, denoted by $Z_\gamma$.
In this section, we introduce the first three cumulants of $Z_\gamma$, show how to describe them in terms of the density function and the cumulative density function of the standard Gaussian and give several bounds that will be needed later in our proofs.

For a general random variable $X$ on $\R$, the cumulants are defined as follows.
\begin{definition}\label{DEF:cumulant}
    Let $X$ be a random variable on $\R$.
    We define its cumulant generating function by $K(t)=\log\E e^{t X}$.
    If $K(t)$ exists and is analytic, we define the $n$-th cumulant of $X$, denoted by $\kappa_n(X)$ as the $n$-th order derivative of $K(t)$ evaluated at $t=0$, i.e., $\kappa_n(X) = K^{(n)}(0)$.
\end{definition}
In particular, one can compute that the first cumulant is the mean of $X$, the second cumulant is the variance of $X$ and the third cumulant is the third central moment.
Thus, the first three cumulants of $Z_\gamma$ (which we denote by $\kappa_1(\gamma)$, $\kappa_2(\gamma)$, and $\kappa_3(\gamma)$ respectively) are as follows:
\begin{equation}\label{EQ:ThreeCumulantsfor1DGauss}
    \begin{aligned}
        \kappa_1(\gamma) &=\E Z_\gamma,\\
        \kappa_2(\gamma) &=\E \Paren{Z_\gamma-\E Z_\gamma}^2,\\
        \kappa_3(\gamma) &=\E \Paren{Z_\gamma-\E Z_\gamma}^3.
    \end{aligned}
\end{equation}
They can be characterized in terms of $\phi(\gamma)$ and $\Phi(\gamma)$ as follows.
\newcommand{\factfirstthreecumulantstext}{
    Let $\gamma \in \R$. We have that
    \begin{align*}
    \kappa_1(\gamma) &= \frac{1}{\Phi(\gamma)} \cdot \int_{-\infty}^{\gamma} t \phi(t) dt = - \frac{\phi(\gamma)}{\Phi(\gamma)},\\
        \kappa_2(\gamma)&=1-\gamma\frac{\phi(\gamma)}{\Phi(\gamma)}-\Paren{\frac{\phi(\gamma)}{\Phi(\gamma)}}^2,\\
        \kappa_3(\gamma)&=\frac{\phi(\gamma)}{\Phi(\gamma)} \cdot \left( 1 - \gamma^2 - 3 \gamma \cdot \frac{\phi(\gamma)}{\Phi(\gamma)} - 2 \cdot \left(\frac{\phi(\gamma)}{\Phi(\gamma)}\right)^2 \right).
    \end{align*}
}
\begin{fact}\label{fact:first-3-cumulants}\factfirstthreecumulantstext
\end{fact}
The proof of \Cref{fact:first-3-cumulants} follows by direct computation. We include a proof in~\Cref{APP:facts:cumulantsof1dtruncated}.

Moreover, the \textit{skewness} of a random variable $X$, denoted by $\text{skew}(X)$ on $\R$, is defined as the third moment of the normalized random variable $\frac{X-\mu}{\sigma}$, where $\mu$ denotes the mean of $X$ and $\sigma$ the standard deviation.
This can also be written in terms of the cumulants.
Namely, we have that
\[
    \text{skew}(X) = \frac{\kappa_3(X)}{\kappa_2(X)^{3/2}}.
\]
We denote the skewness of $Z_\gamma$ by $\text{skew}(\gamma)$.
Using \Cref{fact:first-3-cumulants}, we immediately get that
\begin{equation}\label{EQ:defofskew}
        \text{skew}(\gamma) =\frac{\phi(\gamma)}{\Phi(\gamma)} \left( 1 - \gamma^2 - 3 \gamma \frac{\phi(\gamma)}{\Phi(\gamma)} - 2 \left(\frac{\phi(\gamma)}{\Phi(\gamma)}\right)^2 \right) \left(1 - \gamma \frac{\phi(\gamma)}{\Phi(\gamma)}  - \left(\frac{\phi(\gamma)}{\Phi(\gamma)}\right)^2 \right)^{-\frac{3}{2}}.
\end{equation}
We now discuss some properties of $\kappa_1,\kappa_2,\kappa_3$ and give some estimates for them that we will need for our analysis later.
Observe that $\kappa_1(x),\kappa_2(x),\kappa_3(x)$ can be viewed as polynomials in $x$ and~$\frac{\phi(x)}{\Phi(x)}$.
We have that the ratio $\frac{\phi(x)}{\Phi(x)}\approx\phi(x)$ as $x\rightarrow +\infty$ and thus the non-trivial case is when~$x\rightarrow -\infty$.
To address this latter case, as in \cite{lee2024}, we introduce the following function, commonly known as the \textit{inverse Mills ratio} or \textit{hazard function}:
\[
    h(x)=\frac{\phi(x)}{1-\Phi(x)}.
\]
Observe that bounds on $h(x)$ (which typically only work for $x\geq 0$) can be related to $\frac{\phi(-x)}{\Phi(-x)}$ as follows
\[
    \frac{\phi(-x)}{\Phi(-x)} = \frac{\phi(x)}{1-\Phi(x)} = h(x).
\]
We will need the following three facts.
\newcommand{\factrationalappxofimrtext}{
    Define the following functions for $k \in \N_0$ and $x > 0$
    \[
        r_{k+1}(x) \coloneqq {x+\frac{1}{x+\frac{2}{x+\frac{3}{\underset{x+\frac{k}{x}}{\ddots}}}}}.
    \]
    Then, for any $m \in \N$ and $x \geq 0$ we have
    \[
        r_{2m-1}(x) < h(x) < r_{2m}(x).
    \]
    In particular, we have
    \[
        x < h(x) < x+\frac{1}{x}.
    \]
}
\begin{fact}[\cite{approximatingmillsratio}]\label{FACT:rational-appx-of-imr}\factrationalappxofimrtext    
\end{fact}
We have that $r_{2m}(x) \searrow h(x)$ and $r_{2m+1} \nearrow h(x)$ as $m \to \infty$ for all $x > 0$.
In particular, the above estimates are accurate for large $x>0$.
Intuitively, this is because the higher-order term $x+\frac{k}{x}\approx x$ for large $x$.
However, for small fixed $k$, the approximation quality of $r_k(x)$ can be extremely poor near $x=0$.
In particular, despite $h(0)=\sqrt{\frac{2}{\pi}}$, the upper bound $r_{2m}(x)\rightarrow\infty$ as $x\rightarrow 0$ for \textit{all fixed} integers $m>0$. 
For $x$ close to $0$, we can instead use the following two estimates of the inverse Mills ratio, where the second one is exact at $x=0$.
\begin{fact}[\cite{Sampford1953SomeIO,Birnbaum1942}]\label{FACT:ir-appx-of-imr}
    For $x > 0$, we have 
    \[
        \frac{1}{4}\paren{\sqrt{x^2+8}+3x}< h(x)< \frac{1}{2}\paren{\sqrt{x^2+4}+x}.
    \]
\end{fact}
\begin{fact}[\cite{boyd1959inequalities}]\label{FACT:ir-appx-of-imr2}
    For $x > 0$, we have
    \[
        \frac{\sqrt{(\pi-2)^2x^2+2\pi}+2x}{\pi}<h(x)<\frac{\sqrt{x^2+2\pi}+(\pi-1)x}{\pi}.
    \]
\end{fact}
Based on these facts, we prove the following claims that are necessary for our analysis.
We provide proofs of these claims in \Cref{APP:facts:cumulantsof1dtruncated}.

\newcommand{\claimonelipschitzkonetext}{
    The function $\kappa_1$ as defined in \eqref{EQ:ThreeCumulantsfor1DGauss} is 1-Lipschitz on $\R$.
}
\begin{claim}\label{claim:1-lipschitz-k1}\claimonelipschitzkonetext     
\end{claim}
\newcommand{\claimboundonkappatwotext}{
    For every $t \in \R$, we have that $0<\kappa_2(t)\leq 1$.
}
\begin{claim}\label{claim:bound-on-kappa2}\claimboundonkappatwotext
\end{claim}
\newcommand{\claimmonotonicityofkappatwotext}{
    For $\kappa_2$ and $\kappa_3$ as in \eqref{EQ:ThreeCumulantsfor1DGauss}, we have $\ddt\kappa_2(t)=-\kappa_3(t)$. Moreover, $\kappa_3(t) < 0$ for every $t \in \R$ and thus $\kappa_2(t)$ is monotonically increasing.
}
\begin{claim}\label{claim:monotonicity-of-kappa2}\claimmonotonicityofkappatwotext
\end{claim}
\newcommand{\claimkappaonekappatwoboundedtext}{
    For every $t\in\R$, we have $\abs{\kappa_1(t)}\cdot\sqrt{\kappa_2(t)}\leq O(1)$.
}
\begin{claim}\label{claim:kappa1kappa2bounded}\claimkappaonekappatwoboundedtext
\end{claim}
\newcommand{\claimkappaonekappathreeboundedtext}{
    For every $t\in\R$, we have $\abs{\ddt\kappa_2(t)}=\abs{\kappa_3(t)}\leq O(1)$ and ${0<\kappa_1(t)\cdot\kappa_3(t)\leq O(1)}$.
}
\begin{claim}\label{claim:kappa1kappa3bounded}\claimkappaonekappathreeboundedtext
\end{claim}
\newcommand{\claimboundkappaonebykappatwotext}{
    For $t \geq 1$, we have $\abs{\kappa_1(t)} \leq 1-\kappa_2(t)$.
}
\begin{claim}\label{claim:boundkappa1bykappa2}\claimboundkappaonebykappatwotext
\end{claim}
\newcommand{\claimkappatwolowertext}{
    For $t\leq -3$, we have $\kappa_2(t)\geq \frac{1}{4t^2}$.
}
\begin{claim}\label{claim:kappa2-lower}\claimkappatwolowertext    
\end{claim}
\newcommand{\claimkappathreelower}{
    For $t\leq -2$, we have $\Abs{\kappa_3(t)}\geq\frac{1}{50\cdot |t|^3}$.
}
\begin{claim}\label{claim:kappa3-lower}\claimkappathreelower
\end{claim}
\newcommand{\claimoneminusktwo}{
    For $t\geq 1$, we have $1-\kappa_2(t)\leq 16 \cdot (1-\Phi(t))\cdot\log 1/(1-\Phi(t))$.
}
\begin{claim}\label{claim:one-minus-k2}\claimoneminusktwo
\end{claim}
\newcommand{\claimlogalphaandgammatext}{
    For $t\geq 2$, we have $\frac{t^2}{2}\leq \log \Paren{1/(1-\Phi(t))}\leq 2t^2$.
}
\begin{claim}\label{claim:log-alpha-and-gamma}\claimlogalphaandgammatext    
\end{claim}
\newcommand{\claimskewasymptotictext}{
    There are absolute constants $m,M>0$ such that
    \[
        m\leq\frac{\Abs{\sk(t)}}{(1-\Phi(t))\log^{3/2}\Paren{1/(1-\Phi(t))}}\leq M
    \]
    for all $t\geq 0$.
}
\begin{claim}\label{claim:skew-asymptotic}\claimskewasymptotictext
\end{claim}

The following lemma is central for our proof.
\begin{lemma}\label{lemma:monotonicity-of-skewness}
    The function $\sk$ as defined in \eqref{EQ:defofskew} is monotonically increasing on $\R$.
\end{lemma}
    
\begin{proof}
    From \Cref{claim:monotonicity-of-kappa2} we know that $\ddt\kappa_2(t)=-\kappa_3(t)$.
    It follows that
    \begin{align*}
        \ddt\sk(t) &=\frac{\Paren{\ddt\kappa_3(t)}\kappa_2(t)^{3/2}-\frac{3}{2}\sqrt{\kappa_2(t)}\Paren{\ddt\kappa_2(t)}\cdot\kappa_3(t)}{\kappa_2(t)^3}\\
        &=\frac{2\Paren{\ddt\kappa_3(t)}\kappa_2(t)+3\kappa_3(t)^2}{2\kappa_2(t)^{5/2}}
    \end{align*}
    Since $\kappa_2(t) >0$ for every $t \in \R$ by \Cref{claim:bound-on-kappa2}, it suffices to show that the numerator is positive.
    Defining $g(t)=\frac{\Phi(t)}{\phi(t)}\cdot \Paren{2\Paren{\ddt\kappa_3(t)}\kappa_2(t)+3\kappa_3(t)^2}$, since $\frac{\Phi(t)}{\phi(t)}>0$, it is enough to show $g(t) > 0$ for all $t \in \R$.
    After some calculation we obtain 
    \begin{equation}\label{EQ:formulaforg}
        g(t)=2t^3-6t+(t^4+14t^2-5)\frac{\phi(t)}{\Phi(t)}+(2t^3+20t)\Paren{\frac{\phi(t)}{\Phi(t)}}^2+(t^2+8)\Paren{\frac{\phi(t)}{\Phi(t)}}^3.
    \end{equation}
    In order to show that $g(t) > 0 $ for all $t \in \R$, we consider the following function
    \[
       p_t(x)=2t^3-6t+(t^4+14t^2-5)x+(2t^3+20t)x^2+(t^2+8)x^3.
    \]
    Note that $g(t) = p_t\Paren{\frac{\phi(t)}{\Phi(t)}}$.
    Our argument is now structured as follows:
    For fixed $t$, we will define values $x_\ell(t) \leq \frac{\phi(t)}{\Phi(t)} \leq x_u(t)$.
    We then show that $p_t$ is monotonically increasing on $[x_\ell(t), x_u(t)]$. Then we lower bound 
    \[
        g(t) = p_t\Paren{\frac{\phi(t)}{\Phi(t)}} \geq p_t(x_\ell(t))
    \]
    and argue that $p_t(x_\ell(t)) > 0$.
    In order to show that $p_t$ is monotonically increasing, we compute the derivative of $p_t$ with respect to $x$:
    \[
        \Paren{\frac{\partial}{\partial x}p_t}(x)=t^4+14t^2-5+4(t^3+10t)\cdot x+3(t^2+8)\cdot x^2
    \]
    and show that this is positive for $x \in [x_\ell(t),x_u(t)]$.
    
    We distinguish the following cases: (i) $0 < t$; (ii) $-\frac{4}{10} < t \leq 0$; and (iii)~$t < -\frac{4}{10}$.
    For each of the cases we give values for $x_\ell(t)$ and $x_u(t)$ based on approximations for $\frac{\phi(x)}{\Phi(x)}$ and show that $g(t) > 0$ using the strategy described above.
    We will defer some of the computations to \Cref{APP:computations}, where we give Mathematica code to verify the statements we need.

    \textbf{Case (i): $0 < t$}.
    For $t > 0$, we define $x_\ell(t) = \frac{e^{-t^2/2}}{\sqrt{\pi/2}+t}>0$.
    Note that since for $t > 0$ we have $\Phi(t) \leq \frac{1}{2} +\frac{t}{\sqrt{2\pi}}$, we get that $\frac{\phi(t)}{\Phi(t)}\geq x_\ell(t)$ for every $t>0$.
    Since we have $4(t^3+10t) > 0$ and $3(t^2+8) > 0$ for $t > 0$, we get that for any $x \geq x_\ell(t)$\footnote{Since $x_\ell(t) > 0$, this also implies $x^2 \geq x_\ell(t)^2$.}
    \[
        \Paren{\frac{\partial}{\partial x}p_t}\Paren{x} \geq \Paren{\frac{\partial}{\partial x}p_t}(x_\ell(t)).
    \]
    We show in \Cref{FACT:computationmillsratio1} that 
    \[
        \Paren{\frac{\partial}{\partial x}p_t}(x_\ell(t)) > 0
    \]
    for all $t > 0$.
    Hence, we get that for $t > 0$, $p_t$ is monotonically increasing on $[x_\ell(t), \infty)$.
    Thus, we get that
    \[
        g(t)=p_t\Paren{\frac{\phi(t)}{\Phi(t)}}\geq p_t\Paren{x_\ell(t)}.
    \]
    In \Cref{FACT:computationmillsratio1}, we also show that $p_t\Paren{x_\ell(t)} > 0$ for $t > 0$, which shows $g(t) > 0$ for $t > 0$.

    \textbf{Case (ii): $-\frac{4}{10} < t \leq 0$}.
    In this case, we define
    \[
    x_\ell(t) = \frac{\sqrt{(\pi-2)^2t^2+2\pi}-2t}{\pi} \quad \text{and} \quad x_u(t) = \frac{\sqrt{t^2+2\pi}-(\pi-1)t}{\pi}.
    \]
    For $t\leq 0$, we have $x_\ell(t) > 0$. By \Cref{FACT:ir-appx-of-imr2} we have $x_\ell(t) \leq \frac{\phi(t)}{\Phi(t)} \leq x_u(t)$.
    Since for $t \leq 0$ we have $4(t^3+10t) \leq 0$ and $3(t^2+8) > 0$, we get that, for $x_\ell(t) \leq x \leq x_u(t)$ we have\footnote{\label{footnote:xlpositive}We are again using that since $x_\ell(t) > 0$ we also have $x_\ell(t)^2 \leq x^2$.}
    \[
        \Paren{\frac{\partial}{\partial x}p_t}\Paren{x} \geq t^4+14t^2-5+4(t^3+10t)\cdot x_u(t)+3(t^2+8)\cdot x_\ell(t)^2.
    \]
    We show in \Cref{FACT:computationmillsratio2} that
    \[
        t^4+14t^2-5+4(t^3+10t)\cdot x_u(t)+3(t^2+8)\cdot x_\ell(t)^2 > 0
    \]
    for $-\frac{4}{10} < t \leq 0$.
    Thus, we have that $p_t$ is monotonically increasing on $[x_\ell(t),x_u(t)]$ and thus in particular
    \[
        g(t)=p_t\Paren{\frac{\phi(t)}{\Phi(t)}}\geq p_t\Paren{x_\ell(t)}.
    \]
    In \Cref{FACT:computationmillsratio2} we also show that $p_t\Paren{x_\ell(t)} > 0$ for $-\frac{4}{10} < t \leq 0$ and thus $g(t) > 0$ for $-\frac{4}{10} < t \leq 0$.

    \textbf{Case (iii): $t \leq -\frac{4}{10}$}.
    In this case, we define $x_\ell(t) = r_{81}(-t)$ and $x_u = r_{20}(-t)$.
    We verify in \Cref{FACT:computationmillsratio3} that for $t \leq -\frac{4}{10}$ we have $x_\ell(t) > 0$.
    By \Cref{FACT:rational-appx-of-imr}, we have that $x_\ell(t) \leq \frac{\phi(t)}{\Phi(t)} \leq x_u(t)$.
    We again get that\textsuperscript{\ref{footnote:xlpositive}}
    \[
        \Paren{\frac{\partial}{\partial x}p_t}\Paren{\frac{\phi(t)}{\Phi(t)}} \geq t^4+14t^2-5+4(t^3+10t)\cdot x_u(t)+3(t^2+8)\cdot x_\ell(t)^2.
    \]
    In \Cref{FACT:computationmillsratio3} we show that
    \[
        t^4+14t^2-5+4(t^3+10t)\cdot x_u(t)+3(t^2+8)\cdot x_\ell(t)^2 > 0
    \]
    for $t \leq -\frac{4}{10}$.
    We thus again get that $p_t$ is monotonically increasing on $[x_\ell(t),x_u(t)]$ and thus
    \[
        g(t)=p_t\Paren{\frac{\phi(t)}{\Phi(t)}}\geq p_t\Paren{x_\ell(t)}.
    \]
    Finally, we also show in \Cref{FACT:computationmillsratio3} that $p_t\Paren{x_\ell(t)} > 0$ for $t \leq - \frac{4}{10}$ and thus also $g(t) > 0$ for these $t$.
    This completes the proof.
\end{proof}

In \Cref{APP:computations}, we will show the following bound on $\sk$.
\newcommand{\claimskewatmosttwotext}{
    We have $\lim_{t \to -\infty} \sk(t) = -2$ and $\lim_{t \to + \infty} \sk(t) = 0$. Thus, for all $t \in \R$ we have $\sk(t) \in [-2,0]$ and in particular $|\sk(t)| \leq 2$.
}
\begin{claim}\label{claim:skew-atmost2}\claimskewatmosttwotext
\end{claim}

We can even strengthen \Cref{lemma:monotonicity-of-skewness}. In \Cref{APP:computations}, we will show the following two claims that we need for the proofs later. For both of the claims, the idea is that the involved function are point-wise strictly positive. Thus, since the functions are continuous, it is enough to show that the limits as $t \to \pm \infty$ are positive as well to conclude that we can lower bound them by a constant $C > 0$. In \Cref{APP:computations}, we compute these limits and thus prove these claims.
\newcommand{\claimddtskewoverddtkappatwotext}{
    There exists a constant $C > 0$ such that for all $t \in \R$ we have
    \[
        \frac{\ddt\sk(t)}{|\ddt\kappa_2(t)|} \geq C.
    \]
    Moreover, $\Abs{\kappa_2(a)-\kappa_2(b)}\leq \frac{1}{C} \cdot \Abs{\sk(a)-\sk(b)}$ for every $a,b\in\R$.
}
\begin{claim}\label{CLAIM:ddtskewoverddtkappa2}\claimddtskewoverddtkappatwotext
\end{claim}
\newcommand{\claimddtskewoverddtpsitext}{
    Define $\psi(t) \coloneqq \frac{\kappa_1(t)}{\kappa_2(t)^{1/2}}$. Then, there exists a constant $C > 0$ such that for all $t \in \R$ we have
    \[
        \frac{\ddt\sk(t)}{|\ddt\psi(t)|} \geq C \cdot \kappa_2(t)^{2}.
    \]
    Moreover, $\Abs{\psi(a)-\psi(b)}\leq \frac{1}{C} \cdot \Abs{\sk(a)-\sk(b)} \cdot \max\{\kappa_2(a)^{-2}, \kappa_2(b)^{-2}\}$ for every $a,b\in\R$.
}
\begin{claim}\label{CLAIM:ddtskewoverddtpsi}\claimddtskewoverddtpsitext
\end{claim}

Furthermore, we can also show an analogous statement for the relation between $\kappa_1$ and $\kappa_2$, which we will also prove in \Cref{APP:computations}. The proof idea is the same as for the above two claims (note that $\kappa_2$ is also monotonic by \Cref{claim:monotonicity-of-kappa2}).
\newcommand{\claimddtkappatwooverddtkappaonetext}{
    There exists a constant $C > 0$ such that for all $t \in \R$ we have
    \[
        \frac{\ddt\kappa_2(t)}{|\ddt\kappa_1(t)|} \geq C \cdot \kappa_2(t)^{2}.
    \]
    Moreover, $\Abs{\kappa_1(a)-\kappa_1(b)}\leq \frac{1}{C} \cdot \Abs{\kappa_2(a)-\kappa_2(b)} \cdot \max\{\kappa_2(a)^{-2}, \kappa_2(b)^{-2}\}$ for every $a,b\in\R$.
}
\begin{claim}\label{CLAIM:ddtkappa2overddtkappa1}\claimddtkappatwooverddtkappaonetext
\end{claim}

Finally, we need tight estimates of $\kappa_2(\truncpara)$ and $\kappa_3(\truncpara)$ for our proofs later.
In order to relate these bounds to the mass $\alpha$ of the truncation set (the halfspace), we use the relation $\alpha = \Phi(\truncpara)$. We defer its proof to \Cref{lem:first3moments}. 

\begin{claim}\label{claim:var-est}
    Let $\truncpara \in \R$.
    For $\alpha=\Phi(\gamma)$, we have
    \[
        \kappa_2(\gamma)\geq \min\Set{\frac{1}{20},\frac{1}{8\log 1/\alpha}}.
    \]
\end{claim}
\begin{proof}
    From \Cref{claim:monotonicity-of-kappa2} we know that $\kappa_2(\gamma)$ monotonically increases on $\Rsymb$.
    For $\gamma\rightarrow\infty$ the distribution converges to standard Gaussian, and the variance goes to $1$.
    For the case $\gamma\rightarrow-\infty$, we can use the rational approximation of the Mills ratio from \Cref{FACT:rational-appx-of-imr}.
    Concretely, we distinguish the two cases $\gamma > -3$ and $\gamma \leq -3$:

    \textbf{Case (i): $\gamma > -3$}.
    For $\gamma >-3$ we have $\kappa_2(\gamma)\geq\kappa_2(-3)\geq 1/20$ by monotonocity of $\kappa_2$.

    \textbf{Case (ii): $\gamma \leq -3$}.
    For $\gamma\leq -3$, by \Cref{claim:kappa2-lower} we have that $\kappa_2(\gamma)\geq\frac{1}{4\gamma^2}$.
    It remains to get a bound on $\gamma$ in terms of $\alpha$.
    By \Cref{FACT:rational-appx-of-imr}, we know that $\alpha = \Phi(\gamma) \leq \frac{\phi(\gamma)}{-\gamma}$.
    This implies that
    \[
        \frac{1}{\alpha} \geq -\sqrt{2\pi} \gamma \cdot e^{\gamma^2/2} \geq 3\sqrt{2\pi}\cdot e^{\gamma^2/2} \geq e^{\gamma^2/2}
    \]
    and thus we get that $\log 1/\alpha \geq \gamma^2/2$. Combining this with bound on $\kappa_2(\gamma)$ from above, we get that, for $\gamma \leq -3$,
    \[
        \kappa_2(\gamma)\geq\frac{1}{8}\cdot\frac{1}{\log 1/\alpha}.
    \]
    Combining the two cases, we get the lower bound as claimed.
\end{proof}

\begin{claim}\label{claim:skew-est}
    Let $\truncpara \in \R$.
    For $\alpha=\Phi(\gamma)$, we have
    \[
        |\kappa_3(\gamma)|\geq\min\left\{\frac{1}{50\log^{3/2}1/\alpha},\frac{1}{20},\frac{1}{4\sqrt{2}}(1-\alpha)\log^{3/2}\frac{1}{1-\alpha}\right\}.
    \]
\end{claim}
\begin{proof}
    Unlike $\kappa_2(\gamma)$, the function $\kappa_3(\gamma)$ is not monotonic in $\gamma$.
    Moreover, the function approaches~$0$ both for $\gamma \to + \infty$ and $\gamma \to - \infty$.
    Thus, we split our analysis in three parts:
    First, we consider the two cases $\gamma \to \pm \infty$ where we use the rational approximations of the Mills ratio from \Cref{FACT:rational-appx-of-imr}.
    Finally, we bound $\kappa_2(\gamma)$ on a small interval around $0$ by an explicit computation.

    \textbf{Case (i): $\gamma \leq -2$}.
    By \Cref{claim:kappa3-lower}, for $\gamma\leq -2$, we have the lower bound $|\kappa_3(\gamma)|\geq\frac{1}{50\cdot |\gamma|^3}$.
    From the proof of \Cref{claim:var-est} we already know that $\log 1/\alpha\geq\gamma^2/2$ for $\gamma \leq -3 \leq -2$.
    So it holds that $\kappa_3(\gamma)\geq\frac{1}{50\log^{3/2}1/\alpha}$ for $\gamma\leq -2$.

    \textbf{Case (ii): $\gamma \geq 2$}.
    From \Cref{claim:monotonicity-of-kappa2} we know $\kappa_3(\gamma)<0$ and thus we have that
    \[
        |\kappa_3(\gamma)| = h(\gamma)\Paren{\gamma^2-1+3\gamma h(\gamma)+2 h(\gamma)^2}.
    \]
    Noting that $h(\gamma) > 0$ and $\Phi(\gamma) \leq 1$ for all $\gamma$, we have for $\gamma\geq 2>\sqrt{2}$ that
    \[
        |\kappa_3(\gamma)| \geq \phi(\gamma)(\gamma^2-1)\geq\frac{1}{2}\phi(\gamma)\gamma^2.
    \]
    By \Cref{FACT:rational-appx-of-imr}, we have for $\gamma \geq 2 > 0$ that $\gamma \leq \frac{\phi(\gamma)}{1-\Phi(\gamma)}$.
    Combining this with the above bound on $|\kappa_3(\gamma)|$ and using $\alpha = \Phi(\gamma)$ we get
    \[
        |\kappa_3(\gamma)|\geq\frac{1}{2}\gamma^3 (1-\alpha).
    \]
    It again remains to give a bound on $\gamma$ in terms of $\alpha$.
    Note that, for $\gamma \geq 2$, $\alpha=\Phi(\gamma)$, by \Cref{claim:log-alpha-and-gamma} we have
    \[
    \log 1/(1-\alpha)\leq 2\gamma^2
    \]
    Combining this with the inequality above, we get that
    \[
        |\kappa_3(\gamma)| \geq \frac{1}{2} \Paren{\frac{1}{2} \log\Paren{\frac{1}{1-\alpha}}}^{3/2} (1-\alpha) = \frac{1}{4\sqrt{2}}(1-\alpha)\log^{3/2}\frac{1}{1-\alpha}
    \]

    \textbf{Case (iii): $\gamma \in (-2,2)$}.
    For the last case, we verify in \Cref{FACT:computationmillsratio4} in \Cref{APP:computations} that for $\gamma\in (-2,2)$, we have $|\kappa_3(\gamma)| \geq \frac{1}{20} $.
    
    Putting all the three cases together gives the result.
\end{proof}

\subsection{Gaussian concentration and anti-concentration}
We will need the following two results on the concentration and anti-concentration of a Gaussian random variable.
We note that the first fact follows immediately from the boundedness of the Gaussian density.

\begin{fact}[Gaussian anti-concentration]\label{fact:gaussianspread}
For every $\delta > 0$, we have that
    \begin{align*}
        \Pr_{X \sim N(0, \sigma^2)} \left[-\delta \leqslant X \leqslant \delta \right] \leqslant \sqrt{\frac{2}{\pi}}\cdot\frac{\delta}{\sigma}
    \end{align*}
\end{fact}
\begin{fact}[{Gaussian concentration, see e.g. \cite[Theorem 3.25]{vanHandel2016}}]\label{fact:gaussian-concentration}
    Let $\psi:\R^d\rightarrow\R$ be a $\lambda$-Lipschitz function.
    Let $g \sim N(0, I_d)$.
    Then for all $t>0$,
    \[
    \Pr\Brac{\Abs{\psi(g)-\E\psi(g)}\geq t}\leq 2\exp\Paren{-\frac{t^2}{2\lambda^2}}
    \]
\end{fact}

\subsection{Concentration bounds for sub-Gaussian random variables}
In this section, we state two concentration results for sub-Gaussian random variables.
To do so, we define the following (quasi-)norm on a random variable $X$. Let $\psi:\Rsymb^{\geq 0}\rightarrow\Rsymb^{\geq 0}$ be a non-decreasing function that satisfies $\psi(0)=0$ and $\psi(x)\rightarrow\infty$ as $x\rightarrow\infty$.
Then we define the following (quasi)-norm of $X$ as
\[
    \norm{X}_{\psi}=\inf\left\{t>0 \mid \Esymb \psi\left(\frac{|X|}{t}\right)\leq 1\right\}
\]
For a constant $\theta\in (0,1)\cup\{2\}$, let $\psi_{\theta}(x)=\exp(x^{\theta})-1$.
Setting $\psi=\psi_2$ in the above results in the sub-Gaussian norm. 
Recall that a random vector $\x$ is $r$-sub-Gaussian if $\iprod{v,\x}$ is $r$-sub-Gaussian for every unit vector $v$ (meaning that $\|\iprod{v,\x} - \E[\iprod{v,\x}]\|_{\psi_2} \leq r$).\footnote{Often, such a random variable is also called $\sigma^2$-sub-Gaussian, where $\sigma = r$ is an upper bound on the $\psi_2$-norm of the one-dimensional projections.}
The following fact provides a standard concentration bound for the covariance estimation of sub-Gaussian distributions.
Note that the second part follows directly from the first one by bounding $\norm{A}_F \leq \sqrt{d} \cdot \norm{A}$.
\begin{fact}[{\cite[Theorem 4.6.1]{Vershynin_2018}}]\label{fact:random-matrix-concentration}
Let $y_1,\ldots,y_n\in\Rsymb^d$ be i.i.d. samples with covariance $\Sigma\succ 0$ such that $\Sigma^{-1/2}y_i$ is $r$-sub-Gaussian. Let the sample mean be $\bar{y}=\frac{1}{n}\sum_{i=1}^{n}y_i$ and let the sample covariance be ${\widehat{\Sigma}=\frac{1}{n}\sum_{i=1}^{n}(y_i-\bar{y})^{\otimes 2}}$. Then, with probability at least $1-\delta$, we have that
\[
    \norm{\Sigma^{-1/2}\widehat{\Sigma}\Sigma^{-1/2}-I_d}\lesssim r^2 \cdot\Paren{\sqrt{\frac{d+\log 1/\delta}{n}}+\frac{d+\log 1/\delta}{n}}.
\]
Furthermore, there exists some absolute constant $C>0$ such that for $n\geq  C(d+\log 1/\delta)$ we have with probability at least $1-\delta$
\[
    \norm{\Sigma^{-1/2}\widehat{\Sigma}\Sigma^{-1/2}-I_d}_F\lesssim r^2 \cdot\frac{d+\sqrt{d\log 1/\delta}}{\sqrt{n}}.
\]
\end{fact}
In the remainder of this section, we show a concentration result for cubes of sub-Gaussian random variables.
As in \cite{kuchibhotla2022}, we define the function ${\psi^{-1}_{\theta,L}(y)=\sqrt{\log (1+y)}+L\log (1+y)^{1/\theta}}$ for all $y\geq 0$, and let $\psi_{\theta, L}$ be the inverse function of $\psi^{-1}_{\theta,L}$. 
Note that when $\theta\in (0,1)$, $\norm{\cdot}_{\psi_{\theta,L}}$ is not a norm since it only satisfies sub-additivity up to a multiplicative constant depending on $\theta$.
The norm $\norm{\cdot}_{\psi_{\theta, L}}$ can be related to the norm $\norm{\cdot}_{\psi_{\theta}}$ via the following theorem.
\begin{theorem}[Theorem 3.1 in \cite{kuchibhotla2022}]\label{thm:sum-of-rv-of-finite-gbo-norm}
    Fix $\theta\in (0,1)$. Let $k \geq 1$ and consider a sequence of independent zero-mean random variables $\{X_i\}_{i=1}^{k}$ such that ${\Norm{X_i}_{\psi_{\theta}}<\infty}$ for all $i\in [k]$. Then there exist constants $Q_\theta, Q^{\prime}_\theta>0$ that only depend on $\theta$ such that
    \[
    \Norm{\sum_{i=1}^{k}X_i}_{\psi_{\theta, L}}\leq Q_\theta\cdot \sqrt{\sum_{i=1}^{k}\Norm{X_i}_{\psi_{\theta}}^2}
    \]
     where $L=\frac{Q^{\prime}_{\theta}}{\sqrt{k}}$.
\end{theorem}
Finally, we have the following tail bound for a random variable with finite $\Norm{\cdot}_{\psi_{\theta,L}}$.
\begin{fact}[Proposition A.3 in \cite{kuchibhotla2022}]\label{fact:concentration-of-gbo-norm}
    Fix $\theta\in (0,1)$ and $L>0$. Then, for a random variable $X$ with $\norm{X}_{\psi_{\theta,L}}<\infty$, we have that 
    \[
    \Pr\left[|X|\geq\norm{X}_{\psi_{\theta, L}} (\sqrt{t}+L t^{1/\theta})\right]\leq 2\exp(-t)
    \]
    for all $t>0$.
\end{fact}
We will now put the above results together to get the following corollary about the concentration of cubes of sub-Gaussian random variables. 
\begin{corollary}[Concentration inequality for cubes of sub-Gaussians]
\label{corollary:subgaussiancubeconcentration}
Let \(Z_1,\dots,Z_n\) be i.i.d. samples of a real-valued random variable \(Z\) such that \(\|Z\|_{\psi_2}\le r\).
Then there exists an absolute constant \(K>0\) such that, for every \(t>0\),
\[
    \Pr\left(\left| \frac1n\sum_{i=1}^n Z_i^3-\E Z^3 \right|\ge K r^3\left(\sqrt{\frac{t}{n}}+\frac{t^{3/2}}{n}\right)\right) \le 2e^{-t}.
\]
\end{corollary}

\begin{proof}
Since \(\|Z\|_{\psi_2}\le r\), we have from the definition that
\[
    \|Z^3\|_{\psi_{2/3}} \le \|Z\|_{\psi_2}^3 \le r^3.
\]
Also, the \(\psi_{2/3}\)-quasi-norm satisfies the triangle inequality up to an absolute constant, that is,
\[
    \|Z^3-\E Z^3\|_{\psi_{2/3}} \lesssim \|Z^3\|_{\psi_{2/3}} + |\E Z^3|.
\]
We will show that \(|\E Z^3|\) is in fact bounded by \(O\left(\|Z^3\|_{\psi_{2/3}}\right)\). To show this we first note by definition of the $\psi_{2/3}$ norm that 
\[
    \mathbb{E}\left[\exp\left(\frac{|Z^3|}{\|{Z^3}\|_{\psi_{2/3}}} \right)^{2/3} \right] \leq 2.
\]
Using this and applying Markov's inequality we get
\begin{align*}
    \Pr\left(|Z^3|\ge u\,\|Z^3\|_{\psi_{2/3}}\right) &= \Pr\left(\exp\left(\frac{|Z^3|}{\|{Z^3}\|_{\psi_{2/3}}} \right)^{2/3} \geq \exp\left(u^{2/3}\right)\right) \\
    &\leq \mathbb{E}\left[\exp\left(\frac{|Z^3|}{\|{Z^3}\|_{\psi_{2/3}}} \right)^{2/3} \right] \cdot \exp\left(u^{-2/3}\right) \\
    &\leq 2 \exp\left(u^{-2/3}\right)
\end{align*}
Obtaining the above bound will now enable us to apply the integral identity for expectations. We therefore have
\[
    \E |Z^3| = \int_0^\infty \Pr(|Z^3|\ge s)\,\mathrm{d}s \le \|Z^3\|_{\psi_{2/3}} \int_0^\infty 2e^{-u^{2/3}}\,\mathrm{d}u \lesssim \|Z^3\|_{\psi_{2/3}}.
\]
where in the first inequality we used the substitution $s = u \cdot \Vert Z^3 \Vert_{\psi_{2/3}}$.
Hence, we have that
\[
    \|Z^3-\E Z^3\|_{\psi_{2/3}} \lesssim r^3.
\]
Now define
\[
    Y_i := Z_i^3-\E Z_i^3.
\]
Then, the \(Y_1,\dots,Y_n\) are \iid, mean-zero, and satisfy
\[
    \|Y_i\|_{\psi_{2/3}}\lesssim r^3.
\]
Applying the sum-of-sub-Weibulls concentration from \Cref{thm:sum-of-rv-of-finite-gbo-norm} we get
\[
    \left\|\sum_{i=1}^n Y_i\right\|_{\psi_{2/3,L}} \lesssim \sqrt{n}\,r^3,
\]
where $L = C/\sqrt{n}$ for some absolute constant \(C>0\). Using the corresponding tail inequality from \Cref{fact:concentration-of-gbo-norm} we have that
\[
\Pr\left(
\left|\sum_{i=1}^n Y_i\right|
\ge
K \sqrt{n}\,r^3\Paren{\sqrt t + \frac{t^{3/2}}{\sqrt{n}}}
\right)
\le 2e^{-t}.
\]
Finally, dividing by \(n\), we conclude that
\[
\Pr\left(
\left|
\frac1n\sum_{i=1}^n Z_i^3-\E Z^3
\right|
\ge
K r^3\left(\sqrt{\frac{t}{n}}+\frac{t^{3/2}}{n}\right)
\right)
\le 2e^{-t},
\]
which completes the proof.
\end{proof}

%% file: content/preparation.tex
\section{Properities of halfspace truncated Gaussians}\label{sec:setup}
In this section, we prove several statements that we need in the proofs later.
In \Cref{sec:setup:moments}, we compute the first three moments of a truncated Gaussian.
In \Cref{sec:setup:subgaussian}, we show that halfspace truncated Gaussian are sub-Gaussian.
And finally, in \Cref{sec:precondition}, we show how to precondition a truncated Gaussian such that the covariance is well-conditioned.

\subsection{The first three moments}\label{sec:setup:moments}
In order to compute the first three moments, we need the following claim that shows that linear transformations of Gaussians truncated to a halfspace continue to remain Gaussians truncated to a halfspace albeit with different parameters.
We defer the proof of this claim to \Cref{APP:facts:matrixvectorcalc}.
\newcommand{\claimaffinetransformtext}{
    Let $A\in\Rsymb^{d\times d}$ be an invertible matrix and let ${b\in\Rsymb^d}$.
    Consider the random variable ${x\sim\truncatedGauss{\mu}{\Sigma}{\iprod{w,x}\leq\tau}}$.
    Then, the random variable ${y=Ax+b}$ follows the distribution ${\truncatedGauss{A\mu+b}{A\Sigma A^\top}{\Iprod{w',y}\leq\tau'}}$, where 
    \[
        w' = \frac{A^{-\top} w}{\Norm{A^{-\top} w}} \quad \text{and} \quad \tau' = \frac{\tau+\Iprod{A^{-\top} w,b}}{\Norm{A^{-\top} w}}.
    \]
}
\begin{claim}[Linear transformations preserve halfspace truncation]\label{claim:affine-transform}\claimaffinetransformtext
\end{claim}
Using this claim, we can now compute the structure of the first moment as well as the second and third central moments of the Gaussian distribution under arbitrary halfspace truncation.
Each of them can be viewed as the corresponding moment of an untruncated Gaussian and an additive rank-1 term due to truncation.
\begin{lemma}
\label{lem:first3moments}
Let $X\sim\truncatedGauss{\mu}{\Sigma}{\iprod{w,x}\leq\tau}$ for some unit vector $w$. Then, we have
\begin{align*}
\Esymb X&=\mu+\kappa_1(\gamma)\frac{\Sigma w}{\norm{\Sigma^{1/2} w}},
\\
M_2 &\coloneqq \Esymb (X-\Esymb X)^{\otimes 2}=\Sigma-(1-\kappa_2(\gamma))\cdot\frac{\Sigma w}{\norm{\Sigma^{1/2} w}}\tp{\Paren{\frac{\Sigma w}{\norm{\Sigma^{1/2} w}}}},
\\
M_3 &\coloneqq \Esymb (X-\Esymb X)^{\otimes 3}=\kappa_3(\gamma)\cdot\Paren{\frac{\Sigma w}{\norm{\Sigma^{1/2} w}}}^{\otimes 3},
\end{align*}
where $\gamma=\frac{\tau-\iprod{w,\mu}}{\Norm{\Sigma^{1/2}w}}$.
Furthermore, we have that $\alpha=\Phi(\gamma)$.
\end{lemma}
\begin{proof}
Consider a random vector $\y$ that follows a $d$-dimensional Gaussian $\truncatedGauss{0}{I_d}{\iprod{e_1,y}\leq\gamma}$. Since only the first coordinate of $\y$ is truncated (and equal to $Z_\gamma$ as in \Cref{sec:millsratio}), we have that
    \begin{align*}
        \E \y&=\kappa_1(\gamma)e_1,\\
        \E(\y-\E\y)^{\otimes 2}&=I_d-(1-\kappa_2(\gamma))\cdot e_1^{\otimes 2},\\
        \E(\y-\E\y)^{\otimes 3}&=\kappa_3(\gamma)\cdot e_1^{\otimes 3}.
    \end{align*}
    Consider an orthonormal matrix $U$ such that $Ue_1=\frac{\Sigma^{1/2}w}{\norm{\Sigma^{1/2}w}}$ or equivalently $e_1=U^\top\frac{\Sigma^{1/2}w}{\Norm{\Sigma^{1/2}w}}$.
    By \Cref{claim:affine-transform}, we have that that $U\y$ follows a standard Gaussian $\cN(0, I_d)$ truncated to the set
    \[
        \Set{x:\Paren{\frac{\Sigma^{1/2}w}{\Norm{\Sigma^{1/2}w}}}^\top x\leq\gamma}.
    \]
    By scaling $U\y$ by $\Sigma^{1/2}$ and shifting by $\mu$, again by \Cref{claim:affine-transform}, we obtain a Gaussian $\cN(\mu,\Sigma)$ truncated to the halfspace
    \[
        H=\Set{x:\iprod{w,x}\leq\gamma\Norm{\Sigma^{1/2}w}+\iprod{w,\mu}}.
    \]
    Thus, we get a direct relationship between $\y$ and $\x$.
    Namely, if we have ${\tau=\gamma\Norm{\Sigma^{1/2}w}+\iprod{w,\mu}}$ or equivalently $\gamma=\frac{\tau-\iprod{w,\mu}}{\Norm{\Sigma^{1/2}w}}$, then we can get $\x$ by ${\x=\Sigma^{1/2}U\y+\mu}$.
    By construction, this also implies
    \[
        \alpha=\E_{x\sim\cN(\mu,\Sigma)}\mathbbb{1}[x\in H] = \E_{y \sim \cN(0,I_n)}\mathbbb{1}[\{e_1^\top y \leq \gamma\}] = \E_{y_1 \sim \cN(0,1)}\mathbbb{1}[y_1 \leq \gamma] = \Phi(\gamma).
    \]
    Finally, from the relationship $\x=\Sigma^{1/2}U\y+\mu$ we compute the moments as follows.
    First, the mean of $X$ is
    \[
        \E\x=\E\Brac{\mu+\Sigma^{1/2}U\y} = \mu+\Sigma^{1/2}U\cdot\kappa_1(\gamma)e_1 = \mu+\kappa_1(\gamma)\frac{\Sigma w}{\Norm{\Sigma^{1/2} w}}.
    \]
    Next, we can compute the second central moment of $X$ as 
    \begin{align*}
        \E\Paren{\x-\E\x}^{\otimes 2}&=(\Sigma^{1/2}U)^{\otimes 2}\E\Paren{\y-\E\y}^{\otimes 2}\\
        &=(\Sigma^{1/2}U)^{\otimes 2}\Paren{I_d-(1-\kappa_2(\gamma))e_1^{\otimes 2}}\\
        &=\Sigma-(1-\kappa_2(\gamma))\Paren{\Sigma^{1/2}Ue_1}^{\otimes 2}\\
        &=\Sigma-(1-\kappa_2(\gamma))\Paren{\frac{\Sigma w}{\Norm{\Sigma^{1/2}w}}}^{\otimes 2}.
    \end{align*}
    Finally, the third central moment equals
    \[
         \E\Paren{\x-\E\x}^{\otimes 3}=(\Sigma^{1/2}U)^{\otimes 3}\E\Paren{\y-\E\y}^{\otimes 3} =(\Sigma^{1/2}U)^{\otimes 3}\kappa_3(\gamma)e_1^{\otimes 3} =\kappa_3(\gamma)\Paren{\frac{\Sigma w}{\Norm{\Sigma^{1/2}w}}}^{\otimes 3},
    \]
    which completes the proof.
\end{proof}

\subsection{Sub-Gaussianity}\label{sec:setup:subgaussian}
For our algorithm we need high accuracy estimates of the moments of a halfspace truncated Gaussian.
To get those, we first establish the following two claims about sub-Gaussianity of the truncated Gaussian.
\begin{claim}\label{claim:subgaussian:M2invX}
    Suppose $\x\sim\truncatedGauss{\mu}{\Sigma}{\iprod{w,x}\leq\tau}$. 
    Then $\Sigma^{-1/2}\x$ is $1$-sub-Gaussian and $M_2^{-1/2}\x$ is $1/\sqrt{\kappa_2(\gamma)}$-sub-Gaussian.
\end{claim}

\begin{corollary}\label{cor:subgaussian:X}
    Suppose $\x\sim\truncatedGauss{\mu}{\Sigma}{\iprod{w,x}\leq\tau}$.
    Then $\x$ is $\norm{\Sigma}^{1/2}$-sub-Gaussian.
\end{corollary}

For the special case of the standard one-dimensional Gaussian, this was already shown in \cite{barreto2024}. We use this fact to prove the general case.
\begin{fact}[cf. Theorem 2.1 in \cite{barreto2024}]\label{fact:sub-gaussian-of-one-sized-gaussian}
    A random variable $Z_\gamma$ from the standard Gaussian distribution $\cN(0,1)$ truncated to $(-\infty,\gamma]$ is $1$-sub-Gaussian.
\end{fact}

\begin{proof}[Proof of \Cref{claim:subgaussian:M2invX} and \Cref{cor:subgaussian:X}]
    A random vector is $r$-sub-Gaussian if its univariate projection along every direction $v$ is $r$-sub-Gaussian.
    For the first claim it therefore suffices to show that the random variable $\iprod{v,\Sigma^{-1/2}\x}$ is $1$-sub-Gaussian.
    Define $\y \coloneqq U^{\top}\Sigma^{-1/2}\x$.
    We first show that this is $1$-sub-Gaussian.
    Let $v$ be a unit vector. Then, we need to show that $\iprod{v,\y}$ is $1$-sub-Gaussian.
    Consider the canonical basis for $\R^d$ defined using the vectors $\Set{e_j}_{j=1}^{d}$. Then, by expressing $v$ in this canonical basis we have
    \[
\iprod{v,\y}=\sum_{j=1}^{d}v_j\iprod{e_j,\y}=\sum_{j=1}^{d}v_j\y_j.
    \]
    By construction, $\y_1$ is distributed identical to $Z_\tau$. Since one-dimensional truncated Gaussians are $1$-sub-Gaussian due to \Cref{fact:sub-gaussian-of-one-sized-gaussian}, $\y_1$ is $1$-sub-Gaussian. Now, $\y_2,\ldots,\y_d$ are independent one-dimensional standard Gaussians. Therefore they are also $1$-sub-Gaussian. Since the density function is a product distribution over the coordinates, we have independence between the $\{\y_i\}_{i=1}^n$. Using the fact that addition of independent sub-Gaussian random variables is sub-Gaussian, and the fact that $v$ is a unit vector, it follows that $\y$ is $1$-sub-Gaussian.

    Multiplying a random vector by a matrix $A$ scales the sub-Gaussian parameter by $\|A^T\|$ (because $|\iprod{v,A\y}| \leq \|A^\top\| \cdot |\iprod{v', \y}|$ for $v' = A^\top v/\norm{A^\top v}$). Thus, we have that $\Sigma^{-1/2}X$ is also $1$-sub-Gaussian (because $\|U^\top\| = 1$).
    Furthermore, we also get that $X$ is $\|\Sigma^{1/2}\|$-sub-Gaussian.
    Finally, by \Cref{lem:first3moments}, we have that
    \[
        M_2 = \Sigma^{1/2} \cdot \underbrace{\Paren{I_d - (1-\kappa_2(\gamma)) \Paren{\frac{\Sigma^{1/2}w}{\|\Sigma^{1/2}\|}}\Paren{\frac{\Sigma^{1/2}w}{\|\Sigma^{1/2}\|}}^\top}}_{B \coloneqq} \cdot \Sigma^{1/2}.
    \]
    We have that $B^{-1/2}\Sigma^{-1/2}X$ is $\|B^{-1/2}\|$-sub-Gaussian.
    Since we get $B$ by removing a rank-1 component from the identity matrix, the eigenvalues of $B$ are 1 with multiplicity $d-1$ and $\kappa_2(\gamma)$.
    Thus, $\|B^{-1/2}\| = 1/\sqrt{\kappa_2(\gamma)}$.
    Finally, note that $B^{-1/2}\Sigma^{-1/2}$ and $M_2^{-1/2}$ differ by an orthogonal matrix (namely $B^{-1/2}\Sigma^{-1/2}M_2^{1/2}$) and thus also $M_2^{-1/2}X$ is $1/\sqrt{\kappa_2(\gamma)}$-sub-Gaussian, which completes the proof.
\end{proof}

\subsection{Preconditioning}\label{sec:precondition}
For our algorithm, we need to estimate the inverse of $M_2$.
A straight-forward application of \Cref{fact:random-matrix-concentration} to get an estimate $\widehat{M}_2$ for $M_2$ and then taking the inverse of $\widehat{M}_2$ would yield a sample complexity that depends on the condition number of $M_2$.
The reason for this is that from \Cref{fact:random-matrix-concentration} we get a bound on $\norm{M_2^{-1/2}\widehat{M}_2 M_2^{-1/2}-I_d}$ but we will need a bound on $\norm{\widehat{M}_2^{-1}-M_2^{-1}}$.
We can relate these two via the following inequality
\[
    \norm{\widehat{M}_2^{-1}-M_2^{-1}} \leq \norm{\widehat{M}_2^{-1}}\cdot\norm{M_2^{1/2}}\cdot\norm{I_d-M_2^{-1/2}\widehat{M}_2 M_2^{-1/2}}\cdot\norm{M_2^{-1/2}},
\]
but in order to make this norm small, the sample complexity will depend on the condition number of $M_2$, which can be arbitrarily large.
To avoid this additional dependency, we compute a natural preconditioner to ensure that our problem becomes well-conditioned.
The guarantee of this preconditioning is the following, where, as before, $\alpha$ denotes the probability mass of the halfspace under the target Gaussian distribution.
\begin{lemma}
\label{lem:preconditioner}
Without loss of generality, given $n = \Theta(d\cdot\max\{1,\log^4 1/\alpha\})$ samples from $\truncatedGauss{\mu}{\Sigma}{\iprod{w,x}\leq\tau}$, we can assume that $\frac{1}{2}I_d\preceq M_2\preceq 2I_d$  and $\frac{1}{2}I_d\preceq\Sigma\preceq \frac{2}{\kappa_2(\gamma)}I_d$ with high probability.
\end{lemma}

Using \Cref{lem:preconditioner}, we can obtain an inverse for $M_2$ without incurring a condition number factor in the sample complexity.
Throughout the rest of the paper, we will work with the assumption that we have preconditioned the data as described in this lemma.
The proof follows directly from the following two claims.

\begin{claim}\label{cor:rough-covariance-est}
    Given $n = \Theta(d\cdot\max\{1,\log^4 1/\alpha\})$ i.i.d. samples from $\truncatedGauss{\mu}{\Sigma}{\iprod{w,x}\leq\tau}$, with high probability the sample covariance $\widetilde{M}_2=\frac{1}{n}\sum_{i=1}^{n}(x_i-\bar{x})(x_i-\bar{x})^{\top}$ satisfies
    \[
    \Norm{I_d-M_2^{-1/2}\widetilde{M}_2M_2^{-1/2}}\leq \frac{\kappa_2(\gamma)}{2}.
    \]
\end{claim}
\begin{claim}\label{claim:preconditioner}
    Let $\widetilde{M}_2\succ 0$ satisfy $\Norm{I_d-M_2^{-1/2}\widetilde{M}_2M_2^{-1/2}}\leq \frac{\kappa_2(\gamma)}{2}$. Then, without loss of generality, we can assume that $\frac{1}{2}I_d\preceq M_2\preceq 2I_d$  and $\frac{1}{2}I_d\preceq\Sigma\preceq \frac{2}{\kappa_2(\gamma)}I_d$. 
\end{claim}
\newcommand{\factaffineinvariancetext}{
    Let $A,B\in\R^{d\times d}$ be positive definite matrices. Let $L\in\R^{d\times d}$ be an invertible matrix. Then, we have that
    \[
        \Norm{I_d-A^{-\frac{1}{2}}B A^{-\frac{1}{2}}}_F=\Norm{I_d-(LAL^\top)^{-\frac{1}{2}}LBL^\top(LAL^\top)^{-\frac{1}{2}}}_F
    \]
}

\begin{proof}[Proof of \Cref{cor:rough-covariance-est}]
    As long as we have $n = \Omega(d/\beta^2)$ samples, we know from \Cref{fact:random-matrix-concentration} that
    \[
        \norm{I_d-M_2^{-1/2}\widetilde{M}_2M_2^{-1/2}}\lesssim \frac{1}{\kappa_2(\gamma)}\beta,
    \]
    since for every sample $x_i$, $M_2^{-1}x_i$ is $1/\sqrt{\kappa_2(\gamma)}$-sub-Gaussian by \Cref{claim:subgaussian:M2invX}.
    Setting $\beta = \Omega(\kappa_2(\truncpara)^2)$ and using \Cref{claim:var-est} which lower bounds $\kappa_2(\gamma)$ in terms of $\alpha$ as $\kappa_2(\truncpara) \geq \min\Set{1/20, \frac{1}{8\log 1/\alpha}}$, we obtain the result.
\end{proof}

\begin{proof}[Proof of \Cref{claim:preconditioner}]
    The idea of our preconditioning is to left-multiply the samples by $\widetilde{M}_2^{-1/2}$, then apply our algorithm and in the end undo the preconditioning by changing the estimates appropriately.
    We will first argue that the moments after this preconditioning satisfy what we claim and then explain how to undo the preconditioning while keeping the error the same.

    \textbf{Moments are well-conditioned.}
    By \Cref{claim:affine-transform}, the covariance of the untruncated Gaussian after this preconditioning is $\Sigma' = \widetilde{M}_2^{-1/2} \Sigma \widetilde{M}_2^{-1/2}$.
    The second central moment of the truncated Gaussian is $M_2' = \widetilde{M}_2^{-1/2}M_2\widetilde{M}_2^{-1/2}$.
    We want to argue that $\frac{1}{2}I_d\preceq M_2'\preceq 2I_d$ and $\frac{1}{2}I_d\preceq\Sigma'\preceq \frac{2}{\kappa_2(\gamma)}I_d$.
    From our assumption, we get that (since $\kappa_2(\gamma) \leq 1$ by \Cref{claim:bound-on-kappa2})
    \[
        \frac{1}{2}I_d\preceq M_2^{-1/2}\widetilde{M}_2M_2^{-1/2}\preceq\frac{3}{2}I_d.
    \]
    Left multiplying by $\widetilde{M}_2^{-1/2}M_2^{1/2}$ and right multiplying by the transpose of this matrix (which preserves the Loewner order), we get
    \[
        \frac{1}{2} I_d \preceq \frac{2}{3}I_d \preceq M_2' = \widetilde{M}_2^{-1/2}M_2\widetilde{M}_2^{-1/2} \preceq 2 I_d,
    \]
    which proves that $M_2'$ is well-conditioned.
    For $\Sigma'$, recall from \Cref{lem:first3moments} that
    \begin{align*}
        M_2&=\Sigma-(1-\kappa_2(\gamma))\cdot\Sigma^{1/2}\Paren{\frac{\Sigma^{1/2} w}{\Norm{\Sigma^{1/2}w}} \Paren{\frac{\Sigma^{1/2} w}{\Norm{\Sigma^{1/2} w}}}^\top}\Sigma^{1/2}\\
        &= \Sigma^{1/2}\Paren{I_d - (1-\kappa_2(\gamma))\cdot\frac{\Sigma^{1/2} w}{\Norm{\Sigma^{1/2}w}} \Paren{\frac{\Sigma^{1/2} w}{\Norm{\Sigma^{1/2} w}}}^\top}\Sigma^{1/2}.
    \end{align*}
    We have
    \[
        \kappa_2(\gamma) I_d \preceq I_d - (1-\kappa_2(\gamma))\cdot\frac{\Sigma^{1/2} w}{\Norm{\Sigma^{1/2}w}} \Paren{\frac{\Sigma^{1/2} w}{\Norm{\Sigma^{1/2} w}}}^\top \preceq I_d
    \]
    and thus, by left multiplying by $\widetilde{M}_2^{-1/2}\Sigma^{1/2}$ and right multiplying by the transpose of this matrix, we get
    \[
        \kappa_2(\gamma) \Sigma' \preceq M_2' \preceq \Sigma'.
    \]
    Hence, it follows that
    \[
        \frac{1}{2} I_d \preceq M_2' \preceq \Sigma' \preceq \frac{1}{\kappa_2(\gamma)} M_2' \preceq \frac{2}{\kappa_2(\gamma)} I_d.
    \]

    \textbf{Undoing the preconditioning.}
    Assume we have estimates $\hat{\mu}_T$ and $\widehat{\Sigma}_T$ in the transformed space.
    We claim that $\hat{\mu} \coloneqq \widetilde{M}_2^{1/2} \hat{\mu}_T$ and $\widehat{\Sigma} \coloneqq \widetilde{M}_2^{1/2}\widehat{\Sigma}_T\widetilde{M}_2^{1/2}$ have the same error in the original space as $\hat{\mu}_T$ and $\widehat{\Sigma}_T$ have in the transformed space.
    This relies on the fact that the Mahalanobis norm, in which we measure the error, is affine-invariant, as the following claim shows.
    We defer the proof to \Cref{APP:facts:matrixvectorcalc}.
    \begin{fact}\label{fact:affine-invariance}\factaffineinvariancetext
    \end{fact}
    By \Cref{fact:affine-invariance} for $A = \Sigma$, $B = \widetilde{\Sigma}$ and $L = \widetilde{M}^{-1/2}$, we have
    \[
        \Norm{I_d - \Sigma^{-1/2}\widetilde{\Sigma}\Sigma^{-1/2}}_F = \Norm{I_d - \Sigma'^{-1/2}\widetilde{\Sigma}_T\Sigma'^{-1/2}}_F.
    \]
    Furthermore, the mean after preconditioning is $\mu' = \widetilde{M}_2^{-1/2} \mu$ by \Cref{claim:affine-transform} and thus we also have
    \begin{align*}
        &\Norm{\Sigma'^{-1/2}(\hat{\mu}_T-\mu')}^2\\
        &\qquad= \Iprod{\Paren{\widetilde{M}_2^{-1/2} \Sigma \widetilde{M}_2^{-1/2}}^{-1/2}\widetilde{M}_2^{-1/2}(\hat{\mu}-\mu),\Paren{\widetilde{M}_2^{-1/2} \Sigma \widetilde{M}_2^{-1/2}}^{-1/2}\widetilde{M}_2^{-1/2}(\hat{\mu}-\mu)}\\
        &\qquad=\Iprod{\widetilde{M}_2^{-1/2}\Paren{\widetilde{M}_2^{-1/2} \Sigma \widetilde{M}_2^{-1/2}}^{-1}\widetilde{M}_2^{-1/2}(\hat{\mu}-\mu),(\hat{\mu}-\mu)}\\
        &\qquad=\Iprod{\Sigma^{-1}(\hat{\mu}-\mu), \hat{\mu}-\mu} = \Norm{\Sigma^{-1/2}(\hat{\mu}-\mu)}^2.
    \end{align*}
    Thus, both the mean estimator $\hat{\mu}$ and the covariance estimator $\widehat{\Sigma}$ have the same error as $\hat{\mu}_T$ and $\widehat{\Sigma}_T$.
    
    Hence, we can indeed precondition the samples by left-multiplying them by $\widetilde{M}_2^{-1/2}$, then run our algorithm on the well-conditioned instance and output $\hat{\mu}$ and $\widehat{\Sigma}$.
    The error does not change by this and thus we can without loss of generality assume that our instance is well-conditioned.
\end{proof}

%% file: content/algorithms.tex
\section{Proof of the main result}
\label{sec:fullproof}

In this section, we prove our main result, \Cref{THM:INTRO:main}, that we restate below.

\begin{theorem}[Full statement of \Cref{THM:INTRO:main}]\label{thm:main}
    Consider a halfspace truncated Gaussian $\truncatedGauss{\mu}{\Sigma}{\iprod{w,x}\leq\tau}$, where $\Sigma\succ 0$.
    Let $\alpha$ be the mass of the truncation set.
    Then, given 
    \[
        n = \begin{cases} O\Paren{\frac{d^2}{\epsilon^2}}\polylog(\alpha^{-1}) & \text{if $\alpha\leq 0.99$}\\ O\Paren{\frac{d^2}{\epsilon^4}} & \text{if $\alpha>0.99$} \end{cases}
    \]
    samples, \Cref{algo:main} computes estimators $\hat{\mu},\widehat{\Sigma}$ such that with probability at least $0.99$
    \[
        \TVdist\Paren{\cN(\hat{\mu},\widehat{\Sigma}),\cN(\mu,\Sigma)}\leq \epsilon
    \]
    The runtime of the algorithm is $O(T(n,d))$, where $T(n,d)$ is the time needed to multiply a $d \times n$ matrix with its transpose.
\end{theorem}
\begin{algorithm}[!ht]
\caption{Algorithm for estimating the mean and covariance of a halfspace truncated Gaussian}
\label{algo:main}
\KwIn{Dimension $d$, sample size $n$, samples $\{x_i\}_{i=1}^{n}$, target accuracy $\varepsilon$.}
\KwOut{Estimated mean $\hat{\mu} \in \mathbb{R}^d$ and covariance $\widehat{\Sigma} \in \mathbb{R}^{d \times d}$.}

\BlankLine

1. Compute sample mean: $\bar{x} = \frac{1}{n}\sum_{i=1}^{n} x_i$.\;

\BlankLine

2. Compute sample second central moment: $\widehat{M}_2 = \frac{1}{n}\sum_{i=1}^{n}(x_i - \bar{x})^{\otimes 2}$.\;

\BlankLine

3. Draw a standard Gaussian vector: $g \sim \mathcal{N}(0, I_d)$.\;

\BlankLine

4. Compute the random contraction: $R = (I_d\otimes I_d\otimes g^\top) \cdot \widehat{M}_3 = \frac{1}{n}\sum_{i=1}^{n}(x_i-\bar{x})^{\otimes 2}\cdot \iprod{g,x_i-\bar{x}}$.

\BlankLine

5. Compute $\hat{u}$, the eigenvector of $R$ corresponding to the largest eigenvalue in absolute value (with unit norm).

\BlankLine

6. Estimate the relative truncation parameter: $\hat{\gamma}=\sk^{-1}\Paren{-\Paren{\hat{u}^\top\widehat{M}_2^{-1}\hat{u}}^{3/2}\cdot\Abs{\Iprod{\hat{u}^{\otimes 3},\widehat{M}_3}}}$.

\BlankLine

7. Flip the sign of $\hat{u}$ if $\Iprod{\hat{u}^{\otimes 3}, \widehat{M}_3}<0$.

\BlankLine

8. Estimate the mean of the Gaussian: $\hat{\mu}=\bar{x}+\frac{\kappa_1(\hat{\gamma})}{\sqrt{\kappa_2(\hat{\gamma})}}\Paren{\hat{u}^\top\widehat{M}_2^{-1}\hat{u}}^{-1/2}\hat{u} $.

\BlankLine

9. Estimate the covariance of the Gaussian: $\widehat{\Sigma}=\widehat{M}_2+\Paren{\frac{1}{\kappa_2(\hat{\gamma})}-1}\cdot\Paren{\hat{u}^\top\widehat{M}_2^{-1}\hat{u}}^{-1}\hat{u}\hat{u}^\top$.

\BlankLine

10. Return estimates $\left(\hat{\mu}, \widehat{\Sigma}
\right)$.

\end{algorithm}

\begin{remark}
    We do not need to know $\alpha$ to run our algorithm. It is sufficient that we get as input a lower bound $\alpha_0$ on $\alpha$. Then we can replace $\alpha$ by $\alpha_0$ in the sample complexity bound of \Cref{thm:main} (which makes the sample complexity bound only larger) and run the algorithm with $O\Paren{d^2/\epsilon^2 \cdot \polylog(\alpha_0^{-1}) + d^2/\varepsilon^4}$ samples.
    If we additionally know whether $\alpha \leq 0.99$ or not, we can run the algorithm with $O\Paren{d^2/\epsilon^2 \cdot \polylog(\alpha_0^{-1})}$ or $O\Paren{d^2/\varepsilon^4}$ samples respectively.
\end{remark}

Note that we argued in \Cref{sec:techniques} that if we would have access to the exact moments, then the estimators $\hat{\mu}$ and $\widehat{\Sigma}$ are equal to $\mu$ and $\Sigma$ respectively.
In this section, we prove that even if we only have access to the sample moments, the above estimators work.
The proof of this theorem is structured as follows:
\begin{itemize}[noitemsep]
    \item In \Cref{SEC:fullproof:directionofsigmaw}, we prove that $\hat{u}$ (after the potential sign flip) is good estimator for $-\frac{\Sigma w}{\norm{\Sigma w}}$, assuming that $\widehat{M}_3$ is a good estimator for $M_3$.
    \item In \Cref{SEC:fullproof:relativetruncpara}, we prove that $\hat{\gamma}$ is a good estimator for $\gamma$, given that we have a good estimator for $\hat{u}$ and that $\widehat{M}_2$ is a good estimator for $M_2$.
    \item In \Cref{SEC:fullproof:meancovariance}, we prove that $\hat{\mu}$ and $\widehat{\Sigma}$ are good estimators for $\mu$ and $\Sigma$, given that $\widehat{M}_2$, $\hat{u}$ and $\hat{\gamma}$ are good estimators.
    \item Finally, in \Cref{SEC:fullproof:puttingtogether}, we put everything together: We determine the sample complexity such that $\widehat{M}_2$ and $\widehat{M}_3$ are good enough estimators such that the lemmas proved in \Cref{SEC:fullproof:directionofsigmaw,SEC:fullproof:relativetruncpara,SEC:fullproof:meancovariance} imply the guarantees of \Cref{thm:main}.
\end{itemize}

We remark that the assumption $\Sigma\succ 0$ is not necessary. Suppose $\Sigma$ is of rank $0<r<d$, then all samples would live in a subspace of $\R^d$ of dimension $r$. One can identify this subspace with $O(d)$ samples with high probability, and this reduces the original problem to an $r$-dimensional problem via an appropriate linear transform (e.g., an orthogonal projection from $\R^d$ to the subspace).

\subsection{\texorpdfstring{Recovering the direction of $\Sigma w$}{Recovering the direction of Σw}}\label{SEC:fullproof:directionofsigmaw}
In order to find the direction of $\Sigma w$, a natural attempt would be to consider the tensor analogue of low-rank approximations of the sample third moment $\widehat{M}_3$. However, unlike the matrix case, it is NP-hard to find the best rank-1 approximation of a general symmetric order-3 tensor \cite{hillar2013most}. To bypass this computational barrier, we reduce the problem to that of finding the eigenvector corresponding to the largest absolute eigenvalue in a certain matrix. We obtain such a matrix by contracting the tensor with a random vector. Throughout this section we will use the notation $u^{\ast} = \kappa_3(\gamma)^{1/3}\frac{\Sigma w}{\norm{\Sigma^{1/2}w}}$. It then follows that the population third central moment $M_3=(u^{\ast})^{\otimes 3}$.
Our estimator $\hat{u}$ will approximate $\frac{u^*}{\norm{u^*}} = - \frac{\Sigma w}{\norm{\Sigma w}}$.
We will require the following claims for our proof that will follow next. We defer the proofs to \Cref{APP:facts:matrixvectorcalc}.
\newcommand{\claimbestrankoneapproximationtext}{
    Let $A,B\in\Rsymb^{d\times d}$ be symmetric matrices with $B$ being an arbitrary rank-1 matrix.
    Consider the spectral decomposition of $A$,
    $$
    A=\sum_{i=1}^{d}\lambda_i v_i v_i^\top
    $$
    where the vectors $\{v_i\}_{i=1}^{d}$ form an orthonormal basis of $\R^d$ and the singular values $\{\lambda_i\}_{i=1}^d$ satisfy ${|\lambda_1|\geq |\lambda_2|\geq\cdots\geq |\lambda_d|}$. We then have that
    \[
    \Norm{\lambda_1 v_1 v_1^\top-B}_F^2\leq 2\iprod{\lambda_1 v_1 v_1^\top-B,A-B}.
    \]
}
\begin{claim}\label{claim:best_rank-1_approximation}\claimbestrankoneapproximationtext    
\end{claim}
\newcommand{\claimcloseuptosigntext}{
    Let $a,b\in\mathbb{R}^d$.
    Let $\beta > 0$.
    Suppose that $\Norm{aa^\top-bb^\top}_F\leq\beta$.
    Then, we have that, for $u=\frac{a}{\norm{a}}$ and $v=\frac{b}{\norm{b}}$,
    $$
    \min\{\norm{u+v}^2,\norm{u-v}^2\}\leq 8\frac{\beta^2}{\norm{a}^4}
    $$
}
\begin{claim}\label{claim:close-upto-sign}\claimcloseuptosigntext
\end{claim}

The main result of this section is the following lemma that shows that given an accurate estimate for $M_3$, we can recover the direction of $u^*$.
We will show afterwards how many samples we need to get such an estimate for $M_3$.

\begin{lemma}\label{lem:jenrich}
    Given $\beta > 0$ and a symmetric tensor $\widehat{M}_3\in\R^{d\times d\times d}$ satisfying
    $$
    \sup_{\norm{v}=1}\Iprod{v^{\otimes 3}, M_3-\widehat{M}_3}\leq\beta,
    $$
    one can compute a unit vector $\hat{u}\in\Rsymb^d$ such that with probability at least $1-10^{-4}-e^{-2d}$
    $$
    \min\Set{\Norm{\hat{u}+\frac{u^{\ast}}{\norm{u^{\ast}}}}_2,\Norm{\hat{u}-\frac{u^{\ast}}{\norm{u^{\ast}}}}_2}\leq O\Paren{\frac{\sqrt{d}}{|\kappa_3(\gamma)|}\beta}.
    $$
\end{lemma}
\begin{proof}
To compute the unit vector $\hat{u}$, we first draw $g \sim \cN(0, I_d)$ and compute the matrix
\[
R = (I_d \otimes I_d \otimes g^\top) \cdot \widehat{M}_3
\]
We then compute the eigenvector corresponding to the largest (absolute) eigenvalue of $R$. We will show that this gives an estimator $\hat{u}$ with the claimed guarantees.

    We begin by observing that $(I_d \otimes I_d \otimes g^\top) \cdot \widehat{M}_3$ is a symmetric matrix.
    Let $\hat{u}$ be the (unit) eigenvector associated with its largest eigenvalue in absolute value, and $\lambda$ be the corresponding eigenvalue.
    Then $\lambda\cdot \hat{u}\hat{u}^\top$ is the best rank-1 approximation of $(I_d \otimes I_d \otimes g^\top) \cdot \widehat{M}_3$.
    Our goal is to show that $\hat{u}$ is close to $\frac{u^\ast}{\Norm{u^\ast}}$.
    
    First, we will show that $(I_d \otimes I_d \otimes g^\top) \cdot \widehat{M}_3$ and $(I_d \otimes I_d \otimes g^\top) \cdot M_3$ are close in Frobenius norm.
    Then, since $(I_d \otimes I_d \otimes g^\top) \cdot M_3= \Iprod{g, u^{\ast}}\cdot u^{\ast}(u^{\ast})^{\top}$ has rank 1, we will be able to use \Cref{claim:close-upto-sign} that relates closeness in Frobenius norm of rank-1 matrices to the closeness of the vectors themselves to conclude the lemma.
    Invoking \Cref{claim:best_rank-1_approximation}, we obtain
    \begin{align*}
        \norm{\lambda \hat{u}\hat{u}^\top-\Iprod{g, u^{\ast}}u^{\ast}(u^{\ast})^{\top}}_F^2&\leq 2\iprod{\lambda \hat{u}\hat{u}^\top-\Iprod{g, u^{\ast}} \cdot u^{\ast}(u^{\ast})^{\top}, (I_d \otimes I_d \otimes g^\top) \cdot (\widehat{M}_3-M_3)}\\
        &\overset{\text{(i)}}{\leq} 2\norm{\lambda \hat{u}\hat{u}^\top-\Iprod{g, u^{\ast}}\cdot u^{\ast}(u^{\ast})^{\top}}_{\ast}\cdot\norm{(I_d \otimes I_d \otimes g^\top) \cdot (\widehat{M}_3-M_3)}\\
        &\overset{\text{(ii)}}{\leq} 2\sqrt{2}\norm{\lambda \hat{u}\hat{u}^\top-\Iprod{g, u^{\ast}}\cdot u^{\ast}(u^{\ast})^{\top}}_F\cdot\Norm{(I_d \otimes I_d \otimes g^\top) \cdot (\widehat{M}_3-M_3)}
    \end{align*}
    In inequality (i) we used matrix H\"older's inequality. In (ii) we used that $\lambda \hat{u}\hat{u}^\top - \Iprod{g, u^{\ast}}u^{\ast}(u^{\ast})^\top$ has rank at most 2 and thus its nuclear norm is bounded by $\sqrt{2}$ times its Frobenius norm.
    We can do the exact same computation with $-g$ instead of $g$ and thus we get
    \[
        \norm{\lambda \hat{u}\hat{u}^\top\pm \Iprod{g, u^{\ast}}u^{\ast}(u^{\ast})^{\top}}_F \leq 2\sqrt{2} \Norm{(I_d \otimes I_d \otimes g^\top) \cdot (\widehat{M}_3-M_3)}.
    \]
    From our assumption we have
    \[
        \Norm{(I_d \otimes I_d \otimes g^\top) \cdot (\widehat{M}_3-M_3)} \leq \|g\| \cdot \sup_{\norm{v}=1}\iprod{v^{\otimes 3},\widehat{M}_3-M_3} \leq \Norm{g} \cdot \beta.
    \]
    Using \Cref{fact:gaussian-concentration} about the norm of a Gaussian random vector we have
    \[
        \Pr\Brac{\norm{g}\geq 3\sqrt{d}}\leq\Pr\Brac{\norm{g}-\Esymb\norm{g}\geq 2\sqrt{d}}\leq e^{-2d}.
    \]
    Thus, conditioning on the above event we have with probability at least $1 - e^{-2d}$ that
    \[
        \norm{\lambda \hat{u}\hat{u}^\top \pm \Iprod{g, u^{\ast}}u^{\ast}(u^{\ast})^{\top}}_F \leq 6\sqrt{2} \cdot \sqrt{d} \cdot \beta.
    \]
    Dividing by $\Abs{\Iprod{g,u^{\ast}}}$ and choosing the appropriate sign, we get that
    \[
        \Norm{\frac{\Abs{\lambda}}{\Abs{\Iprod{g,u^{\ast}}}} \hat{u}\hat{u}^\top-u^{\ast}(u^{\ast})^{\top}}_F \leq 6\sqrt{2} \cdot \sqrt{d} \cdot \beta \cdot \frac{1}{\Abs{\Iprod{g, u^{\ast}}}}.
    \]
    Since $\iprod{g,u^{\ast}}$ is a Gaussian of variance $\norm{u^{\ast}}^2$, \Cref{fact:gaussianspread} implies that with probability at least ${1-10^{-4}}$ we have $|\iprod{g,u^{\ast}}|\geq 10^{-4}\norm{u^{\ast}}$.
    Combining this with the above using a union bound, with probability $1-e^{-2d}-10^{-4}$ we have
    \[
        \Norm{\frac{\Abs{\lambda}}{\Abs{\Iprod{g,u^{\ast}}}} \hat{u}\hat{u}^\top-u^{\ast}(u^{\ast})^\top}_F \leq 6\sqrt{2} \cdot \sqrt{d} \cdot \beta \cdot \frac{1000}{\Norm{u^{\ast}}}.
    \]
    Using \Cref{claim:close-upto-sign} that relates Frobenius closeness to closeness of vectors themselves we obtain 
    \[
        \min\{\norm{\hat{u}-\tilde{u}},\norm{\hat{u}+\tilde{u}}\}\leq \frac{24000 \sqrt{d} \cdot \beta }{\Norm{u^{\ast}}^3},
    \]
    where $\tilde{u} \coloneqq \frac{u^{\ast}}{\norm{u^{\ast}}}$.
    Now, since $\Sigma\succeq \frac{1}{2}I_d$, we get that
    \[
        \norm{u^{\ast}}^3= |\kappa_3(\gamma)|\cdot\Norm{\frac{\Sigma w}{\norm{\Sigma^{1/2}w}}}^3\geq |\kappa_3(\gamma)|\cdot\lambda_{\min}(\Sigma^{1/2})^3\geq\frac{|\kappa_3(\gamma)|}{2\sqrt{2}}.
    \]
    Therefore we have
    \[
            \min\{\norm{\hat{u}-\tilde{u}},\norm{\hat{u}+\tilde{u}}\}\lesssim \frac{\sqrt{d} \cdot \beta}{|\kappa_3(\gamma)|},
    \]
 completing the proof.
    
\end{proof}
So far we have shown how to recover a unit vector that is (approximately) parallel to $\Sigma w$ using a random contraction on the sample third moment. To this end, we required a strong assumption on the finite-sample error to which we learn the third moment in the \emph{weaker} injective tensor norm as opposed to the Frobenius norm. In particular, we required this error to be as small as inverse polynomial in the dimension (up to $\kappa_2, \kappa_3$ factors) for the statement of the previous lemma. We will now show that we can indeed achieve such a high-accuracy estimate in the injective tensor norm using only $O(d^2)$ samples (up to $\kappa_2, \kappa_3$ factors). We start with the following claim that we will require. 

\begin{claim}\label{claim:tensor-projection-concentration}
    Let $\nu$ be the mean of $\truncatedGauss{\mu}{\Sigma}{\iprod{w,x}\leq\tau}$. 
    Let $v$ be a fixed unit vector.
    Given $n$ i.i.d. samples $\{x_i\}_{i=1}^{n}$ from $\truncatedGauss{\mu}{\Sigma}{\iprod{w,x}\leq\tau}$, for every $t > 0$, we have that 
    \[
    \Pr\Brac{\Abs{\frac{1}{n}\sum_{i=1}^{n} \iprod{v,x_i-\nu}^3-\iprod{v,u^{\ast}}^3}\geq C \kappa_2(\gamma)^{-3/2}\left(\sqrt{\frac{t}{n}}+\frac{t^{3/2}}{n}\right)}\leq 2\exp(-t).
    \]
    for an absolute constant $C > 0$.
\end{claim}
\begin{proof}
    The above claim follows by a direct application of \Cref{corollary:subgaussiancubeconcentration} for the concentration of averages-of-cubes of sub-Gaussian random variables that we proved earlier. Since $\iprod{v,x_i-\nu}$ is sub-Gaussian with parameter $O\left(1/\sqrt{\kappa_2(\gamma)}\right)$ (after preconditioning) and $\E\left[ \iprod{v,x_i-\nu}^3 \right]=\iprod{v,u^{\ast}}^{3}$, we apply \Cref{corollary:subgaussiancubeconcentration} with $Z_i = \iprod{v,x_i-\nu}$ and $\sigma = O\left(1/\sqrt{\kappa_2(\gamma)}\right)$ to obtain the result.
\end{proof}

We remark that the above claim argues about concentration of $\emph{zero-mean}$ random variables. We will now account for the error incurred by finite sample approximations to $\nu$. More formally we show the following.
\newcommand{\claimsamplemeannoise}{
    Let $\nu$ be the mean of $\truncatedGauss{\mu}{\Sigma}{\iprod{w,x}\leq\tau}$. 
    Let $v$ be a fixed unit vector.
    Let $\{x_i\}_{i=1}^{n}$ be i.i.d samples from $\truncatedGauss{\mu}{\Sigma}{\iprod{w,x}\leq\tau}$ and let $\bar{x}=\frac{1}{n}\sum_{i=1}^n x_i$. Then, we have the following with probability at least $1-4e^{-10d}$
    \[
    \Abs{\frac{1}{n}\sum_{i=1}^{n} \iprod{v, x_i-\bar{x}}^3-\iprod{v,x_i-\nu}^3} \lesssim \kappa_2(\gamma)^{-3/2}\cdot\Paren{\sqrt{\frac{d}{n}}+\Paren{\frac{d}{n}}^{3/2}}.
    \]
}
\begin{claim}\label{claim:sample-mean-noise}\claimsamplemeannoise    
\end{claim}
We defer the proof of \Cref{claim:sample-mean-noise} to the \Cref{sec:erroranalysissamplesmoments}.
Given the above concentration results, we will now present the main sample complexity result for this section. For $n = O(d^2)$ (up to $\kappa_2, \kappa_3$ factors), the result below enables us to satisfy the requirements for the main result of this section \Cref{lem:jenrich}.
\newcommand{\lemtensorspectralconcentrationtext}{
    Let $\{x_i\}_{i=1}^{n}$ be i.i.d. samples from $\truncatedGauss{\mu}{\Sigma}{\iprod{w,x}\leq\tau}$.
    Let $\bar{x}=\frac{1}{n}\sum_{i=1}^n x_i$. Let ${\widehat{M}_3=\frac{1}{n}\sum_{i=1}^{n}\Paren{x_i-\bar{x}}^{\otimes 3}}$. Then we have with probability at least $1-6e^{-7d}$ that the following holds
    \[
        \sup_{\norm{v}=1}\Iprod{v^{\otimes 3},\widehat{M}_3-M_3}\lesssim \kappa_2(\gamma)^{-3/2}\left(\sqrt{\frac{d}{n}}+\frac{d^{3/2}}{n}\right).
    \]
}
\begin{lemma}\label{lem:tensor-spectral-concentration}\lemtensorspectralconcentrationtext
\end{lemma}
We give a proof sketch here and show that the above holds for any \emph{fixed} $v$ with probability $1-6e^{-10d}$.
The lemma then follows by a standard $\varepsilon$-net argument over the unit sphere.
We defer the full proof to \Cref{sec:erroranalysissamplesmoments}.

\begin{proof}[Proof sketch]
    For a fixed vector $v$, taking $t = \Theta(d)$ in  \Cref{claim:tensor-projection-concentration}, we have for an absolute constant $K_1 > 0$ that
    \begin{align*}
    &\Pr\Brac{\Abs{\frac{1}{n}\sum_{i=1}^{n} \iprod{v,x_i-\nu}^3-\iprod{v,u^{\ast}}^3}\geq K_1\cdot\kappa_2(\gamma)^{-3/2}\left(\sqrt{\frac{d}{n}}+\frac{d^{3/2}}{n}\right)}
    \leq 2e^{-10d}.
    \end{align*}
    For the same fixed choice of $v$, we have for some other absolute constant $K_2>0$ from \Cref{claim:sample-mean-noise} that
    \[
    \Pr\Brac{\Abs{\frac{1}{n}\sum_{i=1}^{n} \iprod{v, x_i-\bar{x}}^3-\iprod{v,x_i-\nu}^3}\geq K_2\cdot\kappa_2(\gamma)^{-3/2}\cdot\Paren{\sqrt{\frac{d}{n}}+\Paren{\frac{d}{n}}^{3/2}}}\leq 4e^{-10d}.
    \]
    Now since we have that 
    \[
    \Abs{\Iprod{v^{\otimes 3},\widehat{M}_3-M_3}} \leq \Abs{\frac{1}{n}\sum_{i=1}^{n} \iprod{v, x_i-\bar{x}}^3-\iprod{v,x_i-\nu}^3} + \Abs{\frac{1}{n}\sum_{i=1}^{n} \iprod{v,x_i-\nu}^3-\iprod{v,u^{\ast}}^3},
    \] we can conclude that 
    \[
        \Abs{\Iprod{v^{\otimes 3},\widehat{M}_3-M_3}} \leq K\cdot\kappa_2(\gamma)^{-3/2}\left(\sqrt{\frac{d}{n}}+\frac{d^{3/2}}{n}\right)
    \]
    with probability at least $1 - 6e^{-10d}$ via a union bound.
    So far, we have argued that for a fixed vector $v$, it is the case that $\Iprod{v^{\otimes 3}, \widehat{M}_3-M_3}$ is well concentrated. What remains is to show that the supremum of the quantity over the unit sphere can still be controlled. To this end, we construct an $\epsilon$-net, and complete the proof with a union bound argument over this net in \Cref{sec:erroranalysissamplesmoments}. 
\end{proof}

\subsection{\Truncparaword{ recovery }}\label{SEC:fullproof:relativetruncpara}
Our goal in the section will be to recover the \truncparaword{} $\truncpara$. From the previous section, we have a good estimate for the direction of $\Sigma w$. In this section we will show how we can use this vector to recover the \truncparaword{} $\gamma$. Recovering $\truncpara$ allows us to recover the parameters of the Gaussian \emph{directly} from the moments.
We first show how to get an estimate $\sk(\hat{\gamma})$ for $\sk(\gamma)$.
Then, we show that this implies generally (meaning without any assumptions except that $\sk(\hat{\gamma})$ and $\sk(\gamma)$ are close) that also $\hat{\gamma}$ is close to $\gamma$.

\begin{lemma}
\label{lem:threshold-estimation:closeinskew}
    Let $0<\beta<c$, where $c$ is some sufficiently small absolute constant. Assume that 
    \begin{align*}
        &\min\Set{\Norm{\hat{u}+\frac{u^{\ast}}{\norm{u^{\ast}}}}_2,\Norm{\hat{u}-\frac{u^{\ast}}{\norm{u^{\ast}}}}_2}\leq\beta, \\ 
        &\norm{M_2^{-1/2}\widehat{M}_2 M_2^{-1/2}-I_d}\leq\beta
    \end{align*}
    for a unit vector $\hat{u}$ and a symmetric matrix $\widehat{M}_2\succ 0$. Furthermore, let $\widehat{M}_3$ be a symmetric tensor satisfying
    \[
    \sup_{\norm{v}=1}\Iprod{v^{\otimes 3}, \widehat{M}_3-M_3}\leq\beta.
    \]
    Denoting by $\sk^{-1}$ the inverse function of $\sk$, consider the estimator
    \begin{align*}
        \hat{\gamma} \coloneqq \sk^{-1}\Paren{-\Paren{\hat{u}^\top\widehat{M}_2^{-1} \hat{u}}^{3/2}\cdot \lvert\iprod{\hat{u}^{\otimes 3},\widehat{M}_3}\rvert}.
    \end{align*}
    Then, we have that
    \begin{align*}
        \Abs{\sk(\hat{\gamma})-\sk(\gamma)}&\leq O\left(\beta\cdot\kappa_2(\gamma)^{-3/2}\right)
    \end{align*}
    as well as
    \begin{align*}
        \left\lvert \hat{u}^\top\widehat{M}_2^{-1} \hat{u}-{\tilde{u}}^\top M_2^{-1} \tilde{u}\right\rvert \leq O(\beta).
    \end{align*}
\end{lemma}

At a high level, the above lemma say that we can learn the \truncparaword{} accurately given access to high-accuracy estimates for $\Sigma w$, $M_2$ and $M_3$ by inverting the one-dimensional $\sk$ function. We will require the following claim that bounds $\Vert u^\ast \Vert$.

\begin{claim}\label{claim:upper-bound-on-u-ast}
For $u^{\ast} \coloneqq \kappa_3(\gamma)^{1/3}\frac{\Sigma w}{\norm{\Sigma^{1/2}w}}$, we have \(\frac{1}{\sqrt{2}}\abs{\kappa_3(\gamma)}^{1/3}\leq\norm{u^{\ast}}\leq\sqrt{2}\cdot |\sk(\gamma)|^{1/3}\).
\end{claim}
\begin{proof}
    First, note that $\frac{\Sigma w}{\norm{\Sigma^{1/2}w}} = \Sigma^{1/2} \cdot \frac{\Sigma^{1/2}w}{\norm{\Sigma^{1/2}w}}$.
    By \Cref{lem:preconditioner}, we have $\frac{1}{\sqrt{2}}I_d\preceq\Sigma^{1/2}\preceq \frac{\sqrt{2}}{\kappa_2(\gamma)^{1/2}}I_d$ and thus
    \[
        \frac{1}{\sqrt{2}} \leq \Norm{\Sigma^{1/2} \frac{\Sigma^{1/2}w}{\norm{\Sigma^{1/2}w}}} \leq \frac{\sqrt{2}}{\kappa_2(\gamma)^{1/2}}.
    \]
    This proves the claim by recalling that $\sk(\gamma) = \kappa_3(\gamma)/\kappa_2(\gamma)^{3/2}$.
\end{proof}

\begin{proof}[Proof of \Cref{lem:threshold-estimation:closeinskew}]
    The proof of the lemma has two main parts: a population level analysis and a finite-sample error analysis. We further split the finite sample error analysis into two sub-parts. 
    
    We will now begin with a population-level analysis that will show that knowing $\Sigma w$, $M_2$ and $M_3$ will allow us to recover $\gamma$. Recalling the definition of the population second central moment we have
    \begin{align*}
        M_2 = \Sigma - \frac{1 - \kappa_2(\gamma)}{\norm{\Sigma^{1/2} w}^2} \cdot \left( \Sigma  w \right) \left( \Sigma w \right)^\top
    \end{align*}
    This implies that when applying the linear transformation $M_2$ on $w$, we get that
    \begin{align*}
        M_2 w = \kappa_2(\gamma) \cdot  \Sigma w \implies w = \kappa_2(\gamma) \cdot M_2^{-1} \Sigma w 
    \end{align*}
    Now by left-multiplying both sides by $(\Sigma w)^\top$ and dividing by $\kappa_2(\gamma)$ we obtain that
    \begin{align*}
        (\Sigma w)^\top M_2^{-1} (\Sigma w) = \frac{\norm{\Sigma^{1/2}w}^2}{\kappa_2(\gamma)}.
    \end{align*}
    Recall that $u^{\ast}=\kappa_3(\gamma)^{1/3}\frac{\Sigma w}{\norm{\Sigma^{1/2}w}}$.
    We define ${\tilde{u} = \frac{u^*}{\Vert u^*\Vert}} = -\frac{\Sigma w}{\norm{\Sigma w}}$ as the unit vector along the direction of $u^{\ast}$ (note that $\kappa_3(\gamma) < 0$ by \Cref{claim:monotonicity-of-kappa2}).
    Then we have
    \begin{equation}\label{eqn:secondmomentwithsigmaw}
        {\tilde{u}}^\top M_2^{-1} \tilde{u} = \frac{\norm{\Sigma^{1/2}w}^2}{\kappa_2(\gamma) \Vert \Sigma w \Vert^2} = \frac{\kappa_3(\gamma)^{2/3}}{\kappa_2(\gamma) \Vert u^{\ast} \Vert^2} 
    \end{equation}
    Furthermore recalling that the population third central moment is $M_3=\paren{u^{\ast}}^{\otimes 3}$ and taking its projection along $\tilde{u}$ we have that 
    \begin{align*}
        \langle M_3, {\tilde{u}}^{\otimes 3} \rangle = \langle u^{\ast}, \tilde{u} \rangle^3 = -\kappa_3(\gamma) \left\langle \frac{\Sigma w}{\norm{\Sigma^{1/2}w}}, \frac{\Sigma w}{\Vert \Sigma w \Vert} \right\rangle^3 = \Abs{\kappa_3(\gamma)} \cdot \frac{\Vert \Sigma w \Vert^3}{\norm{\Sigma^{1/2}w}^3} = \|u^*\|^3.
    \end{align*}
    This naturally gives us the following identity that we will need for building our estimator:
    \begin{align*}
        -\left({\tilde{u}}^\top M_2^{-1} \tilde{u}\right)^{3/2} \cdot \langle M_3, {\tilde{u}}^{\otimes 3} \rangle = \frac{\kappa_3(\gamma)}{\kappa_2(\gamma)^{3/2}} 
    \end{align*}
    
    In the rest of the proof, we will carry out finite-sample error analysis since we only have access to the sample moments. We will first deal with the first term on the left hand side above.
    \subsection*{Error analysis for $\left({\tilde{u}}^\top M_2^{-1} \tilde{u}\right)^{3/2}$}
    
    Observe that, since $\widehat{M}_2^{-1}$ is symmetric, we get that
    \begin{align*}
        \hat{u}^\top\widehat{M}_2^{-1}\hat{u} &=\tilde{u}^\top\widehat{M}_2^{-1}\tilde{u}+\left(\hat{u}-\tilde{u}\right)^\top \widehat{M}_2^{-1}\left(\hat{u}+\tilde{u}\right)\\
        &=\tilde{u}^\top M_2^{-1}\tilde{u}+\tilde{u}^\top(\widehat{M}_2^{-1}-M_2^{-1})\tilde{u} +\left(\hat{u}-\tilde{u}\right)^\top \widehat{M}_2^{-1}\left(\hat{u}+\tilde{u}\right)\\
        &=\frac{\kappa_3(\gamma)^{2/3}}{\norm{u^{\ast}}^2\cdot\kappa_2(\gamma)}+\tilde{u}^\top\left(\widehat{M}_2^{-1}-M_2^{-1}\right)\tilde{u}+\left(\hat{u}-\tilde{u}\right)^\top \widehat{M}_2^{-1}\left(\hat{u}+\tilde{u}\right)
    \end{align*}
    We will now bound the second and third term in the last equation above. For the second term, since $(1-\beta)M_2\preceq \widehat{M}_2 \preceq (1+\beta)M_2$, it follows that $\frac{1}{1+\beta}M_2^{-1}\preceq\widehat{M}_2^{-1}\preceq\frac{1}{1-\beta}M_2^{-1}$. Therefore we have
    \begin{align*}
        \norm{M_2^{-1}-\widehat{M}_2^{-1}}&=\norm{\widehat{M}_2^{-1}(\widehat{M}_2-M_2)M_2^{-1}}\\
        &=\norm{\widehat{M}_2^{-1}M_2^{1/2}(M_2^{-1/2}\widehat{M}_2 M_2^{-1/2}-I_d)M_2^{-1/2}}\\
        &\leq\beta\cdot\norm{\widehat{M}_2^{-1}}\cdot\norm{M_2^{1/2}}\cdot\norm{M_2^{-1/2}}\\
        &\overset{\text{(i)}}{\leq}\frac{4 \beta}{1-\beta} \leq 8\beta,
    \end{align*}
    as long as $\beta < \frac{1}{2}$, where the inequality (i) follows from using our preconditioning guarantees proved in \Cref{lem:preconditioner}. Therefore
    \[
        \Abs{\tilde{u}^\top(\widehat{M}_2^{-1}-M_2^{-1})\tilde{u}}\leq 8\beta.
    \]
    Now for the third term we have that
    \[
    \lvert (\hat{u}-\tilde{u})^\top \widehat{M}_2^{-1}(\hat{u}+\tilde{u})\rvert\leq \norm{\hat{u}-\tilde{u}}\cdot\norm{\hat{u}+\tilde{u}}\cdot\norm{\widehat{M}_2^{-1}}\leq O(\beta)
    \]
    whenever $\beta < 1/2$. Putting the two bounds together, we obtain,
    \begin{equation*}
    \left\lvert \hat{u}^\top\widehat{M}_2^{-1} \hat{u}-{\tilde{u}}^\top M_2^{-1} \tilde{u}\right\rvert  = 
    \left\lvert \hat{u}^\top\widehat{M}_2^{-1} \hat{u}-\frac{\kappa_3(\gamma)^{2/3}}{\norm{u^{\ast}}^2\cdot\kappa_2(\gamma)}\right\rvert\leq O(\beta)
    \end{equation*}

    We are almost done as we have shown closeness of the quadratic forms. We will now show that raising the quadratic forms to a higher power of $3/2$ still preserves closeness and conclude this part of the proof next. We will do this by an application of the mean value theorem from univariate calculus. To this end, we will try to obtain a bound on the range of values that the quadratic forms can take. Firstly, it is clear that 
    \[
    \hat{u}^\top\widehat{M}_2^{-1}\hat{u} > 0, \ \ \ \frac{\kappa_3(\gamma)^{2/3}} {\norm{u^{\ast}}^2\cdot\kappa_2(\gamma)}>0.
    \]
    Now since $\Vert u^\ast \Vert \geq \left|\kappa_3(\gamma)^{1/3}\right| \cdot \lambda_{\min}(\Sigma^{1/2})$
    we have that 
    \begin{equation}\label{eqn:skew-upper-bound}
        \frac{\kappa_3(\gamma)^{2/3}}{\norm{u^{\ast}}^2\cdot\kappa_2(\gamma)} \leq\frac{1}{\kappa_2(\gamma)\lambda_{\min}^{2}(\Sigma^{1/2})}=\frac{1}{\kappa_2(\gamma)\lambda_{\min}(\Sigma)}\leq\frac{2}{\kappa_2(\gamma)}
    \end{equation} 
    where we used the fact that $\Sigma\succeq\frac{1}{2}I_d$. Applying the mean value theorem for the function ${f(x) = x^{3/2}}$, we know that there exists a $B$ between $a = \hat{u}^\top\widehat{M}_2^{-1} \hat{u}$ and $c = \frac{\kappa_3(\gamma)^{2/3}}{\norm{u^{\ast}}^2\cdot\kappa_2(\gamma)}$ such that the following holds 
    \begin{align*}
        | f(a) - f(c)| = f'(B) \cdot |a - c | = \frac{3}{2} \cdot \sqrt{B} \cdot | a - c |.
    \end{align*} 
    We have $B = O(1/\kappa_2(\gamma))$ since
    \begin{align*}
    B&\leq \max\Set{\hat{u}^\top\widehat{M}_2^{-1} \hat{u}, \frac{\kappa_3(\gamma)^{2/3}}{\norm{u^{\ast}}^2\cdot\kappa_2(\gamma)}} \leq \frac{\kappa_3(\gamma)^{2/3}}{\norm{u^{\ast}}^2\cdot\kappa_2(\gamma)}+\left\lvert \hat{u}^\top\widehat{M}_2^{-1} \hat{u}-\frac{\kappa_3(\gamma)^{2/3}}{\norm{u^{\ast}}^2\cdot\kappa_2(\gamma)}\right\rvert\\
    &\overset{\text{(i)}}{\leq} \frac{2}{\kappa_2(\gamma)}+O(\beta) \overset{\text{(ii)}}{\leq} O\Paren{\frac{1}{\kappa_2(\gamma)}}.
    \end{align*}
    where (i) follows from inequality \eqref{eqn:skew-upper-bound} and (ii) holds as long as $\beta$ is at most a sufficiently small constant. Therefore we have that
    \[
    \left\lvert \Paren{\hat{u}^\top\widehat{M}_2^{-1} \hat{u}}^{3/2}-\frac{|\kappa_3(\gamma)|}{\norm{u^{\ast}}^3\cdot\kappa_2(\gamma)^{3/2}}\right\rvert \leq O(\beta\cdot\kappa_2(\gamma)^{-1/2}),
    \]
    concluding the first part of our error analysis.

\subsection*{Error analysis for $\langle M_3, {\tilde{u}}^{\otimes 3} \rangle$}
    We will first consider the case that $\Norm{\hat{u}-\tilde{u}}\leq\beta$ and analyze the case $\Norm{\hat{u}+\tilde{u}}\leq\beta$ afterwards.
    We have that
    \begin{align*}
        \iprod{\hat{u}^{\otimes 3}, \widehat{M}_3}&=\iprod{\hat{u}^{\otimes 3},\widehat{M}_3-M_3}+ \iprod{\hat{u}^{\otimes 3}, M_3}\\
        &=\iprod{\hat{u}^{\otimes 3},\widehat{M}_3-M_3} + \norm{u^{\ast}}^3\cdot\Paren{1+3\Iprod{\hat{u}-\tilde{u}, \tilde{u}}+3\Iprod{\hat{u}-\tilde{u}, \tilde{u}}^2+\Iprod{\hat{u}-\tilde{u}, \tilde{u}}^3}.
    \end{align*}
    By assumption, we have that $\sup_{\norm{v}=1}\iprod{v^{\otimes 3}, \widehat{M}_3-M_3}\leq\beta$.
    Thus, by applying Cauchy-Schwarz and using $\|\hat{u}-\tilde{u}\| \leq \beta$, we get that
    \[
        \left|\iprod{\hat{u}^{\otimes 3}, \widehat{M}_3}-\norm{u^{\ast}}^3 \right|\leq \beta + O(\beta \cdot\norm{u^{\ast}}^3) \leq O(\beta),
    \]
    where we also used that $\Norm{u^\ast}^3 \leq 2\sqrt{2} |\sk(\gamma)| \leq O(1)$ by \Cref{claim:upper-bound-on-u-ast,claim:skew-atmost2}.
    Since $\|u^*\|^3 > 0$, we even get that
    \begin{equation}\label{eqn:sign-test}
        \left|\Abs{\iprod{\hat{u}^{\otimes 3}, \widehat{M}_3}}-\norm{u^{\ast}}^3\right|\leq O(\beta).
    \end{equation}
    This also holds for the case $\Norm{\hat{u}+\tilde{u}}\leq\beta$ by repeating the above argument for $-\hat{u}$ instead of $\hat{u}$.

    \subsection*{Putting the pieces together}
    We will now look at the total aggregate error that we incur. We begin by defining
    \[
    S_1=\Paren{\hat{u}^\top\widehat{M}_2^{-1} \hat{u}}^{3/2}, \ \  S_2=\left|\iprod{\hat{u}^{\otimes 3}, \widehat{M}_3}\right|
    \]
    and
    \[
    \eps_1 = S_1-\frac{|\kappa_3(\gamma)|}{\norm{u^{\ast}}^3\cdot\kappa_2(\gamma)^{3/2}}, \ \  \eps_2 = S_2-\norm{u^{\ast}}^3.
    \]
    By our computations earlier, we know that $| \eps_1 | \leq O\left(\beta \cdot \kappa_2(\gamma)^{-1/2} \right)$ and $|\eps_2| \leq O\left(\beta\right)$. 
    Now expanding $S_1 S_2$ we have 
    \begin{align*}
        S_1S_2&=\Paren{\frac{|\kappa_3(\gamma)|}{\norm{u^{\ast}}^3\cdot\kappa_2(\gamma)^{3/2}}+\epsilon_1}\Paren{\norm{u^{\ast}}^3+\epsilon_2}\\
        &=\frac{|\kappa_3(\gamma)|}{\norm{u^{\ast}}^3\cdot\kappa_2(\gamma)^{3/2}}\cdot\norm{u^{\ast}}^3+\epsilon_1\cdot\norm{u^{\ast}}^3+\epsilon_2\cdot\frac{|\kappa_3(\gamma)|}{\norm{u^{\ast}}^3\cdot\kappa_2(\gamma)^{3/2}}+\epsilon_1\epsilon_2.
    \end{align*}
    Now rearranging the above inequality gives us 
    \begin{equation}\label{eqn:close-in-skew}
    \begin{aligned}
        &\Abs{\sk(\hat{\gamma})-\sk(\gamma)} =\Abs{S_1S_2-\frac{|\kappa_3(\gamma)|}{\kappa_2(\gamma)^{3/2}}} \\
        &\qquad\leq O\Paren{\beta\cdot\kappa_2(\gamma)^{-1/2}}\cdot\norm{u^{\ast}}^3+O(\beta)\cdot\frac{|\sk(\gamma)|}{\norm{u^{\ast}}^3} +O\Paren{\beta^2 \cdot \kappa_2(\gamma)^{-1/2}}\\
        &\qquad\leq O\Paren{\beta\cdot\kappa_2(\gamma)^{-1/2}}  + O\left(\beta \cdot \frac{|\sk(\gamma)|}{|\kappa_3(\gamma)|}\right) + O\Paren{\beta^2 \cdot \kappa_2(\gamma)^{-1/2}}\\
        &\qquad\leq O\Paren{\beta\cdot \left( \kappa_2(\gamma)^{-1/2} + \kappa_2(\gamma)^{-3/2} + \beta \cdot \kappa_2(\gamma)^{-1/2}\right)}\\
        &\qquad\leq O\Paren{\beta\cdot\kappa_2(\gamma)^{-3/2}},
    \end{aligned}
    \end{equation}
    where the second inequality follows from our bounds on $\Vert u^\ast \Vert$ (cf. \Cref{claim:upper-bound-on-u-ast,claim:skew-atmost2}) and the last inequality follows from $|\kappa_2(\gamma)| \leq 1$ (cf. \Cref{claim:bound-on-kappa2}) and $\beta \leq 1$ by assumption. This completes the proof of \Cref{lem:threshold-estimation:closeinskew}.
\end{proof}

\subsection{\Truncparaword{} recovery implies parameter recovery}\label{SEC:fullproof:meancovariance}
In this section, we show that given an estimate for the \truncparaword{} as well as the vector $\tilde{u} = -\frac{\Sigma w}{\norm{\Sigma w}}$ and the first three moments, we can estimate the mean and covariance of the (untruncated) Gaussian.
\subsubsection{Mean estimation}\label{SEC:fullproof:meanestimation}
The following lemma shows that we can recover the mean.
\begin{lemma}\label{lem:putting-together-2}
    Suppose there is $\beta>0$ such that $\beta < c \cdot \kappa_2(\gamma)^{5/2}$, for some sufficiently small absolute constant $c > 0$, and there are $\bar{x},\widehat{M}_2,\hat{u}, \hat{\gamma}$ such that the following hold
    \begin{align}
    &\Norm{\bar{x}-\mu-\kappa_1(\gamma)\frac{\Sigma w}{\norm{\Sigma^{1/2}w}}}\leq\beta\label{eqn:39-1}\\
    &\norm{\Sigma^{-1/2}\Paren{\widehat{M}_2-M_2}\Sigma^{-1/2}}_F\leq\beta \label{eqn:39-2}\\
    &\Norm{\hat{u}-\frac{u^{\ast}}{\norm{u^{\ast}}}}=\Norm{\hat{u}+\frac{\Sigma w}{\norm{\Sigma w}}}\leq\beta,\label{eqn:39-3}\\
    &\left\lvert \hat{u}^\top\widehat{M}_2^{-1} \hat{u}-\frac{\kappa_3(\gamma)^{2/3}}{\norm{u^{\ast}}^2\cdot\kappa_2(\gamma)}\right\rvert\leq \beta\label{eqn:39-4}\\
    &\abs{\sk(\hat{\gamma})-\sk(\gamma)}\leq O\Paren{\frac{\beta}{\kappa_2(\gamma)^{3/2}}}\label{eqn:39-5}.
    \end{align}
    Then the mean estimator
    \[
    \hat{\mu}=\bar{x}+\frac{\kappa_1(\hat{\gamma})}{\sqrt{\kappa_2(\hat{\gamma})}}\Paren{\hat{u}^\top\widehat{M}_2^{-1}\hat{u}}^{-1/2}\hat{u}
    \]
    satisfies
    \[
    \norm{\hat{\mu}-\mu}\leq O\Paren{\frac{\beta}{\kappa_2(\gamma)^{7/2}}}.
    \]
\end{lemma}

\begin{proof}
First, using assumption \eqref{eqn:39-1}, we have that
\begin{align*}
    \Norm{\hat{\mu}-\mu}&=\Norm{\bar{x}-\mu-\kappa_1(\gamma)\frac{\Sigma w}{\norm{\Sigma^{1/2}w}}+\kappa_1(\gamma)\frac{\Sigma w}{\norm{\Sigma^{1/2}w}}+\frac{\kappa_1(\hat{\gamma})}{\sqrt{\kappa_2(\hat{\gamma})}}\Paren{\hat{u}^\top\widehat{M}_2^{-1}\hat{u}}^{-1/2}\hat{u}}\\
    &\leq\Norm{\bar{x}-\mu-\kappa_1(\gamma)\frac{\Sigma w}{\norm{\Sigma^{1/2}w}}}+\Norm{\kappa_1(\gamma)\frac{\Sigma w}{\norm{\Sigma^{1/2}w}}+\frac{\kappa_1(\hat{\gamma})}{\sqrt{\kappa_2(\hat{\gamma})}}\Paren{\hat{u}^\top\widehat{M}_2^{-1}\hat{u}}^{-1/2}\hat{u}}\\
    &\leq\beta+\underbrace{\Norm{\kappa_1(\gamma)\frac{\Sigma w}{\norm{\Sigma^{1/2}w}}+\frac{\kappa_1(\hat{\gamma})}{\sqrt{\kappa_2(\hat{\gamma})}}\Paren{\hat{u}^\top\widehat{M}_2^{-1}\hat{u}}^{-1/2}\hat{u}}}_{(\ast)}.
\end{align*}
We now want to bound $(\ast)$.
We can furthermore split this up as follows
\begin{align*}
    (\ast) &\leq \Norm{\kappa_1(\gamma)\frac{\Sigma w}{\norm{\Sigma^{1/2}w}} + \frac{\kappa_1(\gamma)}{\sqrt{\kappa_2(\gamma)}}\Paren{\hat{u}^\top\widehat{M}_2^{-1}\hat{u}}^{-1/2}\hat{u}}\\
    &\qquad + \underbrace{\Norm{\frac{\kappa_1(\gamma)}{\sqrt{\kappa_2(\gamma)}}\Paren{\hat{u}^\top\widehat{M}_2^{-1}\hat{u}}^{-1/2}\hat{u}-\frac{\kappa_1(\hat{\gamma})}{\sqrt{\kappa_2(\hat{\gamma})}}\Paren{\hat{u}^\top\widehat{M}_2^{-1}\hat{u}}^{-1/2}\hat{u}}}_{a_1 \coloneqq}\\
    &\leq a_1 + \underbrace{\Norm{\kappa_1(\gamma)\frac{\Sigma w}{\norm{\Sigma^{1/2}w}} + \kappa_1(\gamma)\frac{\norm{\Sigma w}}{\norm{\Sigma^{1/2}w}}\hat{u}}}_{a_2 \coloneqq}\\
    &\qquad + \underbrace{\Norm{\kappa_1(\gamma)\frac{\norm{\Sigma w}}{\norm{\Sigma^{1/2}w}}\hat{u} - \frac{\kappa_1(\gamma)}{\sqrt{\kappa_2(\gamma)}}\Paren{\hat{u}^\top\widehat{M}_2^{-1}\hat{u}}^{-1/2}\hat{u}}}_{a_3 \coloneqq}.
\end{align*}
Putting these together, we get that
\[
    \Norm{\hat{\mu} - \mu} \leq \beta + a_1 + a_2 + a_3
\]
and it remains to bound $a_1$, $a_2$ and $a_3$.
We do so now.

\textbf{Bounding $a_1$.}
We have that
\[
    a_1 \leq \Abs{\frac{\kappa_1(\gamma)}{\sqrt{\kappa_2(\gamma)}} - \frac{\kappa_1(\hat{\gamma})}{\sqrt{\kappa_2(\hat{\gamma})}}} \cdot \Abs{\Paren{\hat{u}^\top\widehat{M}_2^{-1}\hat{u}}^{-1/2}} \cdot \Norm{\hat{u}}.
\]
We have that $\norm{\hat{u}} \leq 1$. Thus, defining $\psi(t)=\frac{\kappa_1(t)}{\sqrt{\kappa_2(t)}}$ we get that
\[
    a_1 \leq \Abs{\Paren{\hat{u}^\top\widehat{M}_2^{-1}\hat{u}}^{-1/2}} \cdot \Abs{\psi(\hat{\gamma}) - \psi(\gamma)}.
\]
By \Cref{CLAIM:ddtskewoverddtpsi}, we have that
\[
    \Abs{\psi(\hat{\gamma}) - \psi(\gamma)} \lesssim \Abs{\sk(\hat{\gamma})-\sk(\gamma)} \cdot \max\{\kappa_2(\hat{\gamma})^{-2}, \kappa_2(\gamma)^{-2}\}.
\]
We can use assumption \eqref{eqn:39-5} to bound the first term, so it remains to give a lower bound on $\kappa_2(\hat{\gamma})$.
It follows from \Cref{CLAIM:ddtskewoverddtkappa2} and assumption \eqref{eqn:39-5} that
\[
    \Abs{\kappa_2(\gamma)-\kappa_2(\hat{\gamma})}\lesssim\abs{\sk(\hat{\gamma})-\sk(\gamma)}\leq\frac{\beta}{\kappa_2(\gamma)^{3/2}}
\]
Moreover, by choosing a sufficiently small constant in our assumption $\beta < c \cdot \kappa_2(\gamma)^{5/2}$, we have that $\kappa_2(\hat{\gamma})\geq\kappa_2(\gamma)-\frac{1}{2}\kappa_2(\gamma)=\frac{1}{2}\kappa_2(\gamma)$.
Thus, we get that
\[
    \Abs{\psi(\hat{\gamma}) - \psi(\gamma)} \lesssim \Abs{\sk(\hat{\gamma})-\sk(\gamma)} \cdot \kappa_2(\gamma)^{-2} \leq \frac{\beta}{\kappa_2(\gamma)^{7/2}}.
\]
On the other hand, by \Cref{lem:preconditioner}, we have that $M_2 \preceq 2 I_d$, which together with assumption \eqref{eqn:39-2} implies
\[
    \Abs{\hat{u}^\top\widehat{M}_2^{-1}\hat{u}}^{-1/2}\leq\lambda_{\min}\Paren{\widehat{M}_2^{-1}}^{-1/2}=\Norm{\widehat{M}_2}^{1/2}\leq \sqrt{(1+\beta)\norm{M_2}}\leq 2,
\]
as long as $\beta \leq 1$.
Together with the above this implies 
\[
    a_1 \leq O\Paren{\frac{\beta}{\kappa_2(\gamma)^{7/2} \cdot}}.
\]

\textbf{Bounding $a_2$.}
For $a_2$, we have that
\[
    a_2 \leq \Abs{\kappa_1(\gamma)} \cdot \frac{\norm{\Sigma w}}{\norm{\Sigma^{1/2} w}} \cdot \Norm{\frac{\Sigma w}{\norm{\Sigma w}} + \hat{u}}.
\]
We have that $\Abs{\kappa_1(\gamma)} \leq O(\kappa_2(\gamma)^{-1/2})$ by \Cref{claim:kappa1kappa2bounded} and that $\frac{\norm{\Sigma w}}{\norm{\Sigma^{1/2} w}} \leq \kappa_2(\gamma)^{-1/2}$ by \Cref{lem:preconditioner}.
Together with assumption \eqref{eqn:39-3}, this implies 
\[
   a_2 \leq O\Paren{\frac{\beta}{\kappa_2(\gamma)}}.
\]

\textbf{Bounding $a_3$.}
For $a_3$, we have that
\begin{align*}
    a_3 &\leq \Abs{\kappa_1(\gamma)\frac{\norm{\Sigma w}}{\norm{\Sigma^{1/2}w}} - \frac{\kappa_1(\gamma)}{\sqrt{\kappa_2(\gamma)}}\Paren{\hat{u}^\top\widehat{M}_2^{-1}\hat{u}}^{-1/2}} \cdot \Norm{\hat{u}}\\
    &\leq \frac{\Abs{\kappa_1(\gamma)}}{\sqrt{\kappa_2(\gamma)}} \cdot \Abs{\sqrt{\kappa_2(\gamma)}\frac{\norm{\Sigma w}}{\norm{\Sigma^{1/2} w}} - \Paren{\hat{u}^\top\widehat{M}_2^{-1}\hat{u}}^{-1/2}}.
\end{align*}
Now, consider $f(x)=\sqrt{x}$.
For $a, c \geq B > 0$, we have that, by the mean value theorem,
\[
    \Abs{f(a)-f(c)}\leq\frac{1}{2\sqrt{B}}\abs{a-c}.
\]
We want to apply this for $a = \kappa_2(\gamma)\frac{\norm{\Sigma w}^2}{\norm{\Sigma^{1/2} w}^2}$ and $c = \Paren{\hat{u}^\top\widehat{M}_2^{-1}\hat{u}}^{-1}$.
By \Cref{lem:preconditioner}, we have that $a \geq \frac{1}{2}\kappa_2(\gamma)$.
From assumption \eqref{eqn:39-4} and using also $u^* = \kappa_3(\gamma)^{1/3} \frac{\Sigma w}{\norm{\Sigma^{1/2}w}}$, we furthermore have
\[
    c \geq \kappa_2(\gamma)\frac{\Norm{\Sigma w}^2}{\norm{\Sigma^{1/2}w}^2}-\beta\geq\frac{1}{2} \kappa_2(\gamma) - \beta \geq \frac{1}{4}\kappa_2(\gamma),
\]
as long as $\beta \leq \frac{\kappa_2(\gamma)}{4}$.
Hence, we can choose $B = \frac{1}{4}\kappa_2(\gamma)$.
Thus, we get that
\begin{align*}
    \Abs{\sqrt{\kappa_2(\gamma)}\frac{\norm{\Sigma w}}{\norm{\Sigma^{1/2} w}} - \Paren{\hat{u}^\top\widehat{M}_2^{-1}\hat{u}}^{-1/2}} &\leq \frac{1}{\sqrt{\kappa_2(\gamma)}} \cdot \Abs{\kappa_2(\gamma)\frac{\norm{\Sigma w}^2}{\norm{\Sigma^{1/2} w}^2} - \Paren{\hat{u}^\top\widehat{M}_2^{-1}\hat{u}}^{-1}}\\
    &\leq  \frac{\beta}{\sqrt{\kappa_2(\gamma)}}
\end{align*}
by assumption \eqref{eqn:39-4}.
Together with the above as well as \Cref{claim:kappa1kappa2bounded}, we get that
\[
    a_3 \leq \beta\cdot\frac{\abs{\kappa_1(\gamma)}}{\kappa_2(\gamma)}\leq O\Paren{\beta\cdot\kappa_2(\gamma)^{-3/2}}.
\]

\textbf{Completing the proof.}
Combining the bounds on $a_1$, $a_2$ and $a_3$, we get that
\[
    \Norm{\hat{\mu} - \mu} \leq \beta + O\Paren{\frac{\beta}{\kappa_2(\gamma)^{7/2}}} + O\Paren{\frac{\beta}{\kappa_2(\gamma)}}+ O\Paren{\frac{\beta}{\kappa_2(\gamma)^{3/2}}} \leq O\Paren{\frac{\beta}{\kappa_2(\gamma)^{7/2}}},
\]
where we used $\kappa_2(\gamma) \leq 1$ and $\kappa_3 \leq O(1)$ by \Cref{claim:bound-on-kappa2,claim:kappa1kappa3bounded} for the last inequality.
\end{proof}

\subsubsection{Covariance estimation}\label{SEC:fullproof:covarianceestimation}
For covariance estimation, we need the following fact, whose proof we defer to \Cref{APP:facts:matrixvectorcalc}.
\newcommand{\factupperboundFrobeniustext}{
    For any $a, b \in \R^d$ we have that
    \(
        \Norm{aa^\top-bb^\top}_F\leq\norm{a+b}\cdot\norm{a-b}.
    \)
}
\begin{fact}\label{fact:upper-bound-Frobenius}\factupperboundFrobeniustext
\end{fact}

\begin{lemma}\label{lem:putting-together}
    Suppose there is $\beta>0$ such that $\beta< c \cdot \kappa_2(\gamma)^{5/2}$, for a sufficiently small absolute constant $c > 0$, and $\widehat{M}_2,\hat{u}, \hat{\gamma}$ such that
    \begin{align}
    &\norm{\Sigma^{-1/2}\Paren{\widehat{M}_2-M_2}\Sigma^{-1/2}}_F\leq\beta \label{eqn:25-1}\\
    &\min\left\{\Norm{\hat{u}-\frac{u^{\ast}}{\norm{u^{\ast}}}},\Norm{\hat{u}+\frac{u^{\ast}}{\norm{u^{\ast}}}}\right\}\leq\beta,\label{eqn:25-2}\\
    &\left\lvert \hat{u}^\top\widehat{M}_2^{-1} \hat{u}-\frac{\kappa_3(\gamma)^{2/3}}{\norm{u^{\ast}}^2\cdot\kappa_2(\gamma)}\right\rvert\leq \beta\label{eqn:25-3}\\
    &\abs{\sk(\hat{\gamma})-\sk(\gamma)}\leq\frac{\beta}{\kappa_2(\gamma)^{3/2}}\label{eqn:25-5}.
    \end{align}
    Then the covariance estimator
    \[
        \widehat{\Sigma}=\widehat{M}_2+\Paren{\frac{1}{\kappa_2(\hat{\gamma})}-1}\paren{\hat{u}^\top\widehat{M}_2^{-1} \hat{u}}^{-1}\cdot \hat{u}\hat{u}^\top
    \]
    satisfies
    \[
    \norm{\Sigma^{-1/2}\widehat{\Sigma}\Sigma^{-1/2}-I_d}_F\le O\Paren{\frac{\beta}{\kappa_2(\gamma)^{7/2}}}.
    \]
\end{lemma}

\begin{proof}
By the definition of our estimator, we get that
\[
    \Sigma^{-1/2}\widehat{\Sigma}\Sigma^{-1/2}-I_d = \Sigma^{-1/2}\widehat{M}_2\Sigma^{-1/2}+\Paren{\frac{1}{\kappa_2(\hat{\gamma})}-1}\paren{\hat{u}^\top\widehat{M}_2^{-1} \hat{u}}^{-1}\Sigma^{-1/2}\hat{u}\hat{u}^\top\Sigma^{-1/2}-I_d.
\]
If all $\widehat{M}_2$, $\hat{\gamma}$ and $\hat{u}$ would be equal to $M_2$, $\gamma$ and $\tilde{u}$, then this norm would be zero.
The proof now consists of accounting for the errors in these estimates.
First, substituting $\widehat{M}_2$ by $M_2$, we get that
\[
    \Sigma^{-1/2}\widehat{\Sigma}\Sigma^{-1/2}-I_d = A_1 + \Sigma^{-1/2}M_2\Sigma^{-1/2}+\Paren{\frac{1}{\kappa_2(\hat{\gamma})}-1}\paren{\hat{u}^\top\widehat{M}_2^{-1} \hat{u}}^{-1}\Sigma^{-1/2}\hat{u}\hat{u}^\top\Sigma^{-1/2}-I_d,
\]
where
\[
    A_1 = \Sigma^{-1/2}(\widehat{M}_2-M_2)\Sigma^{-1/2}.
\]
Now, we can use the formula for $M_2$ from \Cref{lem:first3moments} to get that
\begin{align*}
    &\Sigma^{-1/2}\widehat{\Sigma}\Sigma^{-1/2}-I_d =\\
    & \qquad A_1 + (1-\kappa_2(\gamma))\frac{\Sigma^{1/2}ww^\top\Sigma^{1/2}}{\Norm{\Sigma^{1/2}w}^2}-\underbrace{\Paren{\frac{1}{\kappa_2(\hat{\gamma})}-1}\paren{\hat{u}^\top\widehat{M}_2^{-1} \hat{u}}^{-1}\Sigma^{-1/2}\hat{u}\hat{u}^\top\Sigma^{-1/2}}_{(\ast)}.
\end{align*}
We now want to show that
\[
    (\ast) \approx (1-\kappa_2(\gamma))\frac{\Sigma^{1/2}ww^\top\Sigma^{1/2}}{\Norm{\Sigma^{1/2}w}^2}.
\]
Note that if we would replace $\widehat{M}_2$, $\hat{\gamma}$ and $\hat{u}$ by $M_2$, $\gamma$ and $\tilde{u}$, then they would be equal.
We will do so now and account for the error.
Substituting $\hat{\gamma}$ by $\gamma$, we get that
\[
    (\ast) = A_2 +\Paren{\frac{1}{\kappa_2(\gamma)}-1}\paren{\hat{u}^\top\widehat{M}_2^{-1} \hat{u}}^{-1}\Sigma^{-1/2}\hat{u}\hat{u}^\top\Sigma^{-1/2},
\]
where
\[
    A_2 = \Paren{\frac{1}{\kappa_2(\hat{\gamma})}-\frac{1}{\kappa_2(\gamma)}}\paren{\hat{u}^\top\widehat{M}_2^{-1} \hat{u}}^{-1}\Sigma^{-1/2}\hat{u}\hat{u}^\top\Sigma^{-1/2}.
\]
Substituting $\hat{u}^\top \widehat{M}_2^{-1}\hat{u}$ by $\frac{\kappa_3(\gamma)^{2/3}}{\norm{u^{\ast}}^2\cdot\kappa_2(\gamma)}$, we get that
\begin{align*}
    (\ast) &= A_2 + A_3 + \Paren{\frac{1}{\kappa_2(\gamma)}-1}\Paren{\frac{\kappa_3(\gamma)^{2/3}}{\norm{u^{\ast}}^2\cdot\kappa_2(\gamma)}}^{-1}\Sigma^{-1/2}\hat{u}\hat{u}^\top\Sigma^{-1/2}\\
    &= A_2 + A_3 + \paren{1-\kappa_2(\gamma)}\frac{\norm{u^{\ast}}^2}{\kappa_3(\gamma)^{2/3}}\Sigma^{-1/2}\hat{u}\hat{u}^\top\Sigma^{-1/2},
\end{align*}
where 
\[
    A_3 = \Paren{\frac{1}{\kappa_2(\gamma)}-1}\Paren{\Paren{\hat{u}^\top\widehat{M}_2^{-1} \hat{u}}^{-1}-\frac{\norm{u^{\ast}}^2\cdot\kappa_2(\gamma)}{\kappa_3(\gamma)^{2/3}}}\Sigma^{-1/2}\hat{u}\hat{u}^\top\Sigma^{-1/2}.
\]
Finally, substituting $\hat{u}$ with $\tilde{u} = \frac{\Sigma w}{\norm{\Sigma w}}$ and using that $\frac{\norm{\Sigma w}}{\norm{\Sigma^{1/2}w}} = \frac{\norm{u^*}}{\kappa_3(\gamma)^{1/3}}$, we get that
\begin{align*}
    (\ast) &= A_2 + A_3 + A_4 + \paren{1-\kappa_2(\gamma)}\frac{\norm{u^{\ast}}^2}{\kappa_3(\gamma)^{2/3}}\frac{\Sigma^{1/2}w w^\top\Sigma^{1/2}}{\norm{\Sigma w}^2}\\
    &= A_2 + A_3 + A_4 + (1-\kappa_2(\gamma))\frac{\Sigma^{1/2}ww^\top\Sigma^{1/2}}{\Norm{\Sigma^{1/2}w}^2},
\end{align*}
where
\[
    A_4 = (1-\kappa_2(\gamma))\Paren{\frac{\norm{u^{\ast}}^2}{\kappa_3(\gamma)^{2/3}}\Sigma^{-1/2}\hat{u}\hat{u}^\top\Sigma^{-1/2} - \frac{\Sigma^{1/2}ww^\top\Sigma^{1/2}}{\norm{\Sigma^{1/2}w}^2}}.
\]
Combining all of the above we get that
\[
    \Sigma^{-1/2}\widehat{\Sigma}\Sigma^{-1/2}-I_d = A_1 - A_2 - A_3 - A_4.
\]
Thus, in order to bound $\norm{\Sigma^{-1/2}\widehat{\Sigma}\Sigma^{-1/2}-I_d}_F$, it is sufficient to bound the norms of $A_1$, $A_2$, $A_3$ and $A_4$.
We will do so now.

\textbf{Bounding $\|A_1\|_F$.}
For $A_1$, note that by \eqref{eqn:25-1} we have that $\|A_1\|_F \leq \beta$.

\textbf{Bounding $\|A_2\|_F$.}
For $A_2$, we can first bound
\[
    \|A_2\|_F \leq \Abs{\frac{1}{\kappa_2(\hat{\gamma})} - \frac{1}{\kappa_2(\gamma)}} \cdot \Abs{\paren{\hat{u}^\top \widehat{M}_2^{-1}\hat{u}}^{-1}} \Norm{\Sigma^{-1/2} \hat{u} \hat{u}^\top \Sigma^{-1/2}}_F.
\]
We bound these three terms separately.
First, note that, by \Cref{lem:preconditioner},
\begin{equation}\label{eq:covariancestimator:A2:part3}
    \Norm{\Sigma^{-1/2}\hat{u}\hat{u}^\top\Sigma^{-1/2}}_F=\Norm{\Sigma^{-1/2}\hat{u}}^2\leq\norm{\Sigma^{-1/2}}^2\leq 2.
\end{equation}
For the middle term, we first want to bound
\[
    \Abs{\paren{\hat{u}^\top\widehat{M}_2^{-1} \hat{u}}^{-1}-\frac{\norm{u^{\ast}}^2\cdot\kappa_2(\gamma)}{\kappa_3(\gamma)^{2/3}}} = \Abs{\paren{\hat{u}^\top\widehat{M}_2^{-1} \hat{u}}^{-1}-\frac{\norm{\Sigma w}^2\cdot\kappa_2(\gamma)}{\norm{\Sigma^{1/2} w}^2}},
\]
where we used $u^* = \kappa_3(\gamma)^{1/3} \frac{\Sigma w}{\norm{\Sigma^{1/2}w}}$ (which implies $\frac{\norm{u^{\ast}}^2\cdot\kappa_2(\gamma)}{\kappa_3(\gamma)^{2/3}} = \frac{\norm{\Sigma w}^2\cdot\kappa_2(\gamma)}{\norm{\Sigma^{1/2} w}^2}$).
We have that, using assumption \eqref{eqn:25-3},
\begin{align*}
        \Abs{\paren{\hat{u}^\top\widehat{M}_2^{-1} \hat{u}}^{-1}-\frac{\norm{\Sigma w}^2\cdot\kappa_2(\gamma)}{\norm{\Sigma^{1/2} w}^2}} &=\Abs{\hat{u}^\top\widehat{M}_2^{-1} \hat{u}-\frac{\norm{\Sigma^{1/2} w}^2}{\norm{\Sigma w}^2\cdot\kappa_2(\gamma)}}\cdot\Abs{\paren{\hat{u}^\top\widehat{M}_2^{-1} \hat{u}}^{-1}\cdot\frac{\norm{\Sigma w}^2\cdot\kappa_2(\gamma)}{\norm{\Sigma^{1/2} w}^2}}\\
        &\leq \beta\cdot\frac{\norm{\Sigma w}^2\cdot\kappa_2(\gamma)}{\hat{u}^\top\widehat{M}_2^{-1} \hat{u}\cdot\norm{\Sigma^{1/2}w}^2}.
\end{align*}
Using \eqref{eqn:25-3} again to lower bound $\hat{u}^\top\widehat{M}_2^{-1} \hat{u}$, we can furthermore bound
\begin{align*}
    \Abs{\paren{\hat{u}^\top\widehat{M}_2^{-1} \hat{u}}^{-1}-\frac{\norm{\Sigma w}^2\cdot\kappa_2(\gamma)}{\norm{\Sigma^{1/2}w}^2}} &\leq \beta \cdot \frac{\norm{\Sigma w}^2\cdot\kappa_2(\gamma)}{\norm{\Sigma^{1/2}w}^2} \cdot \Paren{\frac{\norm{\Sigma^{1/2}w}^2}{\norm{\Sigma w}^2\cdot\kappa_2(\gamma)} - \beta}^{-1}\\
    &= \beta \cdot \frac{\norm{\Sigma w}^2\cdot\kappa_2(\gamma)}{\norm{\Sigma^{1/2}w}^2} \cdot \frac{\frac{\norm{\Sigma w}^2}{\norm{\Sigma^{1/2}w}^2} \cdot \kappa_2(\gamma)}{1 - \frac{\norm{\Sigma w}^2}{\norm{\Sigma^{1/2}w}^2} \cdot \kappa_2(\gamma) \cdot \beta}.
\end{align*}
By \Cref{lem:preconditioner}, we have that $\frac{\norm{\Sigma w}^2\cdot\kappa_2(\gamma)}{\norm{\Sigma^{1/2}w}^2}\leq 2$.
Thus, assuming that $\beta < \frac{1}{2}$, we get that
\begin{equation}\label{eqn:scaling-est}
    \Abs{\paren{\hat{u}^\top\widehat{M}_2^{-1} \hat{u}}^{-1}-\frac{\norm{u^{\ast}}^2\cdot\kappa_2(\gamma)}{\kappa_3(\gamma)^{2/3}}} = \Abs{\paren{\hat{u}^\top\widehat{M}_2^{-1} \hat{u}}^{-1}-\frac{\norm{\Sigma w}^2\cdot\kappa_2(\gamma)}{\norm{\Sigma^{1/2}w}^2}} \leq O(\beta).
\end{equation}
Using the triangle inequality and $\frac{\norm{\Sigma w}^2\cdot\kappa_2(\gamma)}{\norm{\Sigma^{1/2}w}^2}\leq 2$ again, we furthermore get that
\[
    \Abs{\paren{\hat{u}^\top\widehat{M}_2^{-1} \hat{u}}^{-1}} \leq \frac{\norm{\Sigma w}^2\cdot\kappa_2(\gamma)}{\norm{\Sigma^{1/2}w}^2} + O(\beta) \leq O(1).
\]
Finally, we want to bound
\[
    \Abs{\frac{1}{\kappa_2(\gamma)}-\frac{1}{\kappa_2(\hat{\gamma})}}=\frac{1}{\kappa_2(\gamma)\cdot\kappa_2(\hat{\gamma})} \cdot \Abs{\kappa_2(\gamma)-\kappa_2(\hat{\gamma})}.
\]
As before in the proof of \Cref{lem:putting-together-2}, it follows from \Cref{CLAIM:ddtskewoverddtkappa2} and assumption \eqref{eqn:25-5} that
\[
    \Abs{\kappa_2(\gamma)-\kappa_2(\hat{\gamma})}\lesssim\abs{\sk(\hat{\gamma})-\sk(\gamma)}\leq\frac{\beta}{\kappa_2(\gamma)^{3/2}}
\]
As before, by our assumption $\beta < c \cdot \kappa_2(\gamma)^{5/2}$, we have that ${\kappa_2(\hat{\gamma})\geq\kappa_2(\gamma)-\frac{1}{2}\kappa_2(\gamma)=\frac{1}{2}\kappa_2(\gamma)}$ and thus
\[
    \frac{1}{\kappa_2(\gamma)\cdot\kappa_2(\hat{\gamma})} \cdot \Abs{\kappa_2(\gamma)-\kappa_2(\hat{\gamma})} \leq O\Paren{\frac{\beta}{\kappa_2(\gamma)^{7/2}}}.
\]
Combining the above, we get that
\[
    \|A_2\|_F \leq 2 \cdot O(1) \cdot O\Paren{\frac{\beta}{\kappa_2(\gamma)^{7/2}}} = O\Paren{\frac{\beta}{\kappa_2(\gamma)^{7/2}}}.
\]

\textbf{Bounding $\|A_3\|_F$.}
For $A_3$, note that by \eqref{eq:covariancestimator:A2:part3} and \eqref{eqn:scaling-est} and using that $\kappa_2(\gamma) \leq 1$ (cf. \Cref{claim:bound-on-kappa2}), we immediately get that
\begin{align*}
    \|A_3\|_F &\leq \Abs{\frac{1}{\kappa_2(\gamma)}-1} \cdot \Abs{\Paren{\hat{u}^\top\widehat{M}_2^{-1} \hat{u}}^{-1}-\frac{\norm{u^{\ast}}^2\cdot\kappa_2(\gamma)}{\kappa_3(\gamma)^{2/3}}} \cdot \Norm{\Sigma^{-1/2}\hat{u}\hat{u}^\top\Sigma^{-1/2}}\\
    &\leq \frac{1}{\kappa_2(\gamma)} \cdot O(\beta) \cdot 2 = O\Paren{\frac{\beta}{\kappa_2(\gamma)}}.
\end{align*}

\textbf{Bounding $\|A_4\|_F$.}
For $A_4$, recalling that $\frac{\norm{u^{\ast}}^2}{\kappa_3(\gamma)^{2/3}} = \frac{\norm{\Sigma w}^2}{\norm{\Sigma^{1/2} w}^2}$ and $\kappa_2(\gamma) > 0$, we get that
\begin{align*}
    \|A_4\|_F &\leq \Abs{1-\kappa_2(\gamma)} \cdot \Norm{\frac{\norm{u^{\ast}}^2}{\kappa_3(\gamma)^{2/3}}\Sigma^{-1/2}\hat{u}\hat{u}^\top\Sigma^{-1/2} - \frac{\Sigma^{1/2}ww^\top\Sigma^{1/2}}{\norm{\Sigma^{1/2}w}^2}}_F\\
    &\leq \frac{\norm{\Sigma w}^2}{\norm{\Sigma^{1/2} w}^2} \cdot \Norm{\Sigma^{-1/2}\hat{u}\hat{u}^\top\Sigma^{-1/2} - \frac{\Sigma^{1/2}ww^\top\Sigma^{1/2}}{\norm{\Sigma w}^2}}_F\\
    &= \frac{\norm{\Sigma w}^2}{\norm{\Sigma^{1/2} w}^2} \cdot \Norm{\Sigma^{-1/2}\Paren{\hat{u}\hat{u}^\top - \frac{\Sigma ww^\top\Sigma}{\norm{\Sigma w}^2}}\Sigma^{-1/2}}_F\\
    &\leq \norm{\Sigma^{-1/2}}^2\cdot\frac{\norm{\Sigma w}^2}{\norm{\Sigma^{1/2}w}^2}\cdot\Norm{\hat{u}\hat{u}^\top - \frac{\Sigma w}{\norm{\Sigma w}}\Paren{\frac{\Sigma w}{\norm{\Sigma w}}}^\top}_F
\end{align*}
We now want to use assumption \eqref{eqn:25-2}. Since $u^* \propto -\Sigma w$, this assumption implies that
\[
    \min\left\{\Norm{\hat{u}-\frac{\Sigma w}{\norm{\Sigma w}}},\Norm{\hat{u}+\frac{\Sigma w}{\norm{\Sigma w}}}\right\}\leq\beta
\]
Using \Cref{fact:upper-bound-Frobenius} on the unit vectors $\hat{u}$ and $\frac{\Sigma w}{\norm{\Sigma w}}$ to get that
\[
    \Norm{\hat{u}\hat{u}^\top-\frac{\Sigma w}{\norm{\Sigma w}}\Paren{\frac{\Sigma w}{\norm{\Sigma w}}}^\top}_F\leq 2\beta.
\]
Furthermore, using \Cref{lem:preconditioner}, we get that $\norm{\Sigma^{-1/2}}^2 \leq 2$ and $\frac{\norm{\Sigma w}^2}{\norm{\Sigma^{1/2}w}^2} \leq \frac{1}{\kappa_2(\gamma)}$.
Combining these with the above, we get that
\[
    \|A_4\|_F \leq 2 \cdot \frac{1}{\kappa_2(\gamma)} \cdot 2\beta = O\Paren{\frac{\beta}{\kappa_2(\gamma)}}.
\]

\textbf{Completing the proof.}
Combining the bounds on $A_1$, $A_2$, $A_3$ and $A_4$, we get that
\begin{align*}
    \Norm{\Sigma^{-1/2}\widehat{\Sigma}\Sigma^{-1/2} - I_d}_F &\leq \Norm{A_1}_F + \Norm{A_2}_F + \Norm{A_3}_F + \Norm{A_4}_F\\
    &\leq \beta + O\Paren{\frac{\beta}{\kappa_2(\gamma)^{7/2}}} + O\Paren{\frac{\beta}{\kappa_2(\gamma)}} + O\Paren{\frac{\beta}{\kappa_2(\gamma)}}\\
    &\leq O\Paren{\frac{\beta}{\kappa_2(\gamma)^{7/2}}},
\end{align*}
where the last inequality follows because $\kappa_2(\gamma) \leq 1$.
This completes the proof.
\end{proof}

\subsection{Proof of \texorpdfstring{\Cref{thm:main}}{main theorem}}\label{SEC:fullproof:puttingtogether}
In this section, we prove our main result, \Cref{thm:main}.
We start with combining the results from \Cref{SEC:fullproof:directionofsigmaw,SEC:fullproof:relativetruncpara,SEC:fullproof:meancovariance} to give a sample complexity bound of \Cref{algo:main}.
\begin{lemma}\label{lem:unified}
    Consider a halfspace truncated Gaussian $\truncatedGauss{\mu}{\Sigma}{\iprod{w,x}\leq\tau}$, where $\Sigma\succ 0$.
    Define ${\gamma=\frac{\tau-\iprod{w,\mu}}{\Norm{\Sigma^{1/2}w}}}$.
    Then, given 
    \[
        n = O\Paren{\frac{d^2}{\kappa_2(\gamma)^8\kappa_3(\gamma)^2} + \frac{d^2}{\kappa_2(\gamma)^3\kappa_3(\gamma)^4} + \frac{d^2}{\kappa_2(\gamma)^{10}\kappa_3(\gamma)^2\varepsilon^2}}
    \]
    samples, \Cref{algo:main} computes estimators $\hat{\mu},\widehat{\Sigma}$ such that with probability at least $0.99$
    \[
        \TVdist\Paren{\cN(\hat{\mu},\widehat{\Sigma}),\cN(\mu,\Sigma)}\leq \epsilon.
    \]
\end{lemma}
Later, we will then bound this sample complexity in terms of $\alpha$, as in \Cref{thm:main}.
Finally, we will also argue about the runtime of \Cref{algo:main}.
\begin{proof}[Proof of \Cref{lem:unified}]
For any $\beta > 0$, given $n=O\Paren{\frac{d^2}{\kappa_2(\gamma)^3\kappa_3(\gamma)^2\beta^2}}$ samples, by \Cref{lem:tensor-spectral-concentration}, we have that, with high probability,
\begin{equation}\label{eqn:unified:M3}
    \sup_{\norm{v}=1}\Iprod{v^{\otimes 3}, \widehat{M}_3-M_3}\leq O\Paren{\frac{\beta \cdot \Abs{\kappa_3(\gamma)}}{\sqrt{d}}} \leq O(\beta).
\end{equation}
Now, we get, by \Cref{lem:jenrich}, with high probability,
\begin{equation}\label{eqn:unified:uhat}
    \min\left\{\Norm{\hat{u}-\frac{u^{\ast}}{\norm{u^{\ast}}}},\Norm{\hat{u}+\frac{u^{\ast}}{\norm{u^{\ast}}}}\right\} \leq O(\beta).
\end{equation}
Next, consider the sample covariance $\widehat{M}_2$ computed from  $O(d^2/(\kappa_2(\gamma)^2\beta^2))$ \textrm{i.i.d.} samples.
From \Cref{claim:subgaussian:M2invX} we know each sample is $1/\sqrt{\kappa_2(\gamma)}$-sub-Gaussian, hence as in \Cref{fact:random-matrix-concentration}, we have, with high probability
\begin{equation}\label{eqn:unified:M2}
    \norm{M_2^{-1/2}\widehat{M}_2 M_2^{-1/2}-I_d}\leq\beta \quad\text{and}\quad \norm{\Sigma^{-1/2}\Paren{\widehat{M}_2-M_2}\Sigma^{-1/2}}_F\leq\beta,
\end{equation}
where we also used that $\norm{\Sigma^{-1/2}M_2^{1/2}} \leq \frac{1}{2}$ by \Cref{lem:preconditioner}.
Using equations \eqref{eqn:unified:uhat}, \eqref{eqn:unified:M2} and \eqref{eqn:unified:M3}, we can apply \Cref{lem:threshold-estimation:closeinskew} as long as $\beta < c_1$ (for some absolute constant $c_1$).
This shows that the estimator
\[
    \hat{\gamma}=\sk^{-1}\Paren{-\Paren{\hat{u}^\top\widehat{M}_2^{-1} \hat{u}}^{3/2}\cdot \lvert\iprod{u^{\otimes 3},\widehat{M}_3}\rvert}
\]
satisfies
\begin{equation}\label{eqn:unified:skewgammahat}
    \Abs{\sk(\hat{\gamma})-\sk(\gamma)}\leq O\Paren{\frac{\beta}{\kappa_2(\gamma)^{3/2}}}.
\end{equation}
From \Cref{lem:threshold-estimation:closeinskew}, we furthermore get the inequality
\begin{equation}\label{eqn:unified:uhatM2inverseuhat}
    \left\lvert \hat{u}^\top\widehat{M}_2^{-1} \hat{u}-\frac{\norm{\Sigma^{1/2}w}^2}{\norm{\Sigma w}^2\cdot\kappa_2(\gamma)}\right\rvert\leq O(\beta)
\end{equation}
As long as $\beta < c_2 \cdot \kappa_2(\gamma)^{5/2}$, using \eqref{eqn:unified:M2}, \eqref{eqn:unified:uhat}, \eqref{eqn:unified:uhatM2inverseuhat} and \eqref{eqn:unified:skewgammahat}, we can apply \Cref{lem:putting-together} to get that the covariance estimator
\[
    \widehat{\Sigma}=\widehat{M}_2+\Paren{\frac{1}{\kappa_2(\hat{\gamma})}-1}\paren{\hat{u}^\top\widehat{M}_2^{-1} \hat{u}}^{-1}\cdot \hat{u}\hat{u}^\top
\]
satisfies
\begin{equation}\label{eqn:unified:conclusioncovariance}
    \norm{\Sigma^{-1/2}\widehat{\Sigma}\Sigma^{-1/2}-I_d}_F\le O\Paren{\frac{\beta}{\kappa_2(\gamma)^{7/2}}}.
\end{equation}

For mean estimation, we have that $x_i$ is $\sqrt{2/\kappa_2(\gamma)}$-sub-Gaussian since $M_2^{-1/2}x_i$ is $1/\sqrt{\kappa_2(\gamma)}$-sub-Gaussian by \Cref{claim:subgaussian:M2invX} and $\norm{M_2^{1/2}} \leq \frac{1}{2}$ by \Cref{lem:preconditioner}.
Thus, $\bar{x}$ is $(\sqrt{2/\kappa_2(\gamma)})/\sqrt{n}$-sub-Gaussian (cf. \cite[Proposition 2.7.1]{Vershynin_2018}).
Given $n= O\Paren{d/(\kappa_2(\gamma)\beta^2)}$ samples, we therefore have (see for example \cite[Theorem 1.19]{Rigollet:meanestimationsubgaussian}), with high probability,
\begin{equation}\label{eqn:unified:samplemean}
    \Norm{\bar{x}-\mu-\kappa_1(\gamma)\frac{\Sigma w}{\norm{\Sigma^{1/2}w}}}\leq\beta.
\end{equation}
To apply \Cref{lem:putting-together-2}, it remains to argue that choosing the sign of $\hat{u}$ by computing the sign of $\iprod{\hat{u}^{\otimes 3}, \widehat{M}_3}$ indeed ensures \eqref{eqn:39-3}.
From the proof of \Cref{lem:threshold-estimation:closeinskew}, we know that
\begin{align}
    &\norm{\hat{u}-\tilde{u}} \leq O(\beta) \implies \left|\iprod{\hat{u}^{\otimes 3}, \widehat{M}_3}-\norm{u^{\ast}}^3 \right| \leq O(\beta) \label{eqn:unified:signcheck1}\\
    &\norm{\hat{u}+\tilde{u}} \leq O(\beta) \implies \left|\iprod{\hat{u}^{\otimes 3}, \widehat{M}_3}+\norm{u^{\ast}}^3 \right| \leq O(\beta).\label{eqn:unified:signcheck2}
\end{align}
If the term $O(\beta)$ on the right hand side is smaller than $\norm{u^*}^3$, then only one of the two conclusions \eqref{eqn:unified:signcheck1} or \eqref{eqn:unified:signcheck2} can be true.
By \eqref{eqn:unified:uhat}, one of the two must hold.
Thus, exactly one of the two conclusions \eqref{eqn:unified:signcheck1} or \eqref{eqn:unified:signcheck2} is true.
Furthermore, if the sign of $\iprod{\hat{u}^{\otimes 3}, \widehat{M}_3}$ is positive, then it must be \eqref{eqn:unified:signcheck1} (and hence we want to pick $\hat{u}$ for \eqref{eqn:39-3}).
If this sign is negative, it must be \eqref{eqn:unified:signcheck2} (and we want to pick $-\hat{u}$).
Hence, if $\beta < c_3 \cdot \abs{\kappa_3(\gamma)}$ (for a sufficiently small constant $c_3$ ensuring that $O(\beta) \leq \norm{u^*}^3$ by \Cref{claim:upper-bound-on-u-ast}), then flipping the sign of $\hat{u}$ if $\iprod{\hat{u}^{\otimes 3}, \widehat{M}_3} < 0$ indeed ensures that (after the potential sign flip)
\begin{equation}\label{eqn:unified:uhatwithsign}
    \Norm{\hat{u}-\frac{u^{\ast}}{\norm{u^{\ast}}}} \leq O(\beta).
\end{equation}
Thus, as long as $\beta < c_4 \cdot \kappa_2(\gamma)^{5/2}$, using \eqref{eqn:unified:samplemean}, \eqref{eqn:unified:M2}, \eqref{eqn:unified:uhatwithsign}, \eqref{eqn:unified:uhatM2inverseuhat} and \eqref{eqn:unified:skewgammahat}, we can apply \Cref{lem:putting-together-2} to get that the mean estimator
\[
    \hat{\mu}=\bar{x}+\frac{\kappa_1(\hat{\gamma})}{\sqrt{\kappa_2(\hat{\gamma})}}\Paren{\hat{u}^\top\widehat{M}_2^{-1}\hat{u}}^{-1/2}\hat{u}
\]
satisfies
\[
    \Norm{\hat{\mu} - \mu} \leq O\Paren{\frac{\beta}{\kappa_2(\gamma)^{7/2}}}.
\]
Furthermore, by \Cref{lem:preconditioner} we have that $\norm{\Sigma^{-1/2}} \leq O(1)$ and thus the above also implies
\begin{equation}\label{eqn:unified:meanconclusion}
    \Norm{\Sigma^{-1/2}(\hat{\mu} - \mu)} \leq O\Paren{\frac{\beta}{\kappa_2(\gamma)^{7/2}}}.
\end{equation}

Thus, summarized, for $\beta < c \cdot \min\{\kappa_2(\gamma)^{5/2}, \abs{\kappa_3(\gamma)}\}$ for a sufficiently small constant $c$, we have that, with high probability,\footnote{For Gaussian, closeness in parameters implies closeness in TV-distance, see e.g. \cite[Theorem 1.8]{ArbiasAshtianiLiaw:GaussianparameterclosenessimpliesTVcloseness}.}
\[
    \TVdist\Paren{\cN(\hat{\mu},\widehat{\Sigma}),\cN(\mu,\Sigma)}\leq O\Paren{\frac{\beta}{\kappa_2(\gamma)^{7/2}}}.
\]
The sample complexity is
\[
    n = O\Paren{\max\left\{\frac{d^2}{\kappa_2(\gamma)^3\kappa_3(\gamma)^2\beta^2}, \frac{d^2}{\kappa_2(\gamma)^2\beta^2}, \frac{d}{\kappa_2(\gamma)\beta^2}\right\}} = O\Paren{\frac{d^2}{\kappa_2(\gamma)^3\kappa_3(\gamma)^2\beta^2}}.
\]
Thus, setting $\beta = O\Paren{\min\{\kappa_2(\gamma)^{5/2}, \abs{\kappa_3(\gamma)}, \kappa_2(\gamma)^{7/2} \cdot \varepsilon\}}$ for a sufficiently small constant, we get the lemma.
\end{proof}

The above proof is sufficient as long as $\alpha$ is bounded away from $1$ by a constant.
As $\alpha$ approaches $1$, we have that $\kappa_3(\gamma) \to 0$ and thus the sample complexity increases.
To deal with this, we now give an alterative sample complexity bound that relies on the fact that if $\alpha$ is close to $1$, then the sample mean and covariance are good estimators and the correction terms in \Cref{algo:main} are small.
\begin{lemma}\label{lem:unifiedalphacloseto1}
    Consider a halfspace truncated Gaussian $\truncatedGauss{\mu}{\Sigma}{\iprod{w,x}\leq\tau}$, where $\Sigma\succ 0$.
    Define ${\gamma=\frac{\tau-\iprod{w,\mu}}{\Norm{\Sigma^{1/2}w}}}$.
    Assume that $\alpha = \Phi(\gamma) \geq 1 - c$ for a sufficiently small constant $c > 0$.
    Then, given $n = O\Paren{d^2/\varepsilon^2}$ samples, \Cref{algo:main} computes estimators $\hat{\mu},\widehat{\Sigma}$ such that with probability at least $0.99$
    \[
        \TVdist\Paren{\cN(\hat{\mu},\widehat{\Sigma}),\cN(\mu,\Sigma)}\leq O\Paren{\varepsilon + (1-\alpha)\log\Paren{\frac{1}{1-\alpha}}}.
    \]
\end{lemma}
\begin{proof}
    The strategy of the proof is as follows: We want to argue that we are essentially outputting the sample mean and covariance and that these are good estimates.
    First, note that given ${n = O(d^2/\beta^2)}$ samples, the sample moments satisfy, with high probability
    \begin{align*}
        \Norm{\mu + \frac{\phi(\gamma)}{\Phi(\gamma)}\cdot \frac{\Sigma w}{\norm{\Sigma^{1/2} w}} - \bar{x}} &\leq \frac{\beta}{\sqrt{d}} &\text{(analogous to \eqref{eqn:unified:samplemean})} \\
        \Norm{M_2 - \widehat{M}_2}_F &\leq \beta &\text{(analogous to \eqref{eqn:unified:M2})}\\
        \sup_{\norm{v}=1} \Iprod{v^{\otimes 3}, \widehat{M}_3 - M_3} &\leq \frac{\beta}{\sqrt{d}} &\text{(\Cref{lem:tensor-spectral-concentration})}
    \end{align*}
    Note that we used that for $\alpha \geq 1-c$, we have that $\kappa_2(\gamma) \geq \kappa_2(\Phi^{-1}(1-c))$ can be lower bounded by a constant by \Cref{claim:monotonicity-of-kappa2}.

    \textbf{Covariance estimation.}
    For covariance estimation, we want to argue that $M_2$ is close to $\Sigma$, i.e. that the second central moment would be good estimate.
    Then, we will also argue that for $\alpha$ close to $1$, we are essentially outputting $\widehat{M}_2$.
    
    First, note that, by the formula for $M_2$ of \Cref{lem:first3moments}, we have
    \[
        \Norm{M_2-\Sigma}_F \leq (1-\kappa_2(\gamma)) \cdot \frac{\norm{\Sigma w w^\top \Sigma}_F}{\norm{\Sigma^{1/2} w}^2} \leq O\Paren{(1-\alpha)\log\Paren{\frac{1}{1-\alpha}}},
    \]
    where we bounded $1-\kappa_2(\gamma)$ by \Cref{claim:one-minus-k2} and used $\norm{\Sigma} \leq 2/\kappa_2(\Phi^{-1}(1-c)) \leq O(1)$ for $\alpha \geq 1-c$ by \Cref{claim:monotonicity-of-kappa2} and \Cref{lem:preconditioner}.
    Combined with $\norm{M_2 - \widehat{M}_2}_F \leq \beta$ from above, we furthermore get that
    \[
        \Norm{\widehat{M}_2-\Sigma}_F \leq \beta + O\Paren{(1-\alpha)\log\Paren{\frac{1}{1-\alpha}}}.
    \]
    Recall that the estimator of \Cref{algo:main} is
    \[
        \widehat{\Sigma}=\widehat{M}_2+\Paren{\frac{1}{\kappa_2(\hat{\gamma})}-1}\cdot\Paren{\hat{u}^\top\widehat{M}_2^{-1}\hat{u}}^{-1}\hat{u}\hat{u}^\top.
    \]
    It remains to argue that the correction term $\Paren{\frac{1}{\kappa_2(\hat{\gamma})}-1}\cdot\Paren{\hat{u}^\top\widehat{M}_2^{-1}\hat{u}}^{-1}\hat{u}\hat{u}^\top$ is small.
    We have
    \begin{equation}\label{EQ:alphacloseto1:boundonuhatm2invu}
        c_1 \leq \lambda_{\min}\Paren{\widehat{M}_2^{-1}} \leq \hat{u}^\top \widehat{M}_2^{-1} \hat{u} \leq \lambda_{\max}\Paren{\widehat{M}_2^{-1}} \leq c_2.
    \end{equation}
    for appropriate constants $c_1,c_2 > 0$, by \Cref{lem:preconditioner}, as long as $\beta$ is at most a sufficiently small constant (such that we have ${2c_1 \cdot M_2^{-1} \preceq\widehat{M}_2^{-1} \preceq \frac{c_2}{2} \cdot M_2^{-1}}$).
    
    Thus, it remains to lower bound $\kappa_2(\hat{\gamma})$ in terms of $1-\alpha$.
    For this, we want to show that $\hat{\gamma}$ has to be large and use \Cref{claim:one-minus-k2}.
    Since $\sk$ is monotonic by \Cref{lemma:monotonicity-of-skewness}, if suffices to show that $\sk(\hat{\gamma})$ is large or equivalently that  $\abs{\sk(\hat{\gamma})}$ is small (since $\sk$ is always negative by \Cref{claim:bound-on-kappa2,claim:monotonicity-of-kappa2}).
    We have that, using \eqref{EQ:alphacloseto1:boundonuhatm2invu},
    \[
        \Abs{\sk(\hat{\gamma})} = \Paren{\hat{u}^\top\widehat{M}_2^{-1}\hat{u}}^{3/2}\cdot\Abs{\Iprod{\hat{u}^{\otimes 3},\widehat{M}_3}} \lesssim \Abs{\Iprod{\hat{u}^{\otimes 3},\widehat{M}_3}}.
    \]
    Furthermore, we can bound the right hand side of the above as follows
    \[
        \Abs{\Iprod{\hat{u}^{\otimes 3},\widehat{M}_3}} \leq \Abs{\Iprod{\hat{u}^{\otimes 3},\widehat{M}_3-M_3}} + \Abs{\Iprod{\hat{u}^{\otimes 3},M_3}} \leq \frac{\beta}{\sqrt{d}}+\norm{u^{\ast}}^3 \leq \frac{\beta}{\sqrt{d}}+O(\abs{\sk(\gamma)}),
    \]
    where we used the guarantee on $\widehat{M}_3-M_3$, $M_3 = {u^{\ast}}^{\otimes 3}$, as well as \Cref{claim:upper-bound-on-u-ast} to bound $\norm{u^\ast}$.
    Combined with the above, this gives
     \[
        \Abs{\sk(\hat{\gamma})} \leq O\Paren{\max\Set{\frac{\beta}{\sqrt{d}},\abs{\sk(\gamma)}}}.
    \]
    By choosing $\beta$ to be smaller than a sufficiently small constant and $c$ sufficiently small, we can make the right hand side arbitrarily small and get that $\Abs{\sk(\hat{\gamma})} \leq \abs{\sk(2)}$, which implies $\hat{\gamma} \geq 2$ (since $\sk$ is always negative and increasing, $\abs{\sk}$ is decreasing).
    Using \Cref{claim:skew-asymptotic} for both $\gamma$ and $\hat{\gamma}$, we get that
    \begin{equation}\label{EQ:alphaclosetoone:relationhatalphatoalpha}
        (1-\hat{\alpha}) \cdot \log^{3/2}\Paren{\frac{1}{1-\hat{\alpha}}} \leq O\Paren{\max\Set{\frac{\beta}{\sqrt{d}},(1-\alpha) \cdot \log^{3/2}\Paren{\frac{1}{1-\alpha}}}}
    \end{equation}
    for $\hat{\alpha} \coloneqq \Phi(\hat{\gamma})$.
    Since both $\hat{\alpha}$ and $\alpha$ are close to 1, if the maximum is the second term, then we also have\footnote{\label{footnote:alphacloseto1}To formally prove this, define $a = \log\Paren{\frac{1}{1-\alpha}}$, $\hat{a} = \log\Paren{\frac{1}{1-\hat{\alpha}}}$ and $r = \frac{a}{\hat{a}}$. We have $\frac{\exp(-\hat{a})\hat{a}^{3/2}}{\exp(-a)a^{3/2}} \leq O(1)$ and we want to show $\frac{\exp(-\hat{a})\hat{a}}{\exp(-a)a} \leq O(1)$, for which it is sufficient to show $\frac{a^{1/2}}{\hat{a}^{1/2}} \leq O(1)$.
    If $r \leq 1$, then this is clear. Otherwise, since $\hat{\gamma} \geq 2$, we have $\hat{a} \geq \log\Paren{\frac{1}{1-\Phi(2)}} \geq 1$. Thus, by assumption, we have
    \[
        \exp(r-1) \cdot r^{-3/2} \leq \exp((r-1)\hat{a}) \cdot r^{-3/2} = \frac{\exp(-\hat{a})\hat{a}^{3/2}}{\exp(-a)a^{3/2}} \leq O(1).
    \]
    Now, we have $\lim_{r \to \infty} \exp((r-1)) \cdot r^{-3/2} = \infty$ and thus $\exp((r-1)) \cdot r^{-3/2} \leq O(1)$ can only be true if we also have $r \leq O(1)$.}
    \[
        (1-\hat{\alpha}) \cdot \log\Paren{\frac{1}{1-\hat{\alpha}}} \lesssim (1-\alpha) \cdot \log\Paren{\frac{1}{1-\alpha}}.
    \]
    If the maximum is the first term, using that $\log^{1/2}\Paren{\frac{1}{1-\hat{\alpha}}} \geq 1$, we get that
    \[
        (1-\hat{\alpha}) \cdot \log\Paren{\frac{1}{1-\hat{\alpha}}} \lesssim \frac{\beta}{\sqrt{d}}.
    \]
    Combining these two, we thus get
    \[
        (1-\hat{\alpha}) \cdot \log\Paren{\frac{1}{1-\hat{\alpha}}} \leq O\Paren{\max\Set{\frac{\beta}{\sqrt{d}},(1-\alpha) \cdot \log\Paren{\frac{1}{1-\alpha}}}}.
    \]
    Thus, we have that, using \Cref{claim:one-minus-k2} in the second inequality,
    \begin{align*}
        \frac{1}{\kappa_2(\hat{\gamma})} -1 &\leq \frac{1}{\kappa_2(2)} \cdot (1 - \kappa_2(\hat{\gamma}))\\
        &\leq O\Paren{(1-\hat{\alpha}) \cdot \log\Paren{\frac{1}{1-\hat{\alpha}}}}\\
        &\leq O\Paren{\max\Set{\frac{\beta}{\sqrt{d}},(1-\alpha) \cdot \log\Paren{\frac{1}{1-\alpha}}}}.
    \end{align*}
    This implies, noting that $\norm{\hat{u}\hat{u}^\top}_F = \norm{\hat{u}}^2 = 1$,
    \[
        \Norm{\widehat{\Sigma} - \widehat{M}_2}_F =\Norm{\Paren{\frac{1}{\kappa_2(\hat{\gamma})}-1}\cdot\Paren{\hat{u}^\top\widehat{M}_2^{-1}\hat{u}}^{-1}\hat{u}\hat{u}^\top}_F \leq O\Paren{\frac{\beta}{\sqrt{d}}+(1-\alpha) \cdot \log\Paren{\frac{1}{1-\alpha}}}.
    \]
    Combined with the bound on $\norm{\widehat{M}_2 - \Sigma}_F$ we get that the estimator from \Cref{algo:main} satisfies
    \[
        \Norm{\widehat{\Sigma} - \Sigma}_F \leq O\Paren{\beta+(1-\alpha) \cdot \log\Paren{\frac{1}{1-\alpha}}}
    \]
    Finally, using \Cref{lem:preconditioner}, we get that
    \[
        \Norm{\Sigma^{-1/2}\widehat{\Sigma}\Sigma^{-1/2} - I_d}_F \leq \Norm{\Sigma^{-1/2}}^2 \cdot \Norm{\widehat{\Sigma} - \Sigma}_F \leq O\Paren{\beta+(1-\alpha) \cdot \log\Paren{\frac{1}{1-\alpha}}}.
    \]

    \textbf{Mean estimation.}
    For mean estimation, we want to similarly argue that the sample mean would be a good estimator and that we are outputting something close to the sample mean.
    First, recall that we have that the sample mean is close to the true mean, i.e., 
    \[
        \Norm{\mu + \frac{\phi(\gamma)}{\Phi(\gamma)}\cdot \frac{\Sigma w}{\norm{\Sigma^{1/2} w}} - \bar{x}} \leq \beta.
    \]
    By picking $c$ sufficiently small, we can achieve $\gamma \geq 2$.
    Then, we furthermore get that
    \[
        \Norm{\frac{\phi(\gamma)}{\Phi(\gamma)}\cdot\frac{\Sigma w}{\norm{\Sigma^{1/2}w}}} \leq \frac{1}{\Phi(2)} \cdot \Norm{\Sigma^{1/2}}^2 \cdot \phi(\gamma) \leq O(\phi(\gamma)),
    \]
    where we again used \Cref{lem:preconditioner} and the fact that $\kappa_2(\gamma) \geq \kappa_2(2)$ to bound $\norm{\Sigma^{1/2}}$.
    Using \Cref{FACT:rational-appx-of-imr} and \Cref{claim:log-alpha-and-gamma}, we can bound
    \[
        \phi(\gamma) \leq \Paren{\gamma + \frac{1}{\gamma}} \cdot (1-\Phi(\gamma)) \leq 2 \gamma \cdot (1-\Phi(\gamma)) \leq 2 \sqrt{2} \cdot (1-\Phi(\gamma)) \cdot \log^{1/2}\Paren{\frac{1}{1-\Phi(\gamma)}}.
    \]
    Thus, we get that
    \[
        \Norm{\frac{\phi(\gamma)}{\Phi(\gamma)}\cdot\frac{\Sigma w}{\norm{\Sigma^{1/2}w}}} \leq (1-\alpha) \cdot \log^{1/2}\Paren{\frac{1}{1-\Phi(\gamma)}}
    \]
    and hence
    \[
        \Norm{\mu - \bar{x}} \leq \beta + (1-\alpha) \cdot \log^{1/2}\Paren{\frac{1}{1-\Phi(\gamma)}}.
    \]
    Recall that the mean estimator of \Cref{algo:main} is
    \[
        \hat{\mu}=\bar{x}+\frac{\kappa_1(\hat{\gamma})}{\sqrt{\kappa_2(\hat{\gamma})}}\paren{\hat{u}^\top\widehat{M}_2^{-1}\hat{u}}^{-1/2}\hat{u}.
    \]
    It remains to argue that the correction term $\frac{\kappa_1(\hat{\gamma})}{\sqrt{\kappa_2(\hat{\gamma})}}\paren{\hat{u}^\top\widehat{M}_2^{-1}\hat{u}}^{-1/2}\hat{u}$ is small.
    We can bound this, using \eqref{EQ:alphacloseto1:boundonuhatm2invu}, $\hat{\gamma} \geq 2$ and \Cref{claim:monotonicity-of-kappa2}, as follows
    \[
        \Norm{\frac{\kappa_1(\hat{\gamma})}{\sqrt{\kappa_2(\hat{\gamma})}}\paren{\hat{u}^\top\widehat{M}_2^{-1}\hat{u}}^{-1/2}\hat{u}} \lesssim \frac{1}{\sqrt{\kappa_2(2)}} \cdot \abs{\kappa_1(\hat{\gamma})} \lesssim  \abs{\kappa_1(\hat{\gamma})} = \frac{\phi(\hat{\gamma})}{\Phi(\hat{\gamma})} \leq \frac{\phi(\hat{\gamma})}{\Phi(2)} \lesssim \phi(\hat{\gamma}).
    \]
    As above for $\gamma$, using \Cref{FACT:rational-appx-of-imr} and \Cref{claim:log-alpha-and-gamma}, we can bound
    \[
        \phi(\hat{\gamma}) \leq 2 \sqrt{2} \cdot (1-\Phi(\hat{\gamma})) \cdot \log^{1/2}\Paren{\frac{1}{1-\Phi(\hat{\gamma})}} = 2 \sqrt{2} \cdot (1-\hat{\alpha}) \cdot \log^{1/2}\Paren{\frac{1}{1-\hat{\alpha}}}.
    \]
    Using the same argument as in \Cref{footnote:alphacloseto1}, \eqref{EQ:alphaclosetoone:relationhatalphatoalpha} also implies
    \[
        (1-\hat{\alpha}) \cdot \log^{1/2}\Paren{\frac{1}{1-\hat{\alpha}}} \leq O\Paren{\max\Set{\frac{\beta}{\sqrt{d}},(1-\alpha) \cdot \log^{1/2}\Paren{\frac{1}{1-\alpha}}}}.
    \]
    Thus, it follows that
    \[
        \Norm{\hat{\mu} - \bar{x}} \leq O\Paren{\frac{\beta}{\sqrt{d}}+(1-\alpha) \cdot \log^{1/2}\Paren{\frac{1}{1-\alpha}}}.
    \]
    Combining this with the bound on $\norm{\mu - \bar{x}}$ and using \Cref{lem:preconditioner} again to bound $\norm{\Sigma^{-1/2}}$, we get that
    \[
        \Norm{\Sigma^{-1/2}(\hat{\mu}-\mu)} \leq O\Paren{\beta + (1-\alpha) \cdot \log^{1/2}\Paren{\frac{1}{1-\alpha}}}.
    \]

    \textbf{Conclusion.}
    So far, we have shown that for $\beta$ smaller than a sufficiently small constant $c_3$, we get that, using $O(d^2/\beta^2)$ samples, the estimators in \Cref{algo:main} satisfy
    \[
        \TVdist\Paren{\cN(\hat{\mu},\widehat{\Sigma}),\cN(\mu,\Sigma)}\leq O\Paren{\beta + (1-\alpha)\log\Paren{\frac{1}{1-\alpha}}}.
    \]
    Note that without loss of generality we can assume that $\varepsilon \leq 1$, otherwise any output is good (the TV-distance is at most 1 regardless of the output).
    Finally, picking $\beta = \min\{c_3, \varepsilon\}$ proves the lemma.
    The claimed sample complexity follows since this $\beta$ satisfies $\beta \geq \Omega(\varepsilon)$ (since $\varepsilon \leq 1$).
\end{proof}

Using \Cref{lem:unified,lem:unifiedalphacloseto1}, we can now prove \Cref{thm:main}.
\begin{proof}[Proof of \Cref{thm:main}]
    Correctness follows from \Cref{lem:unified,lem:unifiedalphacloseto1} if the number of samples used by the algorithm is large enough.
    To determine the sample complexity, we distinguish the three cases (i) $\alpha \leq \alpha_1$, (ii) $\alpha \in (\alpha_1, \alpha_2)$ and (iii) $\alpha \geq \alpha_2$, where $\alpha_1 = 1-c$ and $\alpha_2 = \max\Set{1-c,1-\frac{\varepsilon}{\log(1/\varepsilon)}}$ for $c$ as in \Cref{lem:unifiedalphacloseto1}.
    Again, we note that without loss of generality we can assume $\varepsilon \leq 1$, otherwise the guarantees hold trivially.

    \textbf{Case (i): Sample complexity for $\alpha \leq \alpha_1$.}
    In this case, we want to use \Cref{lem:unified}.
    By \Cref{claim:var-est,claim:skew-est} we get that
    \[
        \kappa_2(\gamma)\geq \Omega(\log^{-1}(1/\alpha)) \quad\text{and}\quad \abs{\kappa_3(\gamma)}\geq\Omega(\log^{-3/2}(1/\alpha)).
    \]
    Thus, the sample complexity in this case is
    \[
        n = O\Paren{\frac{d^2}{\varepsilon^2}\cdot\log^{14}(1/\alpha)}.
    \]
    
    \textbf{Case (ii): Sample complexity for $\alpha \in (\alpha_1, \alpha_2)$.}
    In this case, we again use \Cref{lem:unified}. We can bound $\kappa_2(\gamma)\geq \Omega(1)$ by \Cref{claim:var-est}.
    For $\kappa_3$, we have that $1-\alpha \geq 1-\alpha_2 \geq c'\frac{\varepsilon}{\log(1/\varepsilon)}$.
    Thus, using \Cref{claim:skew-est}, we have that
    \[
        \abs{\kappa_3(\gamma)} \geq \Omega\Paren{(1-\alpha) \cdot \log^{3/2}\Paren{\frac{1}{1-\alpha}}} \geq \Omega (\varepsilon).
    \]
    Thus, the sample complexity in this case (again from \Cref{lem:unified}) is
    \[
        n = O\Paren{\frac{d^2}{\varepsilon^4}}.
    \]

    \textbf{Case (iii): Sample complexity for $\alpha \geq \alpha_2$.}
    In this case, we claim that we can use \Cref{lem:unifiedalphacloseto1}.
    The guarantee from this lemma is that the TV-distance is at most
    \[
        O\Paren{\varepsilon + (1-\alpha)\log\Paren{\frac{1}{1-\alpha}}}
    \]
    and we need to show that this is at most $\varepsilon$.
    Using that $x \mapsto (1-x) \cdot \log(1/(1-x))$ is decreasing for $x$ close to 1, we get that the TV-distance is at most
    \[
        O\Paren{\varepsilon + \frac{\varepsilon}{\log(1/\varepsilon)}\log\Paren{\frac{\log(1/\varepsilon)}{\varepsilon}}} \leq O(\varepsilon).
    \]
    The sample complexity by \Cref{lem:unifiedalphacloseto1} is $n = O(d^2/\varepsilon^2)$.

    These three cases together prove correctness of \Cref{algo:main} given the sample complexity claimed in \Cref{thm:main}.\footnote{The above argument shows that the TV-distance is $O(\varepsilon)$. We can make the TV-distance $\varepsilon$ instead of $O(\varepsilon)$ by doing the exact same argument for $\varepsilon' = \frac{\varepsilon}{O(1)}$.\\Also note that we can simplify the sample complexity to be $O((d^2/\varepsilon^2) \cdot\polylog(1/\alpha))$ for $\alpha \leq 1-c$ (or in fact for $\alpha$ being at most any constant smaller than $1$) and $O(n^2/\varepsilon^4)$ for the remaining two cases. Thus, the proof above does indeed show correctness of the algorithm for the sample complexity as in \Cref{thm:main}.}

    \textbf{Runtime.}
    Finally, it remains to prove the runtime of the algorithm. We do so analogous to the arguments presented in \Cref{sec:techniques}.
    The runtime to compute the sample mean is $O(n \cdot d)$. The runtime to compute the sample covariance is $O(T(n,d))$.
    If we first do the contraction and then average over the samples, then also computing the random contraction $R$ of the third sample moment can be computed in time $T(n,d)$.
    Computing a spectral decomposition of this matrix can then be done in time $O(d^3)$, e.g. by computing a singular value decomposition.
    Computing the sign of $\iprod {u^{\otimes3},\widehat{M}_3}$ can be done in time $O(n \cdot d)$ since
    \[
        \iprod {u^{\otimes3},\widehat{M}_3} = \sum_{i=1}^n \iprod{u,x_i-\bar{x}}^3.
    \]
    Finally, computing the estimates can be done in time $O(d^2)$.
    Since $n \geq \Omega(d^2)$, we have that ${T(n,d) \geq \Omega(n \cdot d) \geq \Omega(d^3)}$ and thus the dominating term in the runtime is $O(T(n,d))$, which completes the proof.
\end{proof}

\begin{remark}\label{rem:approximateinverse}
    In \Cref{lem:unified,lem:unifiedalphacloseto1} above we have been assuming that the $\sk^{-1}(\cdot)$ function can be evaluated exactly and efficiently.
    In practice, we would instead only find an approximate inverse using binary search.
    We need to ensure that we get an estimate $\hat{\gamma}$ such that
    \[
        \Abs{\sk(\hat{\gamma})-\sk(\gamma)} \leq \varepsilon \cdot \poly(\min\{O(1), 1/\log(1/\alpha)\}).
    \]
    We can restrict the search to $\hat{\gamma} \geq \Omega(\sqrt{\log 1/\alpha})$ (the lower bound on $\gamma$).
    Furthermore, if ${\gamma \geq \widetilde{O}(\sqrt{\log(1/\varepsilon)})}$ (this implies $\alpha \geq 1-\frac{\epsilon}{\log 1/\epsilon}$), then we showed in \Cref{lem:unifiedalphacloseto1} and the proof of \Cref{thm:main} that $\hat{\gamma}=\infty$ is a good enough estimate and thus we can restrict the binary search to the interval $[\Omega(\sqrt{\log 1/\alpha}), \widetilde{O}(\sqrt{\log(1/\varepsilon)})]$.
    The derivative of $\sk$ is upper bounded by a constant and thus it is sufficient to get an estimate of the inverse up to 
    \[
        \varepsilon \cdot \poly(\min\{O(1), 1/\log(1/\alpha)\}).
    \]
    Hence, we need to do binary search on the interval $[\Omega(\sqrt{\log 1/\alpha}), \widetilde{O}(\sqrt{\log(1/\varepsilon)})]$ up to the above accuracy, which can be done in time
    \[
        \log(1/\varepsilon) \cdot \poly(\log\log(1/\varepsilon), \log\log(1/\alpha)).
    \]
    We remark that this time is negligible with respect to the other computation.
\end{remark}

%% file: content/robustness.tex
\section{Robustness}\label{sec:robustness}
In this section, we show how to extend our result for halfspaces to the robust setting. We will present the results for the non-trivial truncation regime of $\alpha$ bounded away from $1$. Throughout this section as well we will assume $\Sigma \succ 0$.

\subsection{Preliminaries on robustness}
We consider the following definition of robustness analogous to the well-studied strong contamination model in algorithmic statistics literature. We denote a distribution $\cD$ truncated to the set $S$ by $\cD|_S$.
\begin{definition}[$\eta$-corruption of $\cD|_S$] Samples $X_1, \ldots, X_n$ are called an \emph{$\eta$-corruption of $\cD|_S$} if they are obtained as follows:
    \begin{enumerate}
        \item First, we draw $Y_1, \ldots, Y_n \sim \cD|_S$.
        \item Then, an adversary picks $X_1, \ldots, X_n \in \R^d$ such that $|\{i \in [n] : X_i \neq Y_i\}| \leq \eta  n$.
    \end{enumerate}
\end{definition}
Observe that the adversary can remove all information about the truncation set.

In particular the adversary is allowed to also pick $X_i$ such that $X_i \in \R^{d} \backslash S$.

We require the following definition of certifiably sub-Gaussian distributions, see e.g. \cite[Definition 1.1]{KothariSteinhardtSteurer:robustmomentestiamtion} or \cite[Definition 1.3]{DiakonikolasHopkinsPensiaTiegel:subgaussianiscertifiable}.

\begin{definition}[Certifiably sub-Gaussian distribution]
    A distribution $\cD$ with mean $\mu$ is called $s$-certifiably sub-Gaussian if for all even $m \in \N$, there is a sum-of-squares proof of the following inequality in indeterminate $v$
    \[
        \bbE_{X \sim \cD}\left[\Iprod{X-\mu, v}^m\right] \leq (Cs\sqrt{m})^m \Norm{v}^m,
    \]
    where $C$ is a universal constant.
\end{definition}
We also require the following two theorems that we will utilize to obtain robust estimates of the moments.
\begin{theorem}[{\cite[Theorems 2.2 and 2.3]{KothariSteinhardtSteurer:robustmomentestiamtion}}]\label{thm:KSSMain}
    Let $k \geq 6$ be an even constant.
    Let $C > 0$.
    There is a polynomial-time algorithm 
    with the following guarantees:
    Let $\cD$ be a $C$-certifiably sub-Gaussian distribution on $\R^d$ with mean $\mu$ and covariance $\Sigma$. Let $M_{3,\mathrm{nc}}$ be the raw third moment of $\cD$.
    Let $\eta > 0$ such that $C k \cdot \eta^{1-2/k}$ is at most a sufficiently small constant. 
    Then given as input an $\eta$-corrupted sample of $\cD$ of size $n \geq (C+d)^{O(k)}$, the algorithm outputs estimates of the moments of $\cD$ with the following guarantees:
    \begin{align*}
        &\Norm{\mu - \hat{\mu}} \leq O(Ck)^{1/2} \cdot \eta^{1-1/k} \cdot \Norm{\Sigma}^{1/2}\\
        &(1-O(Ck)^{1/2} \cdot \eta^{1-2/k}) \Sigma \preceq \widehat{\Sigma} \preceq (1+O(Ck)^{1/2} \cdot \eta^{1-2/k}) \Sigma\\
        &\forall v \in \R^d: \Iprod{M_{3,\mathrm{nc}} - \widehat{M}_{3,\mathrm{nc}}, v^{\otimes 3}}^2 \leq O(Ck)^{3/2} \cdot \eta^{1-3/k} \Iprod{M_{2,\mathrm{nc}}, v^{\otimes 2}}^{3}
    \end{align*}
    where $M_{2, \mathrm{nc}}$ corresponds to the raw second moment of $\cD$.
\end{theorem}

\begin{theorem}[{\cite[Theorem 1.6]{DiakonikolasHopkinsPensiaTiegel:subgaussianiscertifiable}}]\label{thm:DHPT25SoS}
    There exists a universal constant $C' > 0$ such that any $s$-sub-Gaussian distribution is $C's$-certifiably sub-Gaussian.    
\end{theorem}
\begin{fact}\label{fact:certifabletruncD}
The distribution $\cD \coloneqq\truncatedGauss{\mu}{\Sigma}{\{w^\top x \leq \tau\}}$ is $O(\vnorm{\Sigma}^{1/2})$--certifiably sub-Gaussian. 
\end{fact}
We note that this follows directly by combining \Cref{cor:subgaussian:X} and \Cref{thm:DHPT25SoS}.
As a direct consequence, we can apply \Cref{thm:KSSMain} for estimating the truncated moments, even under \emph{adversarial} corruption.

\subsection{Learning parameters under halfspace truncation and robustness}
As a first step we begin by obtaining a robust preconditioner for our setup. 
\begin{lemma}\label{lemma:robustprecond}
    Let $\eta > 0$ be at most a sufficiently small constant.
    Given an $\eta$-corruption $X_1, \ldots, X_n$ of $\truncatedGauss{\mu}{\Sigma}{\{w^\top x \leq \tau\}}$, we can efficiently compute a matrix $\widetilde{M}_2$ such that 
    \[
        (1 - \zeta) M_2 \preceq \widetilde{M}_2 \preceq (1 + \zeta) M_2
    \]
    whenever  $\vnorm{\Sigma}^{1/2}k \cdot \eta^{1 - 2/k} \leq O(1)$ and $\zeta = \vnorm{\Sigma}^{1/4}k^{1/2} \eta^{1 - 2/k}$ where $k$ is as in the above theorem of \cite{KothariSteinhardtSteurer:robustmomentestiamtion}.
\end{lemma}
The proof of \Cref{lemma:robustprecond} follows directly by an application of \Cref{thm:KSSMain}. Once we have access to such an $\widetilde{M}_2$, if $\zeta \leq \frac{\min\Set{\kappa_2(\gamma),1}}{2}$, then as discussed in \Cref{claim:preconditioner} in \Cref{sec:precondition}, we can assume that 
\begin{align*}
     &0.5 I_d \preceq M_2 \preceq 2 I_d \\
     &0.5 I_d \preceq \Sigma \preceq \frac{2}{\kappa_2(\gamma)} I_d
\end{align*}
After transforming the samples we apply \Cref{thm:KSSMain} (again) to obtain robust estimates for the  second and third moment. 
\begin{lemma}\label{LEM:robustmomentestimationforourcase}
    Let $\delta > 0$ be an arbitrary constant.
    Let $0 < \eta \leq O(1)$ for a sufficiently small constant.
    Given an $\eta$-corruption $X_1, \ldots, X_n$ of $\truncatedGauss{\mu}{\Sigma}{\{w^\top x \leq \tau\}}$ (preconditioned as discussed above), where $n = \poly(d)$, we can efficiently compute estimates for the first non-central moment $\hat{\mu}_{\mathrm{trunc}}$ as well as the second and third (central) moments $\widehat{M}_2$ and $\widehat{M}_3$ of $\truncatedGauss{\mu}{\Sigma}{\{w^\top x \leq \tau\}}$ with the following guarantees
    \begin{align*}
        &\Norm{\mu_{\mathrm{trunc}} - \hat{\mu}_{\mathrm{trunc}}} \leq O_\delta\left(\log^{1/4}(1/\alpha) \cdot \eta^{1-\delta}\right)\\
        &\left(1-O_\delta\left(\log^{1/4}(1/\alpha) \cdot \eta^{1-\delta}\right)\right) M_2 \preceq \widehat{M}_2 \preceq \left(1 + O_\delta\left(\log^{1/4}(1/\alpha) \cdot \eta^{1-\delta}\right)\right) M_2\\
        &\forall v \in \cS^{d-1}: \Iprod{M_3 - \widehat{M}_3, v^{\otimes 3}} \leq O_\delta\Paren{\max\Paren{\log(1/\alpha)^{1/4},\log(1/\alpha)^{7/8}} \cdot \eta^{1/2-\delta}}
    \end{align*} 
    whenever $\eta^{1 - \delta} \leq O_\delta(\log^{-1/2} (1/\alpha))$.
\end{lemma}
\begin{proof}
    We observe that for the estimate for $\mu_{\mathrm{trunc}}$ and $M_2$ it suffices to apply \Cref{thm:KSSMain} (where $C$ is now at most $O(\kappa_2(\gamma)^{-1/2})$) and pick $k$ such that $\delta \geq 3/k$ for some constant $k$ to obtain the result.
    Note that, by \Cref{claim:var-est}, we have that $\kappa_2(\gamma)^{-1} \leq O\left(\log(1/\alpha)\right)$.

    For the estimate for the third central moments $M_3$, we cannot directly apply \Cref{thm:KSSMain} since they only give an estimate for the non-central moments.
    Instead, we again use the estimate $\hat{\mu}_\mathrm{trunc}$ for the truncated mean $\mu_{\mathrm{trunc}}$.
    We note that this satisfies, as discussed before,
    \[
        \vnorm{\hat{\mu}_\mathrm{trunc}-\mu_\mathrm{trunc}} \leq O_\delta\Paren{\log^{1/4}(1/\alpha) \cdot \eta^{1-\delta}}.
    \]
    Then, using new samples $X_i$ (such that they are independent of $\hat{\mu}_\mathrm{trunc}$), we apply the algorithm from \Cref{thm:KSSMain} to the samples ${X_i' = X_i - \hat{\mu}_\mathrm{trunc}}$ to get an estimate $M'_{3,\mathrm{nc}}$ of the third non-central moment of the uncorrupted $Y_i' = Y_i - \hat{\mu}_\mathrm{trunc}$\footnote{Recall that $Y$ corresponds to uncorrupted samples and $X$ corresponds to corrupted samples.}.
    We use this estimator $\widehat{M}'_{3,\mathrm{nc}}$ to estimate the third central moment $M_3$ of the uncorrupted $Y_i$.
    We have
    \begin{align*}
        M_3 &= \mathbb{E}[(Y_i-\mu_\mathrm{trunc})^{\otimes 3}]\\
        &= \mathbb{E}[(Y_i-\hat{\mu}_\mathrm{trunc}+\hat{\mu}_\mathrm{trunc}-\mu_\mathrm{trunc})^{\otimes 3}]\\
        &= \mathbb{E}[(Y_i-\hat{\mu}_\mathrm{trunc})^{\otimes 3}+ \mathrm{Sym}((Y_i-\hat{\mu}_\mathrm{trunc})^{\otimes 2} \otimes (\hat{\mu}_\mathrm{trunc}-\mu_\mathrm{trunc}))\\
        &\mathrel{\phantom{=}} \mathrel{+} \mathrm{Sym}((Y_i-\hat{\mu}_\mathrm{trunc}) \otimes (\hat{\mu}_\mathrm{trunc}-\mu_\mathrm{trunc})^{\otimes 2}) + (\hat{\mu}_\mathrm{trunc}-\mu_\mathrm{trunc})^{\otimes 3}],
    \end{align*}
    where we use $\mathrm{Sym}$ to denote all permutations (i.e., $\mathrm{Sym}(a^{\otimes 2} \otimes b) = a \otimes a \otimes b + a \otimes b \otimes a + b \otimes a \otimes a$).
    Now, \Cref{thm:KSSMain} gives an estimate to $M'_{3,\mathrm{nc}} \coloneqq \mathbb{E}[(Y_i-\hat{\mu}_\mathrm{trunc})^{\otimes 3}]$ with the following guarantee
    \begin{equation}\label{EQ:robustness:approxforcorrectednoncentral}
        \Iprod{M'_{3,\mathrm{nc}} - \widehat{M}'_{3,\mathrm{nc}}, v^{\otimes 3}} \leq O_\delta\Paren{\log(1/\alpha)^{3/8} \cdot \eta^{1/2-\delta} \cdot \Norm{\mathbb{E}[(Y_i-\hat{\mu}_\mathrm{trunc})^{\otimes 2}]}^{3/2}}
    \end{equation}
    for any unit vector $v$.
    We can bound
    \begin{align}
        \Norm{\mathbb{E}[(Y_i-\hat{\mu}_\mathrm{trunc})^{\otimes 2}]} &= \Norm{\mathbb{E}[(Y_i-\mu_\mathrm{trunc})^{\otimes 2}] + (\mu_\mathrm{trunc}-\hat{\mu}_\mathrm{trunc})^{\otimes 2}}\nonumber\\
        &\leq \Norm{\mathbb{E}[(Y_i-\mu_\mathrm{trunc})^{\otimes 2}]} + \Norm{(\mu_\mathrm{trunc}-\hat{\mu}_\mathrm{trunc})^{\otimes 2}}\nonumber\\
        &\leq 2 + \Norm{\mu_\mathrm{trunc}-\hat{\mu}_\mathrm{trunc}}^2,\label{EQ:robustness:boundonsecondmoment}
    \end{align}
    where we used in the last step that the first term is just the second central moment of the (truncated) $Y_i$, which satisfies $M_2 \preceq 2 I_d$.
    We now get the following for the error of our estimate, where $v$ is again a unit vector,
    \[
        \Iprod{M_3 - \widehat{M}'_{3,\mathrm{nc}}, v^{\otimes 3}} = \Iprod{M_3 - M'_{3,\mathrm{nc}}, v^{\otimes 3}} + \Iprod{M'_{3,\mathrm{nc}} - \widehat{M}'_{3,\mathrm{nc}}, v^{\otimes 3}}.
    \]
    We can bound the second term by \eqref{EQ:robustness:approxforcorrectednoncentral}.
    For the first term, note that
    \begin{align*}
        \Iprod{M_3 - M'_{3,\mathrm{nc}}, v^{\otimes 3}} &= \Iprod{\mathrm{Sym}(\mathbb{E}[(Y_i-\hat{\mu}_\mathrm{trunc})^{\otimes 2}] \otimes (\hat{\mu}_\mathrm{trunc}-\mu_\mathrm{trunc})), v^{\otimes 3}}\\
        &\mathrel{\phantom{=}} \mathrel{+} \Iprod{\mathrm{Sym}(\mathbb{E}[Y_i-\hat{\mu}_\mathrm{trunc}] \otimes (\hat{\mu}_\mathrm{trunc}-\mu_\mathrm{trunc})^{\otimes 2}), v^{\otimes 3}}\\
        &\mathrel{\phantom{=}} \mathrel{+} \Iprod{(\hat{\mu}_\mathrm{trunc}-\mu_\mathrm{trunc})^{\otimes 3}, v^{\otimes 3}}\\
        &\leq 4 \vnorm{\hat{\mu}_\mathrm{trunc}-\mu_\mathrm{trunc}}^3 + 3\vnorm{\hat{\mu}_\mathrm{trunc}-\mu_\mathrm{trunc}} \cdot \Norm{\mathbb{E}[(Y_i-\hat{\mu}_\mathrm{trunc})^{\otimes 2}]}\\
        &\leq 7 \vnorm{\hat{\mu}_\mathrm{trunc}-\mu_\mathrm{trunc}}^3 + 6\vnorm{\hat{\mu}_\mathrm{trunc}-\mu_\mathrm{trunc}},
    \end{align*}
    where in the last step we used \eqref{EQ:robustness:boundonsecondmoment}.
    Combining this with \eqref{EQ:robustness:approxforcorrectednoncentral}, we get
    \begin{align*}
        \Iprod{M_3 - \widehat{M}'_{3,\mathrm{nc}}, v^{\otimes 3}} &\leq O_\delta\Paren{\log(1/\alpha)^{3/8} \cdot \eta^{1/2-\delta} \cdot (2 + \Norm{\mu_\mathrm{trunc}-\hat{\mu}_\mathrm{trunc}}^2)}\\
        &\mathrel{\phantom{=}} \mathrel{+} 7 \vnorm{\hat{\mu}_\mathrm{trunc}-\mu_\mathrm{trunc}}^3 + 6\vnorm{\hat{\mu}_\mathrm{trunc}-\mu_\mathrm{trunc}}\\
        &\leq O_\delta\Paren{\max\Paren{\log(1/\alpha)^{1/4},\log(1/\alpha)^{7/8}} \cdot \eta^{1/2-\delta}},
    \end{align*}
    where we used $\eta \leq 1$ for the last step.
\end{proof}
We now prove the following theorem.
\subsection*{Main robustness result}

\begin{theorem}[Full statement of \Cref{THM:INTRO:robust}]\label{THM:robust}
    Let $0 < \delta < \frac{1}{2}$ be an arbitrary constant and let $\alpha \leq 0.99$.
    Define
    \[
        L \coloneqq
        \max\Paren{\log^{1/4}(1/\alpha),\log^{7/8}(1/\alpha)}
    \]
    and
    \[
        \beta_{\mathrm{robust}}
        \coloneqq
        \max\Set{1,\frac{1}{|\kappa_3(\gamma)|}}
        \cdot \sqrt d \cdot L\cdot \eta^{1/2-\delta}.
    \]
        Let $c_\delta$ be a sufficiently small constant depending only on $\delta$.     Assume that \begin{equation}\label{EQ:robustnessconditiononeta}
        0 < L^{\frac{2}{1 -2\delta}} \cdot \eta
        \leq
        O_\delta\Paren{
            \Paren{\frac{\poly\bigl(\kappa_2(\gamma), \kappa_3(\gamma)\bigr)}{\sqrt{d}}}^{\frac{2}{1-2\delta}}
        }.
    \end{equation}
    The polynomial in \eqref{EQ:robustnessconditiononeta} is chosen sufficiently small so that
    \[
        L\eta^{1-\delta}\leq O_\delta(\beta_{\mathrm{robust}}),
        \qquad
        \beta_{\mathrm{robust}}
        \leq
        c_\delta \cdot \Paren{\min\Set{\kappa_2(\gamma)^{5/2},|\kappa_3(\gamma)|}},
    \]
     and so that \Cref{LEM:robustmomentestimationforourcase} applies.
    Furthermore, assume that
    \begin{equation}\label{EQ:robustnessconditiononeta2}
        \eta^{1-\delta}
        \leq
        \min\Set{
            O_\delta(\log^{-1/2} (1/\alpha)),   c_\delta \cdot 
            \Norm{\Sigma}^{-1/2},   c_\delta \cdot 
            \Norm{\Sigma}^{-1/4} \cdot \kappa_2(\gamma)
        }.
    \end{equation}
    Given an $\eta$-corruption $X_1, \ldots, X_n$ of
    $\truncatedGauss{\mu}{\Sigma}{\{w^\top x \leq \tau\}}$, where $n = \poly(d)$,
    with high probability, we can, in time $\poly(d)$, compute estimators
    $\hat{\mu}$ and $\widehat{\Sigma}$ such that
    \[
        \Norm{\Sigma^{-1/2}(\hat{\mu} - \mu)}
        \leq
        O_\delta\Paren{\frac{1}{\kappa_2(\gamma)^{7/2}}\cdot \beta_{\mathrm{robust}}}
    \]
    and
    \[
        \norm{\Sigma^{-1/2}\widehat{\Sigma}\Sigma^{-1/2}-I_d}_F
        \leq
        O_\delta\Paren{\frac{1}{\kappa_2(\gamma)^{7/2}}\cdot \beta_{\mathrm{robust}}}.
    \]
\end{theorem}

\begin{proof}
    We follow the proof of \Cref{thm:main} but replace the sample moments with our robust estimates from \Cref{LEM:robustmomentestimationforourcase}.
    Note that we can do the preconditioning as discussed above by \eqref{EQ:robustnessconditiononeta2}.
    Also note that we can apply \Cref{LEM:robustmomentestimationforourcase} by \eqref{EQ:robustnessconditiononeta2}. By the robust preconditioning step and affine invariance, it suffices to prove the theorem in the preconditioned coordinates, where
\[
    \frac12 I_d\preceq M_2\preceq 2I_d,
    \qquad
    \frac12 I_d\preceq \Sigma\preceq \frac{2}{\kappa_2(\gamma)}I_d.
\]

    \subsection*{Covariance estimation.}
    In order to apply \Cref{lem:putting-together} for the covariance estimation, we need to compute $\hat{u}$ and $\hat{\gamma}$ and bound the following four quantities (\eqref{eqn:25-1}, \eqref{eqn:25-2}, \eqref{eqn:25-5} and \eqref{eqn:25-3} respectively):
    \begin{enumerate}[noitemsep]
        \item $\displaystyle \norm{\Sigma^{-1/2}\Paren{\widehat{M}_2-M_2}\Sigma^{-1/2}}_F$,
        \item $\displaystyle \min\left\{\Norm{\hat{u}-\frac{u^{\ast}}{\norm{u^{\ast}}}},\Norm{\hat{u}+\frac{u^{\ast}}{\norm{u^{\ast}}}}\right\}$,
        \item $\abs{\sk(\hat{\gamma})-\sk(\gamma)}$ ,
        \item $\displaystyle \left\lvert \hat{u}^T\widehat{M}_2^{-1} \hat{u}-\frac{\norm{\Sigma^{1/2}w}^2}{\norm{\Sigma w}^2\cdot\kappa_2(\gamma)}\right\rvert$.
    \end{enumerate}
    We remark that showing 3. is enough to show \eqref{eqn:25-5} by the proof of \Cref{lem:threshold-estimation:closeinskew}.
    
    \noindent \textbf{Bounding 1.}
    From \Cref{LEM:robustmomentestimationforourcase}, we immediately get
    \begin{align}\label{EQ:robustnessestimatorforM2}
    \norm{\Sigma^{-1/2}\Paren{\widehat{M}_2-M_2}\Sigma^{-1/2}}_F &\leq \sqrt{d} \cdot \norm{\Sigma^{-1/2}\Paren{\widehat{M}_2-M_2}\Sigma^{-1/2}} \\
        &\leq \sqrt{d}\cdot O_\delta({\eta}^{1-\delta} \cdot L).
    \end{align}
    The first inequality holds because for any matrix $A$ we have that the Frobenius norm satisfies $\Vert A \Vert_F \leq \sqrt{\textup{rank}(A)} \cdot \Vert A \Vert$ and the last inequality follows from \Cref{LEM:robustmomentestimationforourcase}. 
    As shown in \cite{KothariSteinhardtSteurer:robustmomentestiamtion} the polynomial dependency on the dimension is \emph{necessary} in general for certifiably sub-Gaussian distributions.\\
    \noindent \textbf{Bounding 2.}
    It suffices that we get a good estimate of $M_3$. Indeed, we are promised such a good estimate from \Cref{LEM:robustmomentestimationforourcase}. This directly implies that we can find a unit vector $\hat{u} \in \R^d$ such that
    \begin{equation}\label{EQ:robustnessestimatorforustar}
        \min\Set{\Norm{\hat{u}+\frac{u^{\ast}}{\norm{u^{\ast}}}}_2,\Norm{\hat{u}-\frac{u^{\ast}}{\norm{u^{\ast}}}}_2}\leq O_\delta\left(\frac{\sqrt{d}}{|\kappa_3(\gamma)|} \cdot L \cdot \eta^{1/2-\delta} \right) 
    \end{equation}

    \noindent \textbf{Bounding 3.}
    From  \eqref{EQ:robustnessconditiononeta} we have that
    \begin{align*}
        \max\Set{O_\delta\Paren{\frac{\sqrt{d}}{|\kappa_3(\gamma)|} \cdot L \cdot \eta^{1/2 - \delta}}, O_\delta(\eta^{1-\delta} \cdot L),  O_\delta(\eta^{1/2-\delta} \cdot L)} &= O_\delta\Paren{\frac{\sqrt{d}}{|\kappa_3(\gamma)|}\cdot L \cdot \eta^{1/2- \delta}}\\
        &\leq c,
    \end{align*}
    where $c$ is a sufficiently small constant\footnote{Note that this is \emph{weaker} than the condition that our theorem statement imposes on $\eta$.}.
    This enables us to use \Cref{lem:threshold-estimation:closeinskew} and compute an estimate $\hat{\gamma}$ with $\beta$ in that lemma set to $O_\delta\Paren{\frac{\sqrt{d}}{|\kappa_3(\gamma)|}\cdot L \cdot \eta^{1/2- \delta}}$. Therefore we have
    \begin{equation}\label{EQ:robustnessestimatorfortauhat}
        \Abs{\sk(\hat{\gamma}) - \sk(\gamma)} \leq O_\delta\Paren{\frac{\sqrt{d}}{|\kappa_3(\gamma)|}\cdot L \cdot \eta^{1/2- \delta}} \cdot \kappa_2(\gamma)^{-3/2}. 
    \end{equation}
    
    \noindent \textbf{Bounding 4.}
The second conclusion of \Cref{lem:threshold-estimation:closeinskew}, applied with
\[
    \beta
    =
 O_\delta\Paren{\frac{\sqrt{d}}{|\kappa_3(\gamma)|}\cdot L \cdot \eta^{1/2- \delta}},
\]
gives
\[
    \left\lvert
    \hat u^\top \widehat M_2^{-1}\hat u
    -
    \tilde u^\top M_2^{-1}\tilde u
    \right\rvert
    \le O_\delta\Paren{\frac{\sqrt{d}}{|\kappa_3(\gamma)|}\cdot L \cdot \eta^{1/2- \delta}}.
\]
Using \eqref{eqn:secondmomentwithsigmaw}, where $\tilde{u} = u^\ast / \Vert u^\ast\Vert - \Sigma w/ \Vert \Sigma w \Vert$
\[
    \tilde u^\top M_2^{-1}\tilde u
    =
    \frac{\|\Sigma^{1/2}w\|^2}
    {\kappa_2(\gamma)\|\Sigma w\|^2}
    =
    \frac{\kappa_3(\gamma)^{2/3}}
    {\|u^\ast\|^2\kappa_2(\gamma)}.
\]
Therefore
    \begin{equation}\label{EQ:robustnessestimatorforM2inverse}
        \left\lvert \hat{u}^T\widehat{M}_2^{-1} \hat{u}-\frac{\norm{\Sigma^{1/2}w}^2}{\norm{\Sigma w}^2\cdot\kappa_2(\gamma)}\right\rvert 
        \leq O_\delta\Paren{\frac{\sqrt{d}}{|\kappa_3(\gamma)|} \cdot L \cdot \eta^{1/2 - \delta} }
    \end{equation}

    \noindent \textbf{Conclusion for covariance estimation.}
    Combining \eqref{EQ:robustnessestimatorforM2}, \eqref{EQ:robustnessestimatorforustar},  \eqref{EQ:robustnessestimatorfortauhat} and \eqref{EQ:robustnessestimatorforM2inverse} we can apply \Cref{lem:putting-together} with
    \[
        \beta = O_\delta\Paren{\max\Set{1,\frac{1}{|\kappa_3(\gamma)|} } \cdot \sqrt{d} \cdot L \cdot \eta^{1/2-\delta}}
    \]
    Note that by \eqref{EQ:robustnessconditiononeta} we have $\beta \lesssim \kappa_2(\gamma)^{5/2}$.
    This gives us an estimator $\widehat{\Sigma}$ with error
    \[
        \norm{\Sigma^{-1/2}\widehat{\Sigma}\Sigma^{-1/2}-I_d}_F
        \leq O_\delta\Paren{\frac{1}{\kappa_2(\gamma)^{7/2}} \cdot \max\Set{1,\frac{1}{|\kappa_3(\gamma)|} } \cdot \sqrt{d} \cdot L \cdot \eta^{1/2-\delta}}.
    \]
    
    \subsection*{Mean estimation.}
    For mean estimation, given \Cref{lem:putting-together-2}, we only need to argue that we can additionally bound (\eqref{eqn:39-1} and \eqref{eqn:39-3})
    \begin{enumerate}[noitemsep]
        \item[5.] $\displaystyle \Norm{\bar{x} -\mu-\kappa_1(\gamma)\frac{\Sigma w}{\norm{\Sigma^{1/2}w}}}$
        \item[6.] $\displaystyle \Norm{\hat{u}-\frac{u^{\ast}}{\norm{u^{\ast}}}}$
    \end{enumerate}

    \noindent \textbf{Bounding 5.}
    By \Cref{LEM:robustmomentestimationforourcase}, we are able to get an estimate $\hat{\mu}_\mathrm{trunc}$ for the mean of the truncated distribution ${\mu_\mathrm{trunc} = \mu + \kappa_1(\gamma)\frac{\Sigma w}{\norm{\Sigma^{1/2}w}}}$. This estimate satisfies
    \begin{equation}\label{EQ:robustnessmeanestimation}
        \Norm{\mu_{\mathrm{trunc}} - \hat{\mu}_{\mathrm{trunc}}} \leq O_\delta\left(L \cdot \eta^{1-\delta}\right).
    \end{equation}
    Thus we can take $\bar{x} = \hat{\mu}_\mathrm{trunc} $ and we are done.

\noindent \textbf{Bounding 6.}
Given that we already bounded 2., we only need to determine the sign of $\frac{u^{\ast}}{\norm{u^{\ast}}}$.
We can do so using the third moment tensor.
Let
\[
    \tilde{u} \coloneqq \frac{u^\ast}{\norm{u^\ast}}.
\]
Since $M_3=(u^\ast)^{\otimes 3}$, we have
\[
    \Iprod{M_3,\tilde{u}^{\otimes 3}}
    =
    \Iprod{(u^\ast)^{\otimes 3},\tilde{u}^{\otimes 3}}
    =
    \Iprod{u^\ast,\tilde{u}}^3
    =
    \norm{u^\ast}^3
    >0.
\]
Let $a\in\{\pm1\}$ be the sign such that
\[
    \Norm{a\hat{u}-\tilde{u}}\leq O_\delta\Paren{\frac{\sqrt d}{|\kappa_3(\gamma)|}\cdot L\cdot \eta^{1/2-\delta}}.
\]
from \eqref{EQ:robustnessestimatorforustar}.
We can compute
\begin{align*}
    \Iprod{\widehat{M}_3, (a\hat{u})^{\otimes 3}}
    &=
    \Iprod{\widehat{M}_3-M_3,(a\hat{u})^{\otimes 3}}
    +
    \Iprod{M_3,(a\hat{u})^{\otimes 3}}\\
    &=
    \Iprod{\widehat{M}_3-M_3,(a\hat{u})^{\otimes 3}}
    +
    \Iprod{M_3,(a\hat{u})^{\otimes 3}-\tilde{u}^{\otimes 3}}
    +
    \Iprod{M_3,\tilde{u}^{\otimes 3}}.
\end{align*}
We bound the first term by \Cref{LEM:robustmomentestimationforourcase}:
\[
    \Abs{\Iprod{\widehat{M}_3-M_3,(a\hat{u})^{\otimes 3}}}
    \leq
    O_\delta\Paren{L \cdot \eta^{1/2-\delta}}.
\]
For the second term, using $M_3=(u^\ast)^{\otimes 3}$ and the bound on $\Norm{a\hat u-\tilde u}$, we get
\begin{align*}
    \Abs{\Iprod{M_3,(a\hat{u})^{\otimes 3}-\tilde{u}^{\otimes 3}}}
    &=
    \Abs{\Iprod{u^\ast,a\hat u}^3-\Iprod{u^\ast,\tilde u}^3}\\
    &\leq
    O_\delta\Paren{\norm{u^\ast}^3 \Paren{\frac{\sqrt d}{|\kappa_3(\gamma)|}\cdot L\cdot \eta^{1/2-\delta}}},
\end{align*}
where we used that $a\hat u$ and $\tilde u$ are unit vectors and $O_\delta\Paren{\frac{\sqrt d}{|\kappa_3(\gamma)|}\cdot L\cdot \eta^{1/2-\delta}} \leq O(1)$.
Therefore,
\begin{align*}
    \Iprod{\widehat{M}_3,(a\hat{u})^{\otimes 3}}
    &\geq
    \norm{u^\ast}^3
    -
    O_\delta\Paren{L\eta^{1/2-\delta}}
    -
    O\Paren{\norm{u^\ast}^3(\sqrt d/|\kappa_3(\gamma)|) L \eta^{1/2-\delta}}.
\end{align*}
Hence, since by \eqref{EQ:robustnessconditiononeta} and
\[
\Vert u^\ast \Vert^3 \gtrsim | \kappa_3(\gamma)|
\] in the precondition coordinates, we have 
\[
    O_\delta\Paren{L\eta^{1/2-\delta}}
    +
   O_\delta\Paren{\norm{u^\ast}^3 \Paren{\frac{\sqrt d}{|\kappa_3(\gamma)|}\cdot L\cdot \eta^{1/2-\delta}}}
    <
    \frac{1}{2}\norm{u^\ast}^3,
\]
and we get
\[
    \Iprod{\widehat{M}_3,(a\hat{u})^{\otimes 3}}>0.
\]
Since the tensor is order three,
\[
    \Iprod{\widehat{M}_3,(-a\hat{u})^{\otimes 3}}
    =
    -\Iprod{\widehat{M}_3,(a\hat{u})^{\otimes 3}}
    <0.
\]
Thus, we can determine the correct sign for $\hat u$ by choosing the sign for which
\[
    \Iprod{\widehat{M}_3,\hat u^{\otimes 3}} > 0.
\]
Recall that we need to find a vector parallel to $-\Sigma w$ and $u^*$ is a vector that points in the direction $-\Sigma w$.
After this sign choice, we have
\[
    \Norm{\hat u + \frac{\Sigma w}{\norm{\Sigma w}}}
    \leq
    O_\delta\Paren{\frac{\sqrt d}{|\kappa_3(\gamma)|}\cdot L\cdot \eta^{1/2-\delta}}.
\]

    \noindent \textbf{Conclusion for mean estimation.}
    Combining \eqref{EQ:robustnessmeanestimation}, \eqref{EQ:robustnessestimatorforM2}, \eqref{EQ:robustnessestimatorforustar}, \eqref{EQ:robustnessestimatorforM2inverse} and \eqref{EQ:robustnessestimatorfortauhat}, we can apply \Cref{lem:putting-together-2}
    \[
        \beta = O_\delta\Paren{\max\Set{1,\frac{1}{|\kappa_3(\gamma)|}} \cdot \sqrt{d} \cdot L \cdot \eta^{1/2-\delta}}.
    \]
    Note that by \eqref{EQ:robustnessconditiononeta} we have $\beta \lesssim \kappa_2(\gamma)^{5/2}$.
    This gives us an estimator $\hat{\mu}$ with error
    \begin{align*}
        \Norm{\Sigma^{-1/2}(\hat{\mu} - \mu)} &\leq O(\Norm{\hat{\mu} - \mu})\\
        &\leq O_\delta\Paren{\frac{1}{\kappa_2(\gamma)^{7/2}}\cdot \max\Set{1,\frac{1}{|\kappa_3(\gamma)|}} \cdot \sqrt{d} \cdot L \cdot \eta^{1/2-\delta}}.
    \end{align*}
\end{proof}

%% file: content/extensions.tex
\section{Extension to intersections of two orthogonal halfspaces}\label{sec:extensions}
In this section, we show how to extend our algorithm to truncations sets that are the intersection of two orthogonal halfspaces.
Suppose $d\geq 2$ and consider a random vector $\x$ that follows a $d$-dimensional Gaussian $\cN(\mu,I_d) $ truncated to the intersection of $H_1$ and $H_2$, where $H_1=\Set{x:\iprod{w_1,x}\leq\tau_1}$ and $H_2=\Set{x:\iprod{w_2,x}\leq\tau_2}$ are halfspaces with $\norm{w_1}=\norm{w_2}=1$ and $w_1\perp w_2$.
As before, we define $\alpha>0$ to be the probability measure of $H_1\cap H_2$ under the underlying Gaussian.
We will prove the following result.
For simplicity, we assume that the underlying Gaussian has identity covariance and we only aim to estimate the mean.
\begin{theorem}\label{thm:intersec-of2orthogonal-halfspaces}
    There is an algorithm (\Cref{algo:extensions}) that takes $n=O\Paren{\frac{d^2}{\epsilon^2}\poly\log\Paren{\frac{1}{\alpha}}+\frac{d^2}{\epsilon^4}}$ i.i.d. samples from $\cN(\mu,I_d)$ truncated to $H_1\cap H_2$ as input and efficiently computes an estimate $\hat{\mu}$, such that with probability at least 0.99 it satisfies
    \[
        \TVdist\Paren{\cN(\hat{\mu},I_d),\cN(\mu,I_d)}\leq \epsilon.
    \]
    The runtime of the algorithm is $O(T(n,d))$, where $T(n,d)$ is the time needed to multiply a $d \times n$ matrix with its transpose.
\end{theorem}

\begin{algorithm}[!ht]
\caption{Algorithm for estimating the mean of an identity-covariance Gaussian truncated by two orthogonal halfspaces}
\label{algo:extensions}
\KwIn{Dimension $d$, sample size $n$, samples $\{x_i\}_{i=1}^{n}$, target accuracy $\varepsilon$.}
\KwOut{Estimated mean $\hat{\mu} \in \mathbb{R}^d$.}

\BlankLine

1. Compute sample mean: $\bar{x} = \frac{1}{n}\sum_{i=1}^{n} x_i$.\;

2. Compute sample second central moment: $\widehat{M}_2 = \frac{1}{n}\sum_{i=1}^{n}(x_i - \bar{x})^{\otimes 2}$.\;

\BlankLine

3. Draw a standard Gaussian vector: $g \sim \mathcal{N}(0, I_d)$.\;

\BlankLine

4. Compute the random contraction: $R = (I_d\otimes I_d\otimes g^\top) \cdot \widehat{M}_3 = \frac{1}{n}\sum_{i=1}^{n}(x_i-\bar{x})^{\otimes 2}\cdot \iprod{g,x_i-\bar{x}}$.

\BlankLine

5. Compute $\hat{w}_1$ and $\hat{w}_2$, the eigenvectors of $R$ corresponding to the two largest eigenvalue in absolute value (with unit norm).

\BlankLine

6. For $\ell \in \{1,2\}$, flip the sign of $\hat{w}_\ell$ if $\Iprod{\hat{w}_\ell^{\otimes 3}, \widehat{M}_3}>0$.

7. For $\ell \in \{1,2\}$, estimate the relative truncation parameters $\gamma_\ell$: $\hat{\gamma_\ell}=\kappa_2^{-1}\Paren{\hat{w}_\ell^\top \widehat{M}_2 \hat{w}_\ell}$.

\BlankLine

8. Estimate the mean of the Gaussian: $\hat{\mu}=\bar{x} - \kappa_1(\hat{\gamma}_1) \cdot \hat{w}_1 - \kappa_1(\hat{\gamma}_2) \cdot \hat{w}_2$.

\BlankLine

9. Return estimate $\hat{\mu}$.

\end{algorithm}

\begin{remark}\label{rem:furtherextensions}
Before proving \Cref{thm:intersec-of2orthogonal-halfspaces}, we want to briefly discuss potential further extensions:
\begin{enumerate}[noitemsep]
    \item While we only prove our extension for mean estimation, it seems plausible that one can also do covariance estimation. In this case, we would want that $w_1$ and $w_2$ are `$\Sigma$-orthogonal', meaning that $\Sigma w_1$ and $\Sigma w_2$ are orthogonal. One thing that is unclear however is whether and how we could do preconditioning in this case to avoid a dependence on the condition number (as preconditioning need not preserve orthogonality).
    \item Finally, the correction terms in each direction are (almost) independent of each other due to orthogonality. This should enable us to generalize our estimator to an intersection of $k$ pairwise orthogonal halfspaces (for $k \leq d$). The required sample and time complexity in this case would depend on the minimum eigenvalue gap after the random contraction, which asymptotically could be as small as $\widetilde{O}\Paren{\frac{1}{k^2}}$ and the sample complexity should increase by a factor that is polynomial in $k$. If there is non-trivial truncation in each direction, we would thus expect a sample complexity of $O(d^2/\varepsilon^2) \cdot \polylog(1/\alpha) \cdot \poly(k)$.
\end{enumerate}
\end{remark}

\subsection{Moments}
As for our main result, we first compute the first three moments of $\cN(\mu,I_d)|_{H_1 \cap H_2}$.
The values $\gamma_1,\gamma_2$ will be defined analogously to the \truncparaword{} in the main result.
\begin{lemma}\label{LEM:extensions:moments}
    For a random vector $\x$ defined as above, there are real numbers $\gamma_1,\gamma_2$ such that
    \begin{align*}
        \E\x&=\mu+\kappa_1(\gamma_1) \cdot w_1+\kappa_1(\gamma_2) \cdot w_2\\
        M_2\coloneqq\E(\x-\E\x)^{\otimes 2}&=I_d-(1-\kappa_2(\gamma_1)) \cdot w_1w_1^\top-(1-\kappa_2(\gamma_2)) \cdot w_2w_2^\top\\
        M_3\coloneqq\E(\x-\E\x)^{\otimes 3}&=\kappa_3(\gamma_1) \cdot \Paren{w_1}^{\otimes 3}+\kappa_3(\gamma_2) \cdot \Paren{w_2}^{\otimes 3}
    \end{align*}
    Moreover, we have
    \[
        \alpha=\Phi(\gamma_1)\cdot\Phi(\gamma_2).
    \]
\end{lemma}
\begin{proof}
    For $d\geq 2$, since $w_1,w_2$ are perpendicular unit vectors, there exists an orthonormal matrix $U$ such that $U e_1=w_1$ and $U e_2=w_2$.
    As before, we consider $\y=U^\top\paren{\x-\mu}$.
    This random variable follows a standard Gaussian $\cN(0,I_n)$ truncated as follows:
    \begin{align*}
        w_1^\top X \leq \tau_1 &\quad\Longleftrightarrow\quad (U^\top w_1)^\top Y \leq \tau_1 - w_1^\top \mu\\
        w_2^\top X \leq \tau_2 &\quad\Longleftrightarrow\quad (U^\top w_2)^\top Y \leq \tau_2 - w_2^\top \mu.
    \end{align*}
    Since $U e_1=w_1$ and $U e_2=w_2$, we get a simple characterization of $\y$: Defining $\gamma_1=\tau_1-\iprod{w_1,\mu}$ and $\gamma_2=\tau_2-\iprod{w_2,\mu}$, $\y_1$ and $\y_2$ are truncated standard Gaussian with thresholds $\gamma_1$ and $\gamma_2$, all other coordinates are standard Gaussian and all coordinates are mutually independent.
    Thus, as in \Cref{lem:first3moments}, we get that
    \begin{align*}
        \E\y&=\kappa_1(\gamma_1)e_1+\kappa_1(\gamma_2)e_2,\\
        \E(\y-\E\y)^{\otimes 2}&=I_d-(1-\kappa_2(\gamma_1))e_1e_1^\top-(1-\kappa_2(\gamma_2))e_2e_2^\top,\\
        \E(\y-\E\y)^{\otimes 3}&=\kappa_3(\gamma_1)e_1^{\otimes 3}+\kappa_3(\gamma_2)e_2^{\otimes 3}.
    \end{align*}
    Using this, we can compute the moments of $\x$ as follows.
    First, the mean of $\x$ is
    \[
        \E\x=U\E \y+\mu =\mu+\kappa_1(\gamma_1)\cdot Ue_1+\kappa_1(\gamma_2)\cdot Ue_2 =\mu+\kappa_1(\gamma_1) \cdot w_1+\kappa_1(\gamma_2)\cdot w_2.
    \]
    The covariance of $\x$ is
    \begin{align*}
        M_2=\E(\x-\E\x)^{\otimes 2}&=U^{\otimes 2}\E(\y-\E\y)^{\otimes 2}\\
        &=U\Paren{I_d-(1-\kappa_2(\gamma_1))\cdot e_1e_1^\top-(1-\kappa_2(\gamma_2)) \cdot e_2e_2^\top}U^\top\\
        &=I_d-(1-\kappa_2(\gamma_1))\cdot w_1w_1^\top-(1-\kappa_2(\gamma_2))\cdot w_2w_2^\top.
    \end{align*}
    And the third central moment of $\x$ equals
    \begin{align*}
        M_3=\E(\x-\E\x)^{\otimes 3}&=U^{\otimes 3}\E(\y-\E\y)^{\otimes 3}\\
        &=U^{\otimes 3}\Paren{\kappa_3(\gamma_1)\cdot e_1^{\otimes 3}+\kappa_3(\gamma_2)\cdot e_2^{\otimes 3} }\\
        &=\kappa_3(\gamma_1)\cdot w_1^{\otimes 3}+\kappa_3(\gamma_2)\cdot w_2^{\otimes 3}.
    \end{align*}
    Finally, define $H_1'=\Set{y\in\R^d\mid y_1\leq\gamma_1}$ and $H_2'=\Set{y\in\R^d\mid y_2\leq\gamma_2}$.
    Then, we have that
    \begin{align*}
        \alpha&=\E_{X\sim\normal{\mu,I_d}}\Brac{\ind{X\in H_1\cap H_2}}\\
        &=\E_{Y\sim\normal{0,I_d}}\Brac{\ind{Y\in H_1'\cap H_2'}}\\
        &=\Pr_{Y_1 \sim \normal{0,1}}\Brac{Y_1 \leq\gamma_1}\cdot\Pr_{Y_2 \sim \normal{0,1}}\Brac{Y_2 \leq\gamma_2}\\
        &=\Phi(\gamma_1)\cdot\Phi(\gamma_2),
    \end{align*}
    which completes the proof.    
\end{proof}

Since we assume that the covariance of the underlying Gaussian is the identity, we don't need to precondition the samples.
Instead, we directly get that
\begin{equation}\label{eqn:extension:preconditioning}
    \min\{\kappa_2(\gamma_1), \kappa_2(\gamma_2)\} \preceq M_2\preceq I_d.
\end{equation}

\subsection{Relative truncation parameter estimation}
By orthogonality between $w_1,w_2$, we could isolate them from the second central moment as long as $\gamma_1,\gamma_2$ are \textit{constantly-separated}. However, this is not true in general. It might be that $\gamma_1=\gamma_2$, which makes it information-theoretically impossible to identify $w_1,w_2$ from $M_2$ already when~$d=2$.
Therefore, as in the case of a single halfspace, we use the third central moment $M_3$.
Note that with probability $1$, we have that $\kappa_3(\gamma_1)\cdot \iprod{g,w_1}$ and $\kappa_3(\gamma_2)\cdot \iprod{g,w_2}$ are distinct.
Furthermore, we will be able to bound the distance between these eigenvalues using \Cref{fact:gaussianspread}.
We will use the following theorem from matrix perturbation theory (see also \cite[Corollary III.2.6 and Theorem VII.3.2]{Bhatia:matrixanalysis}).
\begin{theorem}[\cite{davis-kahan1970,Weyl:eigenvalueperturbation}]\label{thm:davis-kahan}
    Let $A, E \in \R^{d \times d}$ be symmetric matrices. Consider the eigenvalues ${\lambda_1 \geq \lambda_2 \geq \ldots \geq \lambda_d}$ of $A$ with associated orthonormal eigenvectors $v_i$ and the eigenvalues $\tilde{\lambda}_1 \geq \tilde{\lambda}_2 \geq \ldots \tilde{\lambda}_d$ and orthonormal eigenvectors $\tilde{v}_i$ of $\tilde{A} = A + E$. Define $\delta_i = \min_{j \neq i} \abs{\lambda_j - \lambda_i}$. Then, if $\delta_i > 0$, we have
    \[
        \min\Set{\norm{\tilde{v}_i+v_i}_2,\norm{\tilde{v}_i-v_i}_2}\leq C\frac{\Norm{E}}{\delta_i}
    \]
    for an absolute constant $C > 0$.
    Furthermore, we have that
    \[
        \abs{\lambda_i - \tilde{\lambda}_i} \leq \Norm{E}.
    \]
\end{theorem}

We now show the following lemma, which is analogous to \Cref{SEC:fullproof:directionofsigmaw}.
We note that after step~6 of \Cref{algo:extensions}, the estimates $\hat{w}_1$ and $\hat{w}_2$ are exactly the ones from \Cref{lem:learning-orthogonal-halfspaces}.
\begin{lemma}\label{lem:learning-orthogonal-halfspaces}
    Let $\beta>0$.
    Given $n$ i.i.d. samples from $\cN\Paren{\mu,\Sigma}$ truncated to $H_1\cap H_2$, we can compute orthogonal unit vectors $\hat{w}_1,\hat{w}_2$ such that the following hold for $\ell \in \{1,2\}$: If 
    \[
        n = O\Paren{\frac{d^2}{\beta^2 \kappa_3(\gamma_\ell)^2} + \frac{d^2}{\kappa_3(\gamma_\ell)^4}},
    \]
    then we have with high probability that $\Norm{\hat{w}_\ell - w_\ell} \leq \beta$.
\end{lemma}

\begin{proof}
    We want to use \Cref{thm:davis-kahan} for the matrix $A = (I_d\otimes I_d\otimes g^\top)M_3$ and $\tilde{A} = (I_d\otimes I_d\otimes g^\top)\widehat{M}_3$.
    Let the spectral decomposition of $(I_d\otimes I_d\otimes g^\top)\widehat{M}_3$ be
    \[
        (I_d\otimes I_d\otimes g^\top)\widehat{M}_3=\sum_{j=1}^{d}\tilde{\lambda}_j \tilde{v}_j\tilde{v}_j^\top,
    \]
    where $\tilde{\lambda}_1\geq\tilde{\lambda}_2\geq\cdots\geq\tilde{\lambda}_d$.
    Since this matrix is symmetric, we can assume that the $\tilde{v}_j$ are an orthonormal basis of $\R^d$.
    To apply \Cref{thm:davis-kahan}, we need to determine the eigenvalues and eigenvectors of $(I_d\otimes I_d\otimes g^\top)M_3$.
    We have, since $\Sigma = I_d$,
    \begin{align*}
        M_3&=\kappa_3(\gamma_1)w_1^{\otimes 3}+\kappa_3(\gamma_2)w_2^{\otimes 3},\\
        (I_d\otimes I_d\otimes g^\top)M_3&=\kappa_3(\gamma_1)\cdot \iprod{g, w_1}\cdot w_1 w_1^\top +\kappa_3(\gamma_2)\cdot \iprod{g, w_2}\cdot w_2 w_2^\top.
    \end{align*}
    Thus, the eigenvalues of $(I_d\otimes I_d\otimes g^\top)M_3$ are $\kappa_3(\gamma_1)\cdot \iprod{g, w_1}$, $\kappa_3(\gamma_2)\cdot \iprod{g, w_2}$ and $0$ with multiplicity $n-2$. Furthermore, the corresponding eigenvectors are $w_1$ and $w_2$ (and an orthogonal basis of the complement).
    Let $i_1, i_2 \in [d]$ be such that $v_{i_1} = w_1$ and $v_{i_2} = w_2$.
    We will define our estimates as the two eigenvectors of $\tilde{A}$ corresponding to the largest eigenvalues in absolute value (up to a sign choice that we discuss later).
    This already implies that they are orthogonal.
    
    If we want one of our two estimates to be a good estimate for $w_\ell$, we need to ensure that even after the perturbation, the eigenvalue corresponding to $w_\ell$ is one of the two largest in absolute value (i.e. it should be later than the perturbation of the $0$ eigenvalues).
    We can ensure this as follows: The eigenvalues of $A$ that are $0$ correspond to eigenvalues of $\tilde{A}$ of absolute value at most $\norm{A-\tilde{A}})$.
    Thus, if the eigenvalue for $w_\ell$ is larger than $2 \cdot \norm{A-\tilde{A}})$ in absolute value, even after perturbation it will belong to the largest two (hence, one of our estimates will be close to $w_\ell$ up to sign).
    Hence, if we ensure that $\norm{A-\tilde{A}} \leq \frac{1}{2} \cdot \Abs{\kappa_3(\gamma_\ell)\cdot \iprod{g, w_\ell}}$, it suffices to argue that one of $\|w_\ell \pm \tilde{v}_{i_\ell}\|$ is small and how to recover the sign.
    
    In order to apply \Cref{thm:davis-kahan}, we also need to compute the values of $\delta_{i_1}$ and $\delta_{i_2}$.
    For this, note that $\iprod{g,\kappa_3(\gamma_1) \cdot w_1-\kappa_3(\gamma_2) \cdot w_2}$ is distributed identical to the Gaussian $\cN(0,\kappa_3(\gamma_1)^2+\kappa_3(\gamma_2)^2)$.
    Thus, by Gaussian anti-concentration (cf. \Cref{fact:gaussianspread}), we get that, with high probability,
    \[
        \abs{\kappa_3(\gamma_1) \cdot \iprod{g,w_1}-\kappa_3(\gamma_2)\cdot \iprod{g,w_2}}\geq \Omega\Paren{\sqrt{\kappa_3(\gamma_1)^2+\kappa_3(\gamma_2)^2}}.
    \]
    Similarly, we get
    \[
        \abs{\kappa_3(\gamma_1) \cdot \iprod{g,w_1}} \geq \Omega\Paren{\abs{\kappa_3(\gamma_1)}} \quad\text{and}\quad \abs{\kappa_3(\gamma_2) \cdot \iprod{g,w_2}} \geq \Omega\Paren{\abs{\kappa_3(\gamma_2)}}.
    \]
    Thus, we can lower bound
    \[
        \delta_{i_1} \geq \Omega\Paren{\abs{\kappa_3(\gamma_1)}} \quad\text{and}\quad \delta_{i_2} \geq \Omega\Paren{\abs{\kappa_3(\gamma_2)}}.
    \]
    It remains to bound $\norm{(I_d\otimes I_d\otimes g^\top)\widehat{M}_3-(I_d\otimes I_d\otimes g^\top)M_3}$.
    First, analogous to \Cref{lem:tensor-spectral-concentration}, we get that, with high probability,
    \[
        \sup_{\norm{v}=1}\Iprod{v^{\otimes 3},\widehat{M}_3-M_3}\lesssim \sqrt{\frac{d}{n}}+\frac{d^{3/2}}{n}.
    \]
    Note that we get a slightly better result here since the samples $x_i$ are in fact $1$-sub-Gaussian in this case (because $\Sigma = I_d$) and not only $O(1/\sqrt{\kappa_2(\gamma)})$-sub-Gaussian as before.
    Now, observe that with high probability,
    \begin{align*}
        \Norm{(I_d\otimes I_d\otimes g^\top)(\widehat{M}_3-M_3)} &\leq \norm{g}\sup_{\norm{v}=1}\Iprod{v^{\otimes 3},\widehat{M}_3-M_3}\\
        &\leq 3\sqrt{d}\cdot\sup_{\norm{v}=1}\Iprod{v^{\otimes 3},\widehat{M}_3-M_3}\\
        &\lesssim \frac{d}{\sqrt{n}}+\frac{d^2}{n}.
    \end{align*}
    Applying \Cref{thm:davis-kahan}, we conclude that, for $\ell \in \{1,2\}$,
    \[
        \min_{b\in\{\pm 1\}}\Set{\norm{b\cdot \tilde{v}_{i_\ell}- w_\ell}} \lesssim \kappa_3(\gamma_\ell)^{-1} \cdot \Paren{\frac{d}{\sqrt{n}}+\frac{d^2}{n}}.
    \]
    
    Let $\ell \in \{1,2\}$. It remains to argue that with the claimed sample complexity, we can make this error at most $\beta$ and that we can recover the sign.
    Note that for the claimed sample complexity we indeed have, using the above bound (and assuming $\beta \leq O(1)$), that
    \[
        \norm{(I_d\otimes I_d\otimes g^\top)(\widehat{M}_3-M_3)} \leq \frac{1}{2} \cdot \Abs{\kappa_3(\gamma_\ell)\cdot \iprod{g, w_\ell}},
    \]
    so $\tilde{v}_{i_\ell}$ will correspond to one of the two largest eigenvalues of $\tilde{A}$ up to sign.
    If
    \[
        n = O\Paren{\frac{d^2}{\beta'^2 \cdot \kappa_3(\gamma_\ell)^2}},
    \]
    then we get that
    \[
        \min_{b\in\{\pm 1\}}\Set{\norm{b\cdot \tilde{v}_{i_\ell}- w_\ell}} \leq \beta'.
    \]
    To determine the sign, we again check the sign of $\iprod{\tilde{v}_{i_\ell}^{\otimes 3}, \widehat{M}_3}$.
    We have ${\iprod{w_\ell^{\otimes 3}, M_3} = \kappa_3(\gamma_\ell) < 0}$.
    Denote by $b_\ell \in \{\pm 1\}$ the correct sign for $\tilde{v}_{i_\ell}$ (i.e. the one for which $\|b_\ell \tilde{v}_{i_\ell} - w_\ell\| \leq \beta$).
    Then, we have that
    \[
        \Abs{\kappa_3(\gamma_\ell) - \Iprod{b_\ell \tilde{v}_{i_\ell}^{\otimes 3}, \widehat{M}_3}} \leq \Abs{\Iprod{(w_\ell- b_\ell \tilde{v}_{i_\ell})^{\otimes 3}, M_3}} + \Abs{\Iprod{b_\ell \tilde{v}_{i_\ell}^{\otimes 3}, M_3 - \widehat{M}_3}}.
    \]
    We can bound the first term using $M_3 = \kappa_3(\gamma_1) w_1^{\otimes 3} + \kappa_3(\gamma_2) w_2^{\otimes 3}$ as follows, assuming $\beta' \leq O(1)$ (such that we can bound $\beta'^3 \leq O(\beta')$),
    \[
        \Abs{\Iprod{(w_\ell- b_\ell \tilde{v}_{i_\ell})^{\otimes 3}, M_3}} \lesssim \Paren{\Abs{\kappa_3(\gamma_1)} + \Abs{\kappa_3(\gamma_2)}} \cdot \norm{w_\ell- b_\ell \tilde{v}_{i_\ell}}^3 \lesssim \beta'.
    \]
    For the second term, note that the sample complexity above implies\footnote{We could even get $\Abs{\kappa_3(\gamma_\ell)} \cdot \frac{\beta'}{\sqrt{d}}$, but the bounding this by $\beta'$ is sufficient.}
    \[
        \Abs{\Iprod{b_\ell \tilde{v}_{i_\ell}^{\otimes 3}, M_3 - \widehat{M}_3}} \lesssim \beta'.
    \]
    Thus, we get that
    \[
        \Abs{\kappa_3(\gamma_\ell) - \Iprod{b_\ell \tilde{v}_{i_\ell}^{\otimes 3}, \widehat{M}_3}} \lesssim \beta'.
    \]
    If $\beta' \leq c \cdot \kappa_3(\gamma_\ell)$ for a sufficiently small constant $c$, then $\iprod{b_\ell \tilde{v}_{i_\ell}^{\otimes 3}, \widehat{M}_3} < 0$ and we can determine the sign $b_\ell$ from $\iprod{\tilde{v}_{i_\ell}^{\otimes 3}, \widehat{M}_3}$.
    Summarized, for $\beta' \leq c \cdot \kappa_3(\gamma_\ell)$ (which also implies $\beta' \leq O(1)$ by \Cref{claim:kappa1kappa3bounded}), we can identify the sign $b_\ell$ such that $\norm{b_\ell \tilde{v}_{i_\ell} - w_\ell} \leq \beta'$ if we get
    \[
        n = O\Paren{\frac{d^2}{\beta'^2 \cdot \kappa_3(\gamma_\ell)^2}}
    \]
    samples.
    We can now pick $\beta' = \min\{\beta, c \cdot \kappa_3(\gamma_\ell)\}$ to get the result.
\end{proof}

\subsection{Proof of \texorpdfstring{\Cref{thm:intersec-of2orthogonal-halfspaces}}{extension result}}
With the learned normal vectors $\hat{w}_1$ and $\hat{w}_2$, we can compute $\hat{w}_\ell^\top\widehat{M}_2\hat{w}_\ell$, which yields an approximation to $\kappa_2(\gamma_\ell)$. By monotonicity of $\kappa_2$ (cf. \Cref{claim:monotonicity-of-kappa2}), we obtain estimation for the relative truncation parameters $\hat{\gamma}_1$ and $\hat{\gamma}_2$, which allows us to compute an estimator for $\mu$.
This allows us to prove \Cref{thm:intersec-of2orthogonal-halfspaces}.
As before, we need to be careful if $\alpha_\ell$ is close to 1 as then $\kappa_3(\gamma_\ell) \to 0$.

\begin{proof}[Proof of \Cref{thm:intersec-of2orthogonal-halfspaces}]
    Recall from \Cref{LEM:extensions:moments} that, for $X \sim \cN(\mu, I_d)|_{H_1 \cap H_2}$ we have
    \[
        \mathbb{E}[X] = \mu + \kappa_1(\gamma_1) \cdot w_1 + \kappa_1(\gamma_2) \cdot w_2.
    \]
    We thus get
    \begin{align*}
        \Norm{\hat{\mu}-\mu}&=\Norm{\bar{x}-\kappa_1(\hat{\gamma}_1)\hat{w}_1-\kappa_1(\hat{\gamma}_2)\hat{w}_2-\mu}\\
        &=\Norm{\bar{x}-\E[\x]+\kappa_1(\hat{\gamma}_1)\hat{w}_1+\kappa_1(\hat{\gamma}_2)\hat{w}_2-\kappa_1(\gamma_1)w_1-\kappa_1(\gamma_2)w_2}\\
        &\leq \Norm{\bar{x}-\E[\x]}+\Norm{\kappa_1(\hat{\gamma}_1)\hat{w}_1-\kappa_1(\gamma_1)w_1}+\Norm{\kappa_1(\hat{\gamma}_2)\hat{w}_2-\kappa_1(\gamma_2)w_2}.
    \end{align*}
    We have that $x_i$ is 1-sub-Gaussian and thus given $n = O(d/\beta^2)$ samples, as in \Cref{lem:unified}, we have $\Norm{\bar{x} - \E[X]} \leq \beta$.
    For the other two terms, we need to distinguish whether $\alpha_\ell$ is close to $1$ or not.
    Both terms are the same, so we only show how to bound $\norm{\kappa_1(\hat{\gamma}_1)\hat{w}_1-\kappa_1(\gamma_1)w_1}$.
    As before in \Cref{thm:main}, we distinguish two cases: If $\alpha_1 \leq \alpha_* \coloneqq \max\{1-c,\frac{1-\beta}{\log(1/\beta)}\}$ for a sufficiently small constant $c > 0$, then we use \Cref{lem:learning-orthogonal-halfspaces} to get an estimate $\|\hat{w}_\ell - w_\ell\| \leq \beta$ and use this to bound $\Norm{\kappa_1(\hat{\gamma}_1)\hat{w}_1-\kappa_1(\gamma_1)w_1}$. If $\alpha_1 \geq \alpha_*$, then we show that this norm is small by arguing that both $\kappa_1(\hat{\gamma}_1)$ and $\kappa_1(\gamma)$ are small.
    
    \textbf{Case (i): $\alpha_1 \leq \alpha_*$.}
    Using
    \[
        n = O\Paren{\frac{d^2}{\beta^2 \kappa_3(\gamma_1)^2} + \frac{d^2}{\kappa_3(\gamma_1)^4}}
    \]
    samples, the estimate $\hat{w}_1$ from \Cref{lem:learning-orthogonal-halfspaces} satisfies $\norm{\hat{w}_1-w_1} \leq \beta$.
    Let $\widehat{M}_2=\frac{1}{n}\sum_{i=1}^{n}\paren{x_i-\bar{x}}^2$ be the sample covariance.
    We have
    \begin{align*}
        \hat{w}_1^\top \widehat{M}_2 \hat{w}_1 &= w_1^\top \widehat{M}_2 w_1 + (\hat{w}_1 - w_1)^\top \widehat{M}_2 (\hat{w}_1 + w_1)\\
        &= w_1^\top M_2 w_1 + w_1^\top (\widehat{M}_2-M_2) w_1 + (\hat{w}_1 - w_1)^\top \widehat{M}_2 (\hat{w}_1 + w_1)\\
        &= \kappa_2(\gamma_1) + w_1^\top (\widehat{M}_2-M_2) w_1 + (\hat{w}_1 - w_1)^\top \widehat{M}_2 (\hat{w}_1 + w_1).
    \end{align*}
    By \Cref{lem:learning-orthogonal-halfspaces}, we have
    \[
      \Norm{(\hat{w}_1 - w_1)^\top \widehat{M}_2 (\hat{w}_1 + w_1)} \leq 2 \Norm{\widehat{M}_2} \beta.  
    \]
    Furthermore, analogous to \Cref{claim:subgaussian:M2invX}, we get that the $M_2^{-1}x_i$ are $O(1/\min\{\sqrt{\kappa_2(\gamma_1)},\sqrt{\kappa_2(\gamma_2)}\})$-sub-Gaussian.
    Thus, using
    \[
        n = O\Paren{\frac{d}{\beta^2 \cdot \min\{\kappa_2(\gamma_1), \kappa_2(\gamma_2)\}}}
    \]
    samples, we get by \Cref{fact:random-matrix-concentration} that $\Norm{M_2-\widehat{M}_2} \leq \beta$.
    Since $M_2 \preceq I_d$, this also implies that $\norm{\widehat{M}_2} \leq 2$, as long as $\beta \leq 1$.
    Putting these together, we get that
    \[
        \Abs{\hat{w}_1^\top \widehat{M}_2 \hat{w}_1 - \kappa_2(\gamma_1)} \leq O(\beta).
    \]
    Define $\hat{\gamma}_1 = \kappa_2^{-1}(\hat{w}_1^\top \widehat{M}_2 \hat{w}_1)$ (this inverse function exists since $\kappa_2$ is monotonic by \Cref{claim:monotonicity-of-kappa2}).
    By \Cref{CLAIM:ddtkappa2overddtkappa1}, we have
    \[
        \Abs{\kappa_1(\hat{\gamma}_1)-\kappa_1(\gamma_1)} \lesssim \Abs{\kappa_2(\hat{\gamma}_1)-\kappa_2(\gamma_1)} \cdot \max\{\kappa_2(\gamma_1)^{-2}, \kappa_2(\hat{\gamma}_1)^{-2}\}.
    \]
    Now assume $\beta<c' \cdot \kappa_2(\gamma_1)$ (for a sufficiently small constant $c'$).
    Then, we get $\kappa_2(\hat{\gamma}_1)\geq\frac{1}{2}\kappa_2(\gamma_1)$ and thus
    \[
        \Abs{\kappa_1(\hat{\gamma}_1)-\kappa_1(\gamma_1)} \lesssim \beta \cdot \kappa_2(\gamma_1)^{-2}.
    \]
    We thus have that
    \begin{align*}
        \Norm{\kappa_1(\hat{\gamma}_1)\hat{w}_1-\kappa_1(\gamma_1)w_1}&\leq\Norm{\kappa_1(\hat{\gamma}_1)\hat{w}_1-\kappa_1(\gamma_1)\hat{w}_1}+\Norm{\kappa_1(\gamma_1)\hat{w}_1-\kappa_1(\gamma_1)w_1}\\
        &\leq O\Paren{\frac{\beta}{\kappa_2(\gamma_1)^2}}+|\kappa_1(\gamma_1)|\cdot\norm{\hat{w}_1-w_1}\\
        &\leq \frac{1}{\kappa_2(\gamma_1)^2} \cdot O(\beta),
    \end{align*}
    where we used \Cref{claim:kappa1kappa2bounded} for $|\kappa_1(\gamma_1)| \lesssim \kappa_2(\gamma_1)^{-1/2}$.
    This holds for all $\beta < c' \cdot \kappa_2(\gamma_1)$ and the sample complexity needed is
    \[
        n =  O\Paren{\frac{d^2}{\beta^2 \kappa_3(\gamma_1)^2} + \frac{d^2}{\kappa_3(\gamma_1)^4} + \frac{d}{\beta^2 \cdot \min\{\kappa_2(\gamma_1), \kappa_2(\gamma_2)\}}}.
    \]

    \textbf{Case (ii): $\alpha_1 \geq \alpha_*$.}
    In this case, we want to argue that both $\abs{\kappa_1(\gamma_1)}$ and $\abs{\kappa_1(\hat{\gamma}_1)}$ are small.
    Note that by picking $c$ sufficiently small, we can achieve $\gamma_1 \geq 2$ and get
    \begin{align*}
        \Abs{\kappa_1(\gamma_1)} &\lesssim \phi(\gamma_1) &\text{since $\Phi(\gamma_1) \geq \Phi(2)$}\\
        &\leq \Paren{\gamma_1 + \frac{1}{\gamma_1}} \cdot (1-\Phi(\gamma_1)) &\text{by \Cref{FACT:rational-appx-of-imr}}\\
        &\leq 2 \gamma_1 \cdot (1-\Phi(\gamma_1)) &\text{since $\gamma_1 \geq 1$}\\
        &\leq 2 \sqrt{2} \cdot (1-\Phi(\gamma_1)) \cdot \log^{1/2}\Paren{\frac{1}{1-\Phi(\gamma_1)}} &\text{by \Cref{claim:log-alpha-and-gamma}}\\
        &\lesssim \beta &\text{since $\alpha_1 = \Phi(\gamma_1) \geq \frac{1-\beta}{\log(1/\beta)}$}.
    \end{align*}
    Similarly, using \Cref{claim:one-minus-k2}, we get that $1-\kappa_2(\gamma_1) \lesssim \beta$.
    It remains to argue that $\abs{\kappa_1(\hat{\gamma}_1)}$ is also at most $O(\beta)$. By \Cref{claim:boundkappa1bykappa2}, it suffices to show that $\kappa_2(\hat{\gamma}_1) \geq 1 - \Omega(\beta)$, as long as $\beta$ is at most a sufficiently small constant (such that $\kappa_2(\hat{\gamma}) \geq \kappa_2(1)$, which implies $\hat{\gamma}_1 \geq 1$).
    Recall that
    \[
        M_2 = I_d - (1-\kappa_2(\gamma_1)) \cdot w_1w_1^\top - (1-\kappa_2(\gamma_1)\cdot w_2w_2^\top.
    \]
    Thus, we have 
    \[
        \hat{w}_1^\top M_2 \hat{w}_1 =  \kappa_2(\gamma_1) \cdot \Iprod{w_1, \hat{w}_1}^2 + \kappa_2(\gamma_2) \cdot \Iprod{w_2, \hat{w}_1}^2 + \Paren{1-\Iprod{w_1, \hat{w}_1}^2 - \Iprod{w_2, \hat{w}_1}^2}.
    \]
    Now, we have $\kappa_2(\gamma_1) \geq 1- \Omega(\beta)$.
    If $\alpha_2 \geq \alpha_*$, then we also have $\kappa_2(\gamma_2) \geq 1 - \Omega(\beta)$ and thus $\hat{w}_1^\top M_2 \hat{w}_1 \geq 1-\Omega(\beta)$.
    If $\alpha_2 \leq \alpha_*$, then we have $\norm{\hat{w}_2-w_2} \leq \beta$ and since $\hat{w}_1$ and $\hat{w}_2$ are orthogonal, we get
    \[
        \Abs{\hat{w}_1^\top w_2} = \Abs{\hat{w}_1^\top(w_2-\hat{w}_2)} \leq \beta.
    \]
    Thus, we have
    \[
        \hat{w}_1^\top M_2 \hat{w}_1 \geq (1-\Omega(\beta)) \Iprod{w_1, \hat{w}_1}^2 + (1-\Iprod{w_1, \hat{w}_1}^2 - \beta^2) \geq 1-\Omega(\beta)
    \]
    in this case as well (assuming $\beta \leq 1$).
    As before, using
    \[
        n = O\Paren{\frac{d}{\beta^2 \cdot \min\{\kappa_2(\gamma_1), \kappa_2(\gamma_2)\}}}
    \]
    samples, we get $\norm{M_2 - \widehat{M}_2} \leq \beta$ and thus also $\kappa_2(\hat{\gamma}_1) = \hat{w}_1^\top \widehat{M}_2 \hat{w}_1 \geq 1-\Omega(\beta)$.
    Hence, we have $\abs{\kappa_1(\hat{\gamma}_1)} \leq O(\beta)$ by \Cref{claim:boundkappa1bykappa2}.
    Together, this shows that for $\beta$ at most a sufficiently small constant, we get that
    \[
        \Norm{\kappa_1(\hat{\gamma}_1)\hat{w}_1-\kappa_1(\gamma_1)w_1} \leq O(\beta).
    \]

    \textbf{Sample complexity.}
    Summarized, we showed that for $\beta < \min\{c' \cdot \kappa_2(\gamma_1), c' \cdot \kappa_2(\gamma_2),c''\}$ (for a sufficiently small constant $c''$) we have
    \[
        \Norm{\hat{\mu} - \mu} \leq \Paren{\frac{1}{\kappa_2(\gamma_1)^2} + \frac{1}{\kappa_2(\gamma_2)^2}} \cdot O(\beta).
    \]
    The sample complexity needed for this consists of different terms. We always have
    \[
        O\Paren{\frac{d}{\beta^2 \cdot \min\{\kappa_2(\gamma_1), \kappa_2(\gamma_2)\}}}
    \]
    and for both $\ell \in \{1,2\}$ we need to add
    \[
        O\Paren{\frac{d^2}{\beta^2 \kappa_3(\gamma_\ell)^2} + \frac{d^2}{\kappa_3(\gamma_\ell)^4}}
    \]
    if $\alpha_\ell \leq \alpha_*$.
    It remains to bound $\kappa_i(\gamma_\ell)$.
    By \Cref{claim:var-est}, we have
    \[
         \kappa_2(\gamma_\ell)\geq \Omega(\min\{1,\log^{-1}(1/\alpha_\ell)\}) \geq \Omega(\min\{1,\log^{-1}(1/\alpha)\}).
    \]
    Similarly, if $\alpha_\ell \leq \alpha_*$, then we can bound $\kappa_3(\gamma_\ell)$ using \Cref{claim:skew-est} as follows
    \[
        \kappa_3(\gamma_\ell) \geq \Omega(\min\{\log^{-3/2}(1/\alpha), \beta\}).
    \]
    Putting these together, the sample complexity becomes
    \[
        n = O\Paren{\frac{d^2}{\beta^2}\poly\log(1/\alpha) + \frac{d^2}{\beta^4}}.
    \]
    Note that we can again assume with loss of generality that $\varepsilon \leq 1$ (otherwise the conclusion is trivial). Thus, picking (for a sufficiently small constant $c'''$)
    \[
        \beta = \min\Set{c' \cdot \kappa_2(\gamma_1), c' \cdot \kappa_2(\gamma_2),c'', c''' \cdot \min\{\kappa_2(\gamma_1), \kappa_2(\gamma_2)\} \cdot\varepsilon}
    \]
    proves the theorem.\footnote{Note that we do not need a $\poly\log(1/\alpha)$ for the $d^2/\varepsilon^4$ term. The reason is the following: We only need this term if $\alpha_\ell$ is close to $1$ (say larger than some constant). But in this case $\kappa_2(\gamma_\ell)$ is a constant and thus $O(d^2/\varepsilon^4)$ is sufficient.}
    Note that the runtime analysis of \Cref{algo:extensions} is exactly the same as the one for \Cref{algo:main} in the proof of \Cref{thm:main}.
\end{proof}

%% file: content/appendix.tex
\section{Deferred proofs from the main text}\label{APP:facts}
In this appendix we provide proofs of some claims that were deferred from the main text.

\subsection{Cumulants of the one-dimensional truncated distribution}\label{APP:facts:cumulantsof1dtruncated}
In this section we prove several facts from \Cref{sec:millsratio} on the first three cumulants of the one-dimensional random variable $Z_\gamma \sim \truncatedGauss{0}{1}{(-\infty, \gamma]}$.
We first prove \Cref{fact:first-3-cumulants} on the values of the cumulants that we restate below.
\begin{fact*}[Restatement of \Cref{fact:first-3-cumulants}]\factfirstthreecumulantstext
\end{fact*}
\begin{proof}
Recall that the first three cumulants can be written in terms of the moments of $Z_\gamma$.
In particular, recall from \eqref{EQ:ThreeCumulantsfor1DGauss} that
\begin{align*}
    \kappa_1(\gamma) &=\E\Brac{Z_\gamma},\\
    \kappa_2(\gamma) &=\E\Brac{\Paren{Z_\gamma-\E Z_\gamma}^2},\\
    \kappa_3(\gamma) &=\E\Brac{\Paren{Z_\gamma-\E Z_\gamma}^3}.
\end{align*}
Based on this, we can now write $\kappa_i$ in terms of $\phi(\gamma),\Phi(\gamma)$.
For $\kappa_1$, observe that
\begin{equation}\label{eqn:kappa1}
    \kappa_1(\gamma) = \E[Z_\gamma]  = \frac{1}{\Phi(\gamma)} \cdot \int_{-\infty}^{\gamma} t \phi(t) \mathrm{d} t = - \frac{\phi(\gamma)}{\Phi(\gamma)},
\end{equation}
where for the last equality we used $\frac{\mathrm{d}}{\mathrm{d}t} \varphi(t) = - t \varphi(t)$.
For $\kappa_2$, we first compute the second moment of $Z_\gamma$ as follows:
\begin{align*}
    \E\Brac{Z_\gamma^2} &= \frac{1}{\Phi(\gamma)} \cdot \int_{-\infty}^{\gamma} t^2 \phi(t) \mathrm{d} t\\
    &= \frac{1}{\Phi(\gamma)} \cdot \int_{-\infty}^{\gamma} t \cdot \frac{\mathrm{d}}{\mathrm{d}t} (-\varphi(t)) \mathrm{d} t\\
    &\overset{\text{(i)}}{=}\frac{1}{\Phi(\gamma)}\Paren{[-t\phi(t)]^{\gamma}_{-\infty}+\int_{-\infty}^{\gamma}\phi(t)\mathrm{d} t}\\
    &=-\gamma\frac{\phi(\gamma)}{\Phi(\gamma)}+1
\end{align*}
where in equality (i) we used integration by parts. Combined with \eqref{eqn:kappa1} this implies that 
\begin{equation}\label{eqn:kappa2}
    \kappa_2(\gamma) = \E\Brac{Z_\gamma^2}-\Paren{\E\Brac{Z_\gamma}}^2 =1-\gamma\frac{\phi(\gamma)}{\Phi(\gamma)}-\Paren{\frac{\phi(\gamma)}{\Phi(\gamma)}}^2.
\end{equation}
Similarly, we can compute the third moment
\begin{align}
     \E\Brac{Z_\gamma^3} &= \frac{1}{\Phi(\gamma)} \cdot \int_{-\infty}^{\gamma} t^3 \phi(t) \mathrm{d} t \nonumber\\
     &= \frac{1}{\Phi(\gamma)} \cdot \int_{-\infty}^{\gamma} t^2 \cdot \frac{\mathrm{d}}{\mathrm{d}t} (-\varphi(t)) \mathrm{d} t \nonumber\\
     &\overset{\text{(i)}}{=} \frac{1}{\Phi(\gamma)}\Paren{[-t^2\phi(t)]^{\gamma}_{-\infty}+2\int_{-\infty}^{\gamma}t\phi(t)\mathrm{d} t} \nonumber\\
     &\overset{\text{(ii)}}{=} -\gamma^2 \cdot \frac{\phi(\gamma)}{\Phi(\gamma)} - 2 \cdot \frac{\phi(\gamma)}{\Phi(\gamma)} \label{eqn:thirdmoment},
\end{align}
where in (i) we again used integration by parts and in (ii) we use \eqref{eqn:kappa1}.
Thus, we can compute $\kappa_3$ as follows
\begin{align*}
    \kappa_3(\gamma)&=\E\Brac{\Paren{Z_\gamma-\E\Brac{Z_\gamma}}^3} \\
    &= \E\Brac{Z_\gamma^3}-3\E\Brac{Z_\gamma^2}\cdot\E\Brac{Z_\gamma}+3\E\Brac{\Z_\gamma}\cdot \Paren{\E\Brac{Z_\gamma}}^2 -\Paren{\E\Brac{Z_\gamma}}^3\\
    &= \E\Brac{Z_\gamma^3} - 3\E\Brac{Z_\gamma} \cdot \Paren{\E\Brac{Z_\gamma^2}-\Paren{\E\Brac{Z_\gamma}}^2}-\Paren{\E\Brac{Z_\gamma}}^3\\
    &= -\gamma^2 \cdot \frac{\phi(\gamma)}{\Phi(\gamma)} - 2 \cdot \frac{\phi(\gamma)}{\Phi(\gamma)} -3 \cdot \Paren{- \frac{\phi(\gamma)}{\Phi(\gamma)}} \cdot \Paren{1-\gamma\frac{\phi(\gamma)}{\Phi(\gamma)}-\Paren{\frac{\phi(\gamma)}{\Phi(\gamma)}}^2}-\Paren{- \frac{\phi(\gamma)}{\Phi(\gamma)}}^3,
\end{align*}
where in the last step we used \eqref{eqn:kappa1}, \eqref{eqn:kappa2} and \eqref{eqn:thirdmoment}.
After simplification we have that 
\begin{equation*}
   \kappa_3(\gamma)=\frac{\phi(\gamma)}{\Phi(\gamma)} \cdot \left( 1 - \gamma^2 - 3 \gamma \cdot \frac{\phi(\gamma)}{\Phi(\gamma)} - 2 \cdot \left(\frac{\phi(\gamma)}{\Phi(\gamma)}\right)^2 \right),
\end{equation*}
which completes the proof.
\end{proof}

For several of the following proofs, we use the approximation of the inverse Mills ratio ${h(t) = \frac{\phi(t)}{1-\Phi(t)}}$ from \Cref{FACT:rational-appx-of-imr}. We will in particular need the following bounds from \Cref{FACT:rational-appx-of-imr} that hold for all $x > 0$:
\begin{align*}
    h(x) &> r_1(x) = x &(k=0)\phantom{.}\\
    h(x) &< r_2(x) = \frac{x^2+1}{x} &(k=1)\phantom{.}\\
    h(x) &> r_3(x) = \frac{x(x^2+3)}{x^2+2} &(k=2)\phantom{.}\\
    h(x) &< r_4(x) = \frac{x^4 + 6x^2+3}{x^3+5x} &(k=3)\phantom{.}\\
    h(x) &> r_5(x) = \frac{x(x^4+10x^2+15)}{x^4+9x^2+8} &(k=4).
\end{align*}

Using these, we now prove the following two claims about $\kappa_1$ and $\kappa_2$.
\begin{claim*}[Restatement of \Cref{claim:1-lipschitz-k1}]\claimonelipschitzkonetext     
\end{claim*}

\begin{claim*}[Restatement of \Cref{claim:bound-on-kappa2}]\claimboundonkappatwotext
\end{claim*}
We first prove \Cref{claim:bound-on-kappa2} and then use it to prove \Cref{claim:1-lipschitz-k1}.

\begin{proof}[Proof of \Cref{claim:bound-on-kappa2}]
    First, note that $\kappa_2(\gamma)=\E\brac{\paren{Z_\gamma-\E \brac{Z_\gamma}}^2} > 0$.
    To see $\kappa_2(t)\leq 1$, recall from \eqref{eqn:kappa2} that
    \[
        1-\kappa_2(t)=t\frac{\phi(t)}{\Phi(t)}+\Paren{\frac{\phi(t)}{\Phi(t)}}^2.
    \]
    For $t \geq 0$, both terms are non-negative and thus $1-\kappa_2(t) \geq 0$. 
    For $t <0$, we get from \Cref{FACT:rational-appx-of-imr} that $\frac{\phi(t)}{\Phi(t)}=\frac{\phi(-t)}{1-\Phi(-t)}\geq -t$ and thus
    \[
        1-\kappa_2(t)=\frac{\phi(t)}{\Phi(t)}\Paren{t+\frac{\phi(t)}{\Phi(t)}}\geq \frac{\phi(t)}{\Phi(t)}(t-t)=0
    \]
    in this case as well.
\end{proof}

\begin{proof}[Proof of \Cref{claim:1-lipschitz-k1}]
    Observe that, since $\kappa_1(t) = -\frac{\phi(t)}{\Phi(t)}$ by \Cref{fact:first-3-cumulants}, we get that
    \[
        \ddt\kappa_1(t)=t\frac{\phi(t)}{\Phi(t)}+\Paren{\frac{\phi(t)}{\Phi(t)}}^2 = 1-\kappa_2(t) \leq 1,
    \]
    where we used again \Cref{fact:first-3-cumulants} and that $\kappa_2(t) \geq 0$ by \Cref{claim:bound-on-kappa2}.
\end{proof}

Next, we prove the following claim that shows that the function $\kappa_2$ is monotonically increasing.
\begin{claim*}[Restatement of \Cref{claim:monotonicity-of-kappa2}]\claimmonotonicityofkappatwotext
\end{claim*}

\begin{proof}[Proof of \Cref{claim:monotonicity-of-kappa2}]
    Recall from \Cref{fact:first-3-cumulants} that $\kappa_2(t)=1-t\frac{\phi(t)}{\Phi(t)}-\Paren{\frac{\phi(t)}{\Phi(t)}}^2$.
    Hence, we get that 
    \begin{align*}
        \ddt\kappa_2(t)&=-\frac{\phi(t)}{\Phi(t)}-t\ddt\frac{\phi(t)}{\Phi(t)}-2\Paren{\frac{\phi(t)}{\Phi(t)}} \cdot \ddt\frac{\phi(t)}{\Phi(t)}\\
        &=-\frac{\phi(t)}{\Phi(t)}+\Paren{t+2\frac{\phi(t)}{\Phi(t)}} \cdot \ddt\frac{-\phi(t)}{\Phi(t)}\\
        &=-\frac{\phi(t)}{\Phi(t)}+\Paren{t+2\frac{\phi(t)}{\Phi(t)}}\Paren{t\frac{\phi(t)}{\Phi(t)}+\Paren{\frac{\phi(t)}{\Phi(t)}}^2}\\
        &=(t^2-1)\frac{\phi(t)}{\Phi(t)}+3t\Paren{\frac{\phi(t)}{\Phi(t)}}^2+2\Paren{\frac{\phi(t)}{\Phi(t)}}^3=-\kappa_3(t)
    \end{align*}
    by \Cref{fact:first-3-cumulants}.
    Thus, $\ddt\kappa_2(t)>0$ if and only if $f(t)=t^2-1+3t\cdot \frac{\phi(t)}{\Phi(t)}+2\Paren{\frac{\phi(t)}{\Phi(t)}}^2>0$.
    We show the latter by distinguishing the three cases (i) $1 \leq t$, (ii) $0 \leq t < 1$ and (iii) $t < 0$.

    \textbf{Case (i): $1 \leq t$}.
    In this case, $t^2 -1 \geq 0$ and the other two terms are positive. Thus, if~$t \geq 1$, then $f(t) > 0$ holds.

    \textbf{Case (ii): $0 \leq t < 1$}.
    A direct computation shows $f(0) = \frac{4}{\pi}-1 > 0$. For $t \in (0,1)$, since~${\phi(t) = \ddt \Phi(t)}$ is decreasing, we can upper bound $\Phi(t)$ by its tangent line at $t=0$, i.e.,
    \[
        \Phi(t)\leq\frac{1}{2}+\frac{t}{\sqrt{2\pi}}.
    \]
    Therefore, we obtain
    \[
        \frac{\phi(t)}{\Phi(t)}\geq\frac{\exp\Paren{-t^2/2}}{\sqrt{\pi/2}+t} > 0
    \]
    and thus
    \[
        f(t)\geq t^2-1+\frac{\exp\Paren{-t^2/2}}{\sqrt{\pi/2}+t}\cdot\Paren{3t+2\frac{\exp\Paren{-t^2/2}}{\sqrt{\pi/2}+t}}.
    \]
    We show in \Cref{FACT:computationkappa2monotonic} in \Cref{APP:computations} that this lower bound is monotonic in $(0,1)$. Thus, the lower bound attains its minimum at $t=0$ and we have for any $t \in (0,1)$ that $f(t) \geq \frac{4}{\pi}-1 > 0$.

    \textbf{Case (iii): $t < 0$}.
    Finally, for $t < 0$ consider the quadratic function
    \[
        q_t(x) = t^2-1+3tx+2x^2.
    \]
    We have $\ddx q_t(x) = 3t + 4x$ and thus $q_t(x)$ is increasing in $x$ for $x > -\frac{3}{4}t$.
    Note that we have
    \[
        {f(t) = q_t\Paren{\frac{\phi(t)}{\Phi(t)}} = q_t(h(-t))}
    \]
    for $h(-t)$ the inverse Mills ratio as in as in \Cref{sec:millsratio} (note that $-t > 0$).
    By \Cref{FACT:ir-appx-of-imr}, we have
    \[
        h(-t) > \frac{1}{4}\Paren{\sqrt{t^2+8}-3t} > -\frac{3}{4}t
    \]
    and hence, we get that
    \begin{align*}
        q_t(h(-t)) &> q_t\Paren{\frac{1}{4}\Paren{\sqrt{t^2+8}-3t}}\\
        &= t^2-1 + 3t \cdot \frac{1}{4}\Paren{\sqrt{t^2+8}-3t} + 2 \Paren{\frac{1}{4}\Paren{\sqrt{t^2+8}-3t}}^2\\
        &= t^2-1 + \frac{3}{4}t\sqrt{t^2+8} -\frac{9}{4}t^2+\frac{1}{8}\Paren{t^2+8-6t\sqrt{t^2+8}+9t^2} = 0,
    \end{align*}
    which shows that $f(t) > 0$ also holds for $t < 0$.
    This completes the proof.
\end{proof}

Next, we prove the following claim that gives a uniform upper bound on $\kappa_1(t)\sqrt{\kappa_2(t)}$.
\begin{claim*}[Restatement of \Cref{claim:kappa1kappa2bounded}]\claimkappaonekappatwoboundedtext
\end{claim*}

\begin{proof}
    Recall from \Cref{fact:first-3-cumulants} that $\kappa_1(t)=-\frac{\phi(t)}{\Phi(t)}$ and $\kappa_2(t) = 1- t \frac{\phi(t)}{\Phi(t)}-\Paren{\frac{\phi(t)}{\Phi(t)}}^2$. By \Cref{claim:bound-on-kappa2}, we also have that $0 < \kappa_2(t) \leq 1$.
    To show that $\abs{\kappa_1(t)}\cdot\sqrt{\kappa_2(t)}$ is uniformly bounded, we distinguish the two cases (i) $t \geq -4$ and (ii) $t < -4$.

    \textbf{Case (i): $t \geq -4$}. In this case we have $\Phi(t) \geq \Phi(-4) > 0$. Thus, we have, using that $\kappa_2(t) \leq 1$ for all $t$,
    \[
        \abs{\kappa_1(t)}\cdot\sqrt{\kappa_2(t)}={\frac{\phi(t)}{\Phi(t)}}\cdot\sqrt{\kappa_2(t)}\leq \frac{\phi(t)}{\Phi(-4)}\leq \frac{1}{\sqrt{2\pi} \Phi(-4)}.
    \]

    \textbf{Case (ii): $t < -4$}. For this case, we want to use \Cref{FACT:rational-appx-of-imr} to bound $\kappa_1(t)$ and $\kappa_2(t)$.
    First, for $\kappa_1$ we have $\abs{\kappa_1(t)} = h(-t) \leq r_2(-t) = -t -1/t \leq 2 \abs{t}$ for $t < -4 < -1$, where, as before, $h(-t)$ is the inverse Mills ratio as in \Cref{sec:millsratio}.
    
    For $\kappa_2$, consider the quadratic function $q_t(x) = -x^2 - tx +1$.
    For $t < -4 < 0$, we can write
    \[
        \kappa_2(t) = 1 - t h(-t) - (h(-t))^2 = q_t(h(-t)).
    \]
    We have $\ddx q_t(x) = -2x -t$ and hence $q_t(x)$ is increasing in $x$ on the interval $(-\frac{t}{2},\infty)$.
    Now, using \Cref{FACT:rational-appx-of-imr}, we get
    \[
        h(-t) \geq r_5(-t) = \frac{-t\Paren{15 + 10t^2 + t^4}}{8 + 9t^2 + t^4} \geq -\frac{t}{2},
    \]
    where we used $-t > 0$ and $t^2-1 \geq 0$ for the last step. Hence, we have $q_t(h(-t)) \leq q_t(r_5(-t))$ and thus
    \[
        \kappa_2(t) = q_t(h(-t)) \leq q_t(r_5(-t)) = \frac{t^6 + 12t^4 + 39t^2 + 64}{\Paren{t^4+9t^2+8}^2} \leq \frac{4t^6}{t^8} = \frac{4}{t^2},
    \]
    where we used $t < -4$ to get $12t^4 + 39t^2 + 64 \leq 3t^6$.

    Finally, combining the bounds on $\kappa_1(t)$ and $\kappa_2(t)$, we get that, for all $t < -4$,
    \[
        \abs{\kappa_1(t)}\cdot \sqrt{\kappa_2(t)}\leq 2 \abs{t} \cdot \frac{2}{\abs{t}} = 4,
    \]
    which completes the proof.
\end{proof}

The following claim shows both $\kappa_3(t)$ and $\kappa_1(t) \cdot \kappa_3(t)$ are also uniformly bounded.
\begin{claim*}[Restatement of \Cref{claim:kappa1kappa3bounded}]\claimkappaonekappathreeboundedtext
\end{claim*}

\begin{proof}
    We first show that $\kappa_3(t)$ is bounded on $\R$.
    By \Cref{claim:bound-on-kappa2} we have $\kappa_3(t) < 0$ and thus, using \Cref{fact:first-3-cumulants}, we get
    \[
        \abs{\kappa_3(t)}=\frac{\phi(t)}{\Phi(t)}\cdot\Paren{2\Paren{\frac{\phi(t)}{\Phi(t)}}^2+3t\frac{\phi(t)}{\Phi(t)}+t^2-1}.
    \]
    We distinguish the two cases (i) $t > -1$ and (ii) $t \leq -1$.

    \textbf{Case (i): $t > -1$}.
    Note that $\phi(t)$, $t\phi(t)$ and $t^2\phi(t)$ are all uniformly bounded over $t \in \R$.
    For $t > -1$ we furthermore have $\frac{1}{\Phi(t)} \leq \frac{1}{\Phi(-1)} \leq O(1)$.
    Putting these together, we get that for~$t > -1$ we have $\abs{\kappa_3(t)} \leq O(1)$.

    \textbf{Case (ii): $t \leq -1$}. For this case, we want to again use \Cref{FACT:rational-appx-of-imr}. For $t < -1 < 0$, we again use that $\frac{\phi(t)}{\Phi(t)} = h(-t)$. By \Cref{FACT:rational-appx-of-imr}, we get that
    \[
        \frac{-t \cdot (t^2 + 3)}{t^2 + 2}= r_3(-t) < h(-t) =\frac{\phi(t)}{\Phi(t)}  < r_2(-t) = -t - \frac{1}{t}.
    \]
    We thus get that
    \begin{align*}
        \abs{\kappa_3(t)} &= h(-t) \cdot \Paren{2h(-t)^2 + 3t \cdot h(-t) + t^2 -1}\\
        &\leq \Paren{-t - \frac{1}{t}} \cdot \Paren{2 \Paren{-t - \frac{1}{t}}^2 + 3t \cdot \frac{-t \cdot (t^2 + 3)}{t^2 + 2} + t^2 -1}\\
        &= \Paren{-t - \frac{1}{t}} \cdot \frac{2(-t^2-1)^2(t^2+2)-3t^4(t^2+3)+t^2(t^2-1)(t^2+2)}{t^2(t^2+2)}\\
        &= \Paren{-t - \frac{1}{t}} \cdot 4 \cdot \frac{2t^2 +1}{t^2(t^2+2)}\\
        &\leq -2t \cdot \frac{8}{t^2} \leq 16,
    \end{align*}
    where we used $-t-1/t \leq -2t$ and $2t^2+1 \leq 2(t^2+2)$ for $t \leq -1$.

    Hence, $\abs{\kappa_3(t)}$ is uniformly bounded over $t \in \R$.
    For $\kappa_1(t) \cdot \kappa_3(t)$, first note that by \Cref{fact:first-3-cumulants} we have $\kappa_1(t) = - \frac{\phi(t)}{\Phi(t)} < 0$ and thus $\kappa_1(t) \cdot \kappa_3(t) > 0$. To show a uniform upper bound, we distinguish the same two cases as above and use a similar strategy. Note that we can use the same bounds and multiply them by $\frac{\phi(t)}{\Phi(t)}$. If $t > -1$, we have
    \[
        \frac{\phi(t)}{\Phi(t)} \leq \frac{\phi(0)}{\Phi(-1)} \leq O(1),
    \]
    and thus $\kappa_1(t) \cdot \kappa_3(t) \leq O(1)$ for these $t$. If $t \leq -1$, again by \Cref{FACT:rational-appx-of-imr}, we get
    \[
        \frac{\phi(t)}{\Phi(t)} = h(-t) \leq -t - \frac{1}{t} \leq -2t \leq 2
    \]
    and thus $\kappa_1(t) \cdot \kappa_3(t) \leq 32$ for these $t$.
    Hence, also $\kappa_1(t) \cdot \kappa_3(t)$ is uniformly bounded, which completes the proof.
\end{proof}

We now prove how to lower bound $\kappa_1$ in terms of $1-\kappa_2$.
\begin{claim*}[Restatement of \Cref{claim:boundkappa1bykappa2}]\claimboundkappaonebykappatwotext
\end{claim*}
\begin{proof}
    Using \Cref{fact:first-3-cumulants}, we have
    \[
        \Abs{\kappa_1(t)} = \frac{\phi(t)}{\Phi(t)} \leq t \cdot \frac{\phi(t)}{\Phi(t)} \leq  t \cdot \frac{\phi(t)}{\Phi(t)} + \Paren{\frac{\phi(t)}{\Phi(t)}}^2 = 1-\kappa_2(t),
    \]
    which completes the proof.
\end{proof}

Next, we prove the following lower bound on $\kappa_2$.

\begin{claim*}[Restatement of \Cref{claim:kappa2-lower}]\claimkappatwolowertext
\end{claim*}
\begin{proof}
    We use an analogous strategy as in the second part of the proof of \Cref{claim:kappa1kappa2bounded}, where we showed an upper bound on $\kappa_2(t)$ for $t < -4$.
    Using \Cref{FACT:rational-appx-of-imr} for $-t$, we get that, for $t \leq -3 < 0$,
    \[
        -t = r_1(-t) < h(-t) < r_4(-t) = - \frac{t^4 + 6t^2 + 3}{t(t^2+5)}.
    \]
    Now, consider again the quadratic function $q_t(x)=-x^2-tx+1$, which decreases on the interval $(-t/2,\infty)$.
    As before, we have $\kappa_2(t) = q_t(h(-t))$ and thus we get that
    \[
        \kappa_2(t) = q_t(h(-t)) \geq q_t(r_4(-t)) = -\Paren{- \frac{t^4 + 6t^2 + 3}{t(t^2+5)}}^2 - t\Paren{- \frac{t^4 + 6t^2 + 3}{t(t^2+5)}} + 1.
    \]
    Simplifying this expression, we can conclude that
    \[
        \kappa_2(t) \geq \frac{t^4+4t^2-9}{t^6+10t^4+25t^2} \geq \frac{t^4}{4t^6} = \frac{1}{4t^2},
    \]
    where we used $4t^2-9 \geq 0$ and $10t^4+25t^2 \leq 3t^6$ for $t \leq -3$.
\end{proof}

Similarly, we prove the following lower bound on $\kappa_3$.

\begin{claim*}[Restatement of \Cref{claim:kappa3-lower}]\claimkappathreelower
\end{claim*}
\begin{proof}
    We refine the argument from \Cref{claim:monotonicity-of-kappa2}, where we showed that $\kappa_3(t) > 0$.
    \Cref{claim:monotonicity-of-kappa2}, together with \Cref{fact:first-3-cumulants}, shows that
    \[
        |\kappa_3(t)|=-\kappa_3(t)= h(-t)\left(t^2-1+3th(-t)+2\Paren{h(-t)}^2\right),
    \]
    where $h(-t)$ is again the inverse Mills ratio.
    From \Cref{FACT:rational-appx-of-imr} we have
    \[
        h(-t) > r_5(-t) = -\frac{t(t^4+10t^2+15)}{t^4+9t^2+8} \geq -t.
    \]
    Consider the quadratic function $q_t(x) = 2x^2 + 3tx + t^2-1$.
    We have $|\kappa_3(t)| = h(-t) q_t(h(-t))$.
    As already argued in the proof of \Cref{claim:monotonicity-of-kappa2}, this function is increasing in $x$ for $x > -\frac{3}{4}t$.
    Thus, we get that
    \[
        q_t(h(-t)) \geq q_t(r_5(-t)) = \frac{2t^4 + 10t^2-64}{(t^4+9t^2+8)^2} \geq \frac{1}{50t^4},
    \]
    where we used $10t^2-64 \geq -\frac{3}{2}t^4$ and $9t^2+8 \leq 4t^4$ for $t \leq -2$.
    Finally, again by \Cref{FACT:rational-appx-of-imr}, we have that $h(-t) \geq -t$ and thus
    \[
        |\kappa_3(t)| = h(-t) q_t(h(-t)) \geq -t \cdot \frac{1}{50t^4} = \frac{1}{50|t|^3}
    \]
    since $t < 0$ and thus $-t^3 = |t|^3$.
\end{proof}

Next, we prove the following lower bound on $\kappa_2(t)$.
\begin{claim*}[Restatement of \Cref{claim:one-minus-k2}]\claimoneminusktwo
\end{claim*}
\begin{proof}
    Note that for $t\geq 1$, we have $\Phi(t)\geq\frac{1}{2}$ as well as $2\phi(t)\leq \sqrt{\frac{2}{\pi}}\leq t$.
    It follows that, using \Cref{fact:first-3-cumulants},
    \begin{align*}
        1-\kappa_2(t) &=t\frac{\phi(t)}{\Phi(t)}+\Paren{\frac{\phi(t)}{\Phi(t)}}^2 = \frac{\phi(t)}{\Phi(t)}\Paren{t+\frac{\phi(t)}{\Phi(t)}} \leq 2\phi(t)\cdot\Paren{t+2\phi(t)} \leq 4t\cdot\phi(t).
    \end{align*}
    Applying \Cref{FACT:rational-appx-of-imr}, we have $\frac{\phi(t)}{1-\Phi(t)} \leq t + \frac{1}{t} \leq 2t$ and thus $\phi(t) \leq 2t (1-\Phi(t))$. We get that
    \[
        1-\kappa_2(t) \leq 8 t^2 \cdot (1-\Phi(t)).
    \]
    Applying \Cref{FACT:rational-appx-of-imr} again, we have $\frac{\phi(t)}{1-\Phi(t)} \geq t$ and hence
    \[
        \frac{1}{1-\Phi(t)} \geq \frac{t}{\phi(t)} \geq t\sqrt{2\pi}\cdot e^{t^2/2}\geq  e^{t^2/2} \quad\implies\quad \log\Paren{\frac{1}{1-\Phi(t)}} \geq t^2/2.
    \]
    Combining this with the above, we get
    \[
        1-\kappa_2(t) \leq 16 \cdot (1-\Phi(t)) \cdot \log\Paren{\frac{1}{1-\Phi(t)}},
    \]
    as claimed.
\end{proof}

We now prove the following claim bounding $\log\Paren{1/(1-\Phi(t))}$.

\begin{claim*}[Restatement of \Cref{claim:log-alpha-and-gamma}]\claimlogalphaandgammatext    
\end{claim*}
\begin{proof}
    Using \Cref{FACT:rational-appx-of-imr} we have
    \[
        \frac{1}{1-\Phi(t)} \leq \frac{1}{\phi(t)} \cdot \Paren{t + \frac{1}{t}}\leq \frac{1}{\phi(t)} \cdot (t+1)= \sqrt{2\pi}e^{t^2/2} \cdot \Paren{t + 1},
    \]
    where we used $\frac{1}{t} \leq 1$ for $t \geq 2 \geq 1$.
    Thus, we get that
    \begin{align*}
        \log\Paren{\frac{1}{1-\Phi(t)}} &\leq \frac{t^2}{2} + \log\Paren{t + 1} + \log\Paren{\sqrt{2\pi}}\\
        &\leq \frac{t^2}{2} + t + \log\Paren{\sqrt{2\pi}}\\
        &\leq 2 t^2, 
    \end{align*}
    where we used $t > 0$ for $\log(t+1) \leq t$ and $t \geq 2$ for $t + \log\Paren{\sqrt{2\pi}} \leq \frac{3}{2} t^2$.
    The other direction was already proved in \Cref{claim:one-minus-k2} for $t \geq 1$.
\end{proof}

Finally, we want to prove the following claim about the asymptotics of $\sk(t)$.

\begin{claim*}[Restatement of \Cref{claim:skew-asymptotic}]\claimskewasymptotictext
\end{claim*}

In order to prove this claim, we first prove the following upper bound on $\kappa_3(t)$.

\begin{claim}[Upper bound on $\kappa_3(t)$]\label{claim:kappa3-upper}
    For $t\geq 1$ and $\delta = 1-\Phi(t)$, we have that
    \[
        \abs{\kappa_3(t)}\leq 40\sqrt{2}\cdot\delta \cdot \log^{3/2}(1/\delta).
    \]
\end{claim}
\begin{proof}[Proof of \Cref{claim:kappa3-upper}]
    For $t \geq 1$, we have as before (cf. \Cref{fact:first-3-cumulants} and \Cref{claim:monotonicity-of-kappa2})
    \[
        \Abs{\kappa_3(t)}=\frac{\phi(t)}{\Phi(t)}\Paren{t^2-1+3t\cdot \frac{\phi(t)}{\Phi(t)}+2\Paren{\frac{\phi(t)}{\Phi(t)}}^2}.
    \]
    For $t \geq 1$, we have $\Phi(t) \geq \Phi(0) \geq \frac{1}{2}$ and thus
    \begin{align*}
        \Abs{\kappa_3(t)} &\leq 2\phi(t)\Paren{t^2-1+6t \cdot \phi(t) +8\phi(t)^2}\\
        &\leq 2\phi(t)\Paren{t^2-1+6t \cdot \frac{1}{\sqrt{2\pi}} +8\cdot \frac{1}{2\pi}}\\
        &\leq 10 \phi(t) t^2,
    \end{align*}
    where we used in the last step that $\frac{6}{\sqrt{2\pi}}t + \frac{8}{2\pi} \leq 4t^2$ for $t \geq 1$.
    Now, by \Cref{FACT:rational-appx-of-imr}, we have that ${h(t) \leq t+1/t \leq 2}t$ for $t \geq 1$.
    Since $h(t) = \frac{\phi(t)}{1-\Phi(t)} = \frac{\phi(t)}{\delta}$, it follows that
    \[
        \Abs{\kappa_3(t)} \leq 20 \cdot \delta \cdot t^3.
    \]
    It remains to bound $t$ in terms of $\delta$.
    Again by \Cref{FACT:rational-appx-of-imr}, we have $h(t) \geq t \geq 1$ for $t \geq 1$ and thus
    \[
        \delta \leq \phi(t) = \frac{1}{\sqrt{2\pi}}e^{t^2/2} \leq e^{-\frac{t^2}{2}} 
    \]
    or equivalently
    \[
        t \leq \sqrt{2\log(1/\delta)}. 
    \]
    Combining this with the above we get that
    \[
        \Abs{\kappa_3(t)} \leq 20 \cdot \delta \cdot \Paren{2\log(1/\delta)}^{3/2} = 40 \sqrt{2} \cdot \delta \cdot \log^{3/2}(1/\delta),
    \]
    which completes the proof.
\end{proof}
\begin{proof}[Proof of \Cref{claim:skew-asymptotic}]
    First, from \Cref{claim:bound-on-kappa2,claim:monotonicity-of-kappa2} we know that $0<\kappa_2(t)\leq 1$ and $\kappa_3(t)<0$ for all $t \in \R$.
    Hence, we have in particular that
    \[
        \sk(t)=\frac{\kappa_3(t)}{\kappa_2(t)^{3/2}}<0.
    \]
    Define $\delta = 1-\Phi(t)$.
    By \Cref{claim:skew-est}, we have that, for $t \geq C$ for some constant $C$ (such that the minimum in \Cref{claim:skew-est} is indeed the third term)
    \[
        |\kappa_3(t)| \geq \frac{1}{4\sqrt{2}}\cdot \delta \cdot \log^{3/2}(1/\delta).
    \]
    Using $\kappa_2(t) \leq 1$, we get that
    \[
        \frac{|\sk(t)|}{\delta \cdot \log^{3/2}(1/\delta)} \geq \frac{1}{4\sqrt{2}}
    \]
    On the interval $t \in [0,C]$, the function $\frac{|\sk(t)|}{\delta \cdot \log^{3/2}(1/\delta)}$ is continuous and strictly positive and thus there is $m' > 0$ such that
    \[
        \frac{|\sk(t)|}{\delta \cdot \log^{3/2}(1/\delta)} \geq m'
    \]
    for $t \in [0,C].$
    Defining $m = \min\left\{m',\frac{1}{4\sqrt{2}}\right\}$ shows the lower bound.
    For the upper bound, we get by \Cref{claim:kappa3-upper} that
    \[
        \abs{\kappa_3(t)}\leq 40\sqrt{2}\cdot\delta \cdot \log^{3/2}(1/\delta)
    \]
    for $t \geq 1$.
    Using \Cref{claim:monotonicity-of-kappa2}, we get that $0 < \kappa_2(0) \leq \kappa_2(t)$ for all $t \geq 0$ and thus
    \[
        \frac{|\sk(t)|}{\delta \cdot \log^{3/2}(1/\delta)} \leq \frac{40\sqrt{2}}{\kappa_2(0)^{3/2}}
    \]
    for $t \geq 1$. As before, the function $\frac{|\sk(t)|}{\delta \cdot \log^{3/2}(1/\delta)}$ is continuous for $t \in [0,1]$ and thus has a uniform upper bound $M'$. We get the claimed upper bound by defining $M = \max\left\{M', \frac{40\sqrt{2}}{\kappa_2(0)^{3/2}}\right\}$.
\end{proof}

\subsection{Matrix and vector calculations}\label{APP:facts:matrixvectorcalc}
In this section, we prove several linear algebra facts that we deferred from \Cref{sec:setup,sec:fullproof}.
We first show that a linear transformation of a Gaussian truncated by a halfspace is again a (different) Gaussian truncated by a (different) halfspace. 

\begin{claim*}[Restatement of \Cref{claim:affine-transform}]\claimaffinetransformtext
\end{claim*}

\begin{proof}
We prove this claim by computing the density of $y$ explicitly.
Note that the density of $x$ is
\[
    f(x)\propto\exp\left(-\frac{1}{2}(x-\mu)^\top\Sigma^{-1}(x-\mu)\right)\mathbbb{1}\Brac{w^\top x\leq\tau}.
\]
Hence, the density of $y$ is
\begin{align*}
g(y)&=\frac{1}{|\text{det}(A)|}f(A^{-1}(y-b))\\
&\propto\exp\left(-\frac{1}{2}(A^{-1}(y-b)-\mu)^\top\Sigma^{-1}(A^{-1}(y-b)-\mu)\right)\mathbbb{1}\Brac{w^\top \Paren{A^{-1}(y-b)}\leq\tau}\\
&=\exp\left(-\frac{1}{2}(y-b-A\mu)^\top(A^{-\top}\Sigma^{-1}A^{-1})(y-b-A\mu)\right)\mathbbb{1}\Brac{( A^{-\top}w)^\top(y-b)\leq\tau} \\ 
&=\exp\left(-\frac{1}{2}\tp{(y-(A\mu + b))}(A\Sigma\tp{A})^{-1}(y-(A\mu + b))\right)\mathbbb{1}\Brac{\tp{( A^{-\top}w)}y\leq\tau+\Iprod{A^{-\top} w,b}}
\end{align*}
Thus,  $g(y)$ is the density of the Gaussian $\normal{A\mu+b,A\Sigma A^\top}$ truncated to a halfspace with normal $A^{-\top} w$. To conclude the lemma, we normalize the normal vector.
Thus, the halfspace is
\[
    H^{\prime}=\Set{y:\Iprod{\frac{A^{-\top} w}{\Norm{A^{-\top} w}},y}\leq\frac{\tau+\Iprod{A^{-\top} w,b}}{\Norm{A^{-\top} w}}},
\]
which completes the proof.
\end{proof}

Next, we show that the Frobenius norm of $I_d - A^{-\frac{1}{2}}BA^{-\frac{1}{2}}$ is invariant under the transformation $A \mapsto LAL^\top$, $B \mapsto LBL^\top$ for an invertible matrix $L$.

\begin{fact*}[Restatement of \Cref{fact:affine-invariance}]\factaffineinvariancetext   
\end{fact*}
\begin{proof}
    Let $X=LA^{\frac{1}{2}}$ and $R=(XX^\top)^{-\frac{1}{2}}X$. Note that $LAL^\top = XX^\top$ and observe that ${RR^\top=R^\top R=I_d}$, i.e., $R$ is orthonormal.
    The right-hand side of the identity can now be rewritten as
    \begin{align*}
        \Norm{I_d-(LAL^\top)^{-\frac{1}{2}}LBL^\top(LAL^\top)^{-\frac{1}{2}}}_F &= \Norm{I_d-(XX^\top)^{-\frac{1}{2}}XA^{-\frac{1}{2}}BA^{-\frac{1}{2}}X^\top(XX^\top)^{-\frac{1}{2}}}_F\\
        &=\Norm{I_d-RA^{-\frac{1}{2}}BA^{-\frac{1}{2}}R^\top}_F\\
        &= \Norm{R\Paren{I_d-A^{-\frac{1}{2}}BA^{-\frac{1}{2}}}R^\top}_F.
    \end{align*}
    Let $M = I_d-A^{-\frac{1}{2}}BA^{-\frac{1}{2}}$.
    Using that the trace is cyclic, we get that
    \[
        \Norm{RMR^\top}_F = \tr\Paren{RMR^\top\Paren{RMR^\top}^\top} = \tr\Paren{MM^\top} = \Norm{M}_F^2,
    \]
    which completes the proof.
\end{proof}

We now prove the following claim that shows that the best rank-1 approximation of a matrix $A$ that is close to some matrix $B$ is also a good approximation to $B$.

\begin{claim*}[Restatement of \Cref{claim:best_rank-1_approximation}]\claimbestrankoneapproximationtext
\end{claim*}

\begin{proof}
    As $\lambda_1 v_1v_1^\top$ is the best rank-1 approximation of $A$, we have 
    \(
        \Norm{\lambda_1 v_1v_1^\top-A}_F^2\leq\Norm{A-B}_F^2.
    \)
    We thus get that
    \begin{align*}
        \Norm{\lambda_1 v_1 v_1^\top-B}_F^2 &= \Norm{\Paren{\lambda_1 v_1 v_1^\top-A} + \Paren{A-B}}_F^2\\
        &= \Norm{\lambda_1 v_1 v_1^\top-A}_F^2 + 2\iprod{\lambda_1 v_1 v_1^\top-A,A-B} + \Norm{A-B}_F^2\\
        &\leq 2\Paren{\Norm{A-B}_F^2 + \iprod{\lambda_1 v_1 v_1^\top-A,A-B}}\\
        &= 2\iprod{\lambda_1 v_1 v_1^\top-B,A-B},
    \end{align*}
    which completes the proof.
\end{proof}

Next, we prove that if $aa^\top - bb^\top$ is small, then the vector $\frac{a}{\norm{a}}$ is close to $\pm \frac{b}{\norm{b}}$.

\begin{claim*}[Restatement of \Cref{claim:close-upto-sign}]\claimcloseuptosigntext
\end{claim*}

\begin{proof}
Note that $\min\{\norm{u+v}^2,\norm{u-v}^2\} = 2 - 2\abs{\iprod{u,v}} \leq 2$ and thus if $\beta \geq \frac{\norm{a}^2}{2}$, then the conclusion trivially holds.
Thus, from now on, we assume that $\beta < \frac{\norm{a}^2}{2}$.

Since both the assumption and the conclusion are invariant under replacing $b$ by $-b$, so we can without loss of generality assume that $\iprod{a,b} \geq 0$.
Let $\theta \in [0, \frac{\pi}{2}]$ be the angle between $a,b$, or equivalently, $\iprod{a,b}=\norm{a}\norm{b}\cos\theta$.
Then observe that 
\begin{align}
    \Norm{aa^\top-bb^\top}_F^2&=\Norm{aa^\top}_F^2+\Norm{bb^\top}_F^2-2\iprod{aa^\top,bb^\top}\notag\\
    &=\norm{a}^4 + \norm{b}^4-2\iprod{a,b}^2\label{eqn:app:aatop-bbtop}\\
    &= \norm{a}^4 + \norm{b}^4 - 2 \norm{a}^2\norm{b}^2 \cos^2\theta\notag\\
    &=\Paren{\norm{a}^2-\norm{b}^2}^2+2\norm{a}^2\norm{b}^2(1-\cos^2\theta)\notag\\
    &=\Paren{\norm{a}^2-\norm{b}^2}^2+2\norm{a}^2\norm{b}^2\sin^2\theta\notag
\end{align}
Both terms are non-negative.
Since the left-hand side is at most $\beta^2$, we get that $\lvert\norm{a}^2-\norm{b}^2\rvert\leq\beta$ and $\sqrt{2} \cdot \norm{a}\norm{b}\sin\theta\leq\beta$.
From the first one we obtain
\[
    \norm{a}^2-\norm{b}^2\leq\beta<\frac{\norm{a}^2}{2} \quad\Longrightarrow\quad \norm{b}\geq\frac{\norm{a}}{\sqrt{2}}.
\]
Using this lower bound on the second term, we get
\[
    \norm{a}^2\sin\theta\leq \sqrt{2} \cdot \norm{a}\norm{b}\sin\theta\leq\beta \quad\Longrightarrow\quad \sin\theta\leq\frac{\beta}{\norm{a}^2}.
\]
Combined with $\beta<\norm{a}^2/2$, we have $\sin\theta<1/2$, which implies $\theta<\frac{\pi}{6}$.
Finally, note that ${\min\{\norm{u+v}^2,\norm{u-v}^2\}=\norm{u-v}^2}$ since we assumed $\iprod{u,v} \geq 0$.
Thus, we get that 
\begin{align*}
    \min\{\norm{u+v}^2,\norm{u-v}^2\}=\norm{u-v}^2=2-2\cos\theta=4\sin^2\Paren{\frac{\theta}{2}}\overset{(\text{i})}{\leq}2\sin^2\theta \leq 2\frac{\beta^2}{\norm{a}^4}
\end{align*}
where in inequality (i) we used that for $0\leq\theta<\frac{\pi}{6}$, $\sin\frac{\theta}{2}\leq\frac{1}{\sqrt{2}}\sin\theta$.
\end{proof}

Finally, we show how to relate $aa^\top-bb^\top$ to $a-b$ and $a+b$.

\begin{fact*}[Restatement of \Cref{fact:upper-bound-Frobenius}]\factupperboundFrobeniustext
\end{fact*}
\begin{proof}
    Recall from \eqref{eqn:app:aatop-bbtop} that 
    \[
        \Norm{aa^\top-bb^\top}_F^2 = \norm{a}^4+\norm{b}^4-2\iprod{a,b}^2.
    \]
    The proof now follows by the following calculation:
    \begin{align*}
        \Norm{aa^\top-bb^\top}_F^2 &=\norm{a}^4+\norm{b}^4-4\iprod{a,b}^2+2\iprod{a,b}^2\\
        &\leq\norm{a}^4+\norm{b}^4-4\iprod{a,b}^2+2\norm{a}^2\norm{b}^2\\
        &=\Paren{\norm{a}^2 + \norm{b}^2 + 2\iprod{a,b}} \cdot \Paren{\norm{a}^2 + \norm{b}^2 - 2\iprod{a,b}} \\
        &=\norm{a+b}^2\cdot\norm{a-b}^2.
    \end{align*}
\end{proof}

\subsection{Concentration results and error analysis for sample moments}\label{sec:erroranalysissamplesmoments}
In this section, we prove two claims that we deferred from \Cref{sec:fullproof} about the error due to the fact that we only have access to the \emph{sample} moments.

\begin{claim*}[Restatement of \Cref{claim:sample-mean-noise}]\claimsamplemeannoise    
\end{claim*}

\begin{proof}
For every $i \in [n]$, we have that
\begin{align*}
    \Iprod{v,x_i-\bar{x}}^3 &= \Paren{\Iprod{v,x_i-\nu} + \Iprod{v,\nu -\bar{x}}}^3\\
    &= \Iprod{v,x_i-\nu}^3 + 3 \Iprod{v,x_i-\nu}^2\Iprod{v,\nu -\bar{x}} + 3\Iprod{v,x_i-\nu}\Iprod{v,\nu -\bar{x}}^2 + \Iprod{v,\nu -\bar{x}}^3.
\end{align*}
Thus, summing over all $i \in [n]$, we get that
\[
    \Abs{\frac{1}{n}\sum_{i=1}^{n} \iprod{v, x_i-\bar{x}}^3-\iprod{v,x_i-\nu}^3} \leq 3 \Abs{\Iprod{v,\bar{x} - \nu}} \cdot \Abs{\frac{1}{n}\sum_{i=1}^{n}\Iprod{v,x_i-\nu}^2} + 4\Abs{\Iprod{v,\bar{x} - \nu}}^3,
\]
where we used that $\sum_{i=1}^n \Iprod{v,x_i-\nu} = \Iprod{v,\bar{x}-\nu}$. Thus, it remains to bound, with high probability, the two terms $\abs{\Iprod{v,\bar{x} - \nu}}$ and $\abs{\frac{1}{n}\sum_{i=1}^{n}\iprod{v,x_i-\nu}^2}$.
We do so separately.

\textbf{Bounding $\abs{\Iprod{v,\bar{x} - \nu}}$.}
By \Cref{cor:subgaussian:X}, for every $i \in [n]$, $\iprod{v,x_i-\nu}$ is sub-Gaussian with parameter $\sigma = \norm{\Sigma^{1/2}} = O\Paren{1/\sqrt{\kappa_2(\gamma)}}$ , where we used \Cref{lem:preconditioner} to bound $\norm{\Sigma}$.
Thus, for any fixed vector $v$, Hoeffding's inequality (cf. \cite[Theorem 2.7.3]{Vershynin_2018}) implies that for every $t > 0$ we have
\[
    \Pr\Brac{|\iprod{v,\bar{x}-\nu}|\geq t}\leq 2\exp\Paren{-\Omega\Paren{nt^2\cdot\kappa_2(\gamma)}}.
\]
In particular, taking $t^2=O\Paren{\frac{d}{\kappa_2(\gamma)n}}$, we get that with probability at least $1-2e^{-10d}$ we have 
\[
    \abs{\iprod{v,\bar{x}-\nu}}\lesssim\sqrt\frac{ d}{\kappa_2(\gamma)n}.
\]

\textbf{Bounding $\abs{\frac{1}{n}\sum_{i=1}^{n}\iprod{v,x_i-\nu}^2}$.}
For every $i \in [n]$, since $\iprod{v,x_i-\nu}$ is $O(1/\kappa_2(\gamma)^{1/2})$-sub-Gaussian, $\iprod{v,x_i-\nu}^2$ is $O(1/\kappa_2(\gamma))$-sub-exponential.
Thus, for fixed $v$, the Bernstein inequality (cf. \cite[Theorem 2.9.1]{Vershynin_2018}) implies that for every $t > 0$, we have
\[
    \Pr\Brac{\Abs{\frac{1}{n}\sum_{i=1}^{n} \iprod{v,x_i-\nu}^2-\E\iprod{v,x_1-\nu}^2}\geq t}\leq 2\exp\Paren{-\Omega\Paren{\min\Set{nt^2\cdot\kappa_2(\gamma)^2, nt\cdot\kappa_2(\gamma)}}}.
\]
Taking $t=O\Paren{\frac{1}{\kappa_2(\gamma)}\cdot\Paren{\sqrt\frac{d}{n}+\frac{d}{n}}}$, it then follows that, with probability at least $1-2e^{-10d}$, we have that
\[
    \Abs{\frac{1}{n}\sum_{i=1}^{n} \iprod{v,x_i-\nu}^2}\lesssim \frac{1}{\kappa_2(\gamma)}\Paren{\sqrt\frac{d}{n}+\frac{d}{n}}+\E\iprod{v,x_1-\nu}^2. 
\]
Since $\E \langle v, x_1 - \mu \rangle^2 \leq O\left(\frac{1}{\kappa_2(\gamma)}\right)$ because of sub-Gaussianity, we have with probability at least $1 - 2e^{-10d}$ that
\[
    \Abs{\frac{1}{n}\sum_{i=1}^{n} \iprod{v,x_i-\nu}^2} \lesssim \frac{1}{\kappa_2(\gamma)}\Paren{\sqrt\frac{d}{n}+\frac{d}{n}} + \frac{1}{\kappa_2(\gamma)} \lesssim \frac{1}{\kappa_2(\gamma)}\Paren{1+\frac{d}{n}}.
\]

Taking a union bound over these two events, we get that, with probability $1-4e^{-10d}$,
\[
    \Abs{\frac{1}{n}\sum_{i=1}^{n} \iprod{v, x_i-\bar{x}}^3-\iprod{v,x_i-\nu}^3} \lesssim \kappa_2(\gamma)^{-3/2} \cdot \Paren{\sqrt{\frac{d}{n}} + \Paren{\frac{d}{n}}^{3/2}},
\]
which completes the proof.
\end{proof}

We now prove \Cref{lem:tensor-spectral-concentration}.
Recall that in \Cref{sec:fullproof} we already proved the statement for fixed $v$.
So, it remains to bound the supremum over \emph{all} unit vectors $v$.

\begin{lemma*}[Restatement of \Cref{lem:tensor-spectral-concentration}]\lemtensorspectralconcentrationtext
\end{lemma*}

\begin{proof}
In \Cref{sec:fullproof}, we showed that for every fixed $v \in \cS^{d-1}$ we have
\[
    \Abs{\Iprod{v^{\otimes 3},\widehat{M}_3-M_3}} \leq K\cdot\kappa_2(\gamma)^{-3/2}\left(\sqrt{\frac{d}{n}}+\frac{d^{3/2}}{n}\right)
\]
with probability $1-6e^{-10d}$, where $K > 0$ is a universal constant (independent of $v$). 
Let 
\[
T:=\widehat M_3-M_3.
\]
Since both \(\widehat M_3\) and \(M_3\) are symmetric order-\(3\) tensors, \(T\) is symmetric as well. Fix a \(\delta\)-net \({\cN_\delta\subseteq \cS^{d-1}}\) of the unit sphere (for some $\delta$ to be chosen later).
We will show that for any $v \in \R^d$, $\abs{\iprod{{v}^{\otimes 3},T}}$ is bounded by $\sup_{u \in \cN_\delta}\abs{\iprod{{u}^{\otimes 3},T}}$ and finish the proof by a union-bound argument over~$\cN_\delta$.

For any \(v\in \cS^{d-1}\), choose \(v'\in \cN_\delta\) such that \(\|v-v'\|\le \delta\) and write \(\Delta:=v-v'\).
By definition, we have \(v=\Delta+v'\), and expanding \((\Delta+v')^{\otimes 3}\) gives
\[
(\Delta+v')^{\otimes 3}
=
\Delta^{\otimes 3}
+\mathrm{Sym}(\Delta^{\otimes 2}\otimes v')
+\mathrm{Sym}(\Delta\otimes (v')^{\otimes 2})
+(v')^{\otimes 3},
\]
where we use $\mathrm{Sym}$ to denote all permutations of a tensor (so, for example, ${\mathrm{Sym}(\Delta^{\otimes 2}\otimes v')}$ is equal to ${\Delta \otimes \Delta \otimes v' + \Delta \otimes v' \otimes \Delta + v' \otimes \Delta \otimes \Delta}$).
Taking inner product with \(T\), and using the symmetry of \(T\), we get
\begin{align*}
\iprod{v^{\otimes 3},T}
&=\iprod{(\Delta+v')^{\otimes 3},T}\\
&=\iprod{\Delta^{\otimes 3},T}
+3\,\iprod{\Delta^{\otimes 2}\otimes v',T}
+3\,\iprod{\Delta\otimes (v')^{\otimes 2},T}
+\iprod{(v')^{\otimes 3},T}.
\end{align*}
Now, we bound the first three terms in terms of the injective tensor norm. Since \(T\) is symmetric,
\[
\sup_{\|a\|=\|b\|=\|c\|=1}\iprod{a\otimes b\otimes c,T}
=
\sup_{\|u\|=1}\iprod{u^{\otimes 3},T}.
\]
Therefore, for any \(a,b,c\in \mathbb R^d\), we have
\[
\iprod{a\otimes b\otimes c,T}
\le
\|a\|\,\|b\|\,\|c\|
\sup_{\|u\|=1}\iprod{u^{\otimes 3},T}.
\]
Applying this with \((a,b,c)=(\Delta,\Delta,\Delta)\), \((\Delta,\Delta,v')\), and \((\Delta,v',v')\), we obtain
\begin{align*}
    \iprod{v^{\otimes 3},T} &\le  \Paren{\|\Delta\|^3 + 3\|\Delta\|^2\|v'\| + 3\|\Delta\|\|v'\|^2} \cdot \sup_{\|u\|=1}\iprod{u^{\otimes 3},T}  + \iprod{(v')^{\otimes 3},T}\\
    &\le \Paren{\delta^3 + 3\delta^2 + 3\delta} \cdot \sup_{\|u\|=1}\iprod{u^{\otimes 3},T}  + \iprod{(v')^{\otimes 3},T}\\
    &\le \Paren{\delta^3 + 3\delta^2 + 3\delta} \cdot \sup_{\|u\|=1}\iprod{u^{\otimes 3},T}  + \sup_{u\in \cN_\delta}\iprod{u^{\otimes 3},T},
\end{align*}
where in the second step we used \(\|v'\|=1\) (since \(v'\in \cN_\delta\subseteq \cS^{d-1}\)) and $\|\Delta\| = \|v - v'\| \leq \delta$.
This holds for every \(v\in \cS^{d-1}\).
Taking the supremum over \(v\in \cS^{d-1}\) yields
\[
    \sup_{\|v\|=1}\iprod{v^{\otimes 3},T} \le \Paren{\delta^3+3\delta^2+3\delta} \cdot \sup_{\|v\|=1}\iprod{v^{\otimes 3},T} + \sup_{u\in \cN_\delta}\iprod{u^{\otimes 3},T}.
\]
Rearranging, we get that
\[
    \Paren{1-3\delta-3\delta^2-\delta^3} \cdot \sup_{\|v\|=1}\iprod{v^{\otimes 3},T} \le \sup_{u\in \cN_\delta}\iprod{u^{\otimes 3},T}.
\]
Now, we set \(\delta=\frac14\). Then, we have that
\[
    1-3\delta-3\delta^2-\delta^3 = 1-\frac34-\frac{3}{16}-\frac{1}{64} = \frac{3}{64},
\]
and thus
\[
    \sup_{\|v\|=1}\iprod{v^{\otimes 3},T} \le \frac{64}{3}\sup_{u\in \cN_{1/4}}\iprod{u^{\otimes 3},T}.
\]
Moreover, there exists a \(1/4\)-net \(\cN_{1/4}\subseteq \cS^{d-1}\) with cardinality \(|\cN_{1/4}|\le 9^d\) \cite[Corollary 4.2.11]{Vershynin_2018}.
Applying our fixed direction bound to each fixed \(u\in \cN_{1/4}\), and then taking a union bound, we get
\begin{align*}
&\Pr\left[\sup_{\|v\|=1}\iprod{v^{\otimes 3},\widehat M_3-M_3}\ge\frac{64K}{3}\kappa_2(\gamma)^{-3/2}\left(\sqrt{\frac dn}+\frac{d^{3/2}}{n}\right)
\right]\\
&\qquad\le \Pr\left[\sup_{u\in \cN_{1/4}}\iprod{u^{\otimes 3},\widehat M_3-M_3} \ge K\kappa_2(\gamma)^{-3/2}\left(\sqrt{\frac dn}+\frac{d^{3/2}}{n}\right)\right]\\
&\qquad\le |\cN_{1/4}|\cdot 6e^{-10d} \le 9^d\cdot 6e^{-10d} \le 6e^{-7d}.
\end{align*}
This proves the lemma.
\end{proof}

\section{Computations}
\label{APP:computations}
In this appendix, we prove some of the results used for the proof of \Cref{lemma:monotonicity-of-skewness}.
Furthermore, we will give a proof of \Cref{CLAIM:ddtskewoverddtkappa2,CLAIM:ddtskewoverddtpsi,CLAIM:ddtkappa2overddtkappa1}.
Finally, we provide a proof for a fact used in \Cref{APP:facts}.
All of the proofs here are (partially) computer-assisted in that we provide Mathematica code to verify some properties of certain functions.

Recall that for \Cref{lemma:monotonicity-of-skewness} we there used the function
\[
   p_t(x)=2t^3-6t+(t^4+14t^2-5)x+(2t^3+20t)x^2+(t^2+8)x^3.
\]
with derivative
\[
    \Paren{\frac{\partial}{\partial x}p_t}(x)=t^4+14t^2-5+4(t^3+10t)\cdot x+3(t^2+8)\cdot x^2.
\]

\begin{fact}\label{FACT:computationmillsratio1}
    Define $x_\ell(t) = \frac{e^{-t^2/2}}{\sqrt{\pi/2}+t}$.
    For any $t \geq 0$, we have that
    \[
        \Paren{\frac{\partial}{\partial x}p_t}(x_\ell(t)) > 0
    \]
    and
    \[
        p_t(x_\ell(t)) > 0.
    \]
\end{fact}
\begin{proof}
    This can be verified with the following Mathematica code:
\begin{lstlisting}
lb[t_] = Exp[-t^2/2]/(Sqrt[Pi/2] + t);
Reduce[t^4 + 14 t^2 - 5 + 4 (t^3 + 10 t) lb[t] + 3 (t^2 + 8) lb[t]^2 > 0]
Reduce[2 t^3 - 6 t + (t^4 + 14 t^2 - 5) lb[t] + (2 t^3 + 20 t) lb[t]^2 + (t^2 + 8) lb[t]^3 > 0]
\end{lstlisting}
     This shows that the two statements are true for $t \neq -\sqrt{\frac{\pi}{2}}$ and $t > -\sqrt{\frac{\pi}{2}}$ respectively and thus in particular for $t \geq 0$.
\end{proof}

\begin{fact}\label{FACT:computationmillsratio2}
    Define
    \[
        x_\ell(t) = \frac{\sqrt{(\pi-2)^2t^2+2\pi}-2t}{\pi} \quad \text{and} \quad x_u(t) = \frac{\sqrt{t^2+2\pi}-(\pi-1)t}{\pi}.
    \]
    For $-\frac{4}{10} < t \leq 0$, we have
    \[
        t^4+14t^2-5+4(t^3+10t)\cdot x_u(t)+3(t^2+8)\cdot x_\ell(t)^2 > 0
    \]
    and
    \[
        p_t(x_\ell(t)) > 0.
    \]
\end{fact}
\begin{proof}
    This can be verified with the following Mathematica code:
\begin{lstlisting}
xl[t_] = (Sqrt[(Pi - 2)^2 t^2 + 2 Pi] - 2 t)/Pi;
xu[t_] = (Sqrt[t^2 + 2 Pi] - (Pi - 1) t)/Pi;
Reduce[t^4 + 14 t^2 - 5 + 4 (t^3 + 10 t) xu[t] + 3 (t^2 + 8) xl[t]^2 > 0]
N[%]
Reduce[2 t^3 - 6 t + (t^4 + 14 t^2 - 5) xl[t] + (2 t^3 + 20 t) xl[t]^2 + (t^2 + 8) xl[t]^3 > 0]
N[%]
\end{lstlisting}
This shows that the statements are true for $-1.83192 < t < 4.5716$ and $-0.419026 < t < 0.819969$ respectively and thus in particular for $ -\frac{4}{10}<t\leq 0$.
\end{proof}

\begin{fact}\label{FACT:computationmillsratio3}
    Define
    \[
        r_{k+1}(x) \coloneqq {x+\frac{1}{x+\frac{2}{x+\frac{3}{\underset{x+\frac{k}{x}}{\ddots}}}}}.
    \]
    Let $x_\ell(t) = r_{81}(-t)$ and $x_u(t) = r_{20}(-t)$.
    For $t \leq -\frac{4}{10}$, we have $x_\ell(t) > 0$,
    \begin{equation}\label{EQ:APP:computationmillsratio3:eq1}
        t^4+14t^2-5+4(t^3+10t)\cdot x_u(t)+3(t^2+8)\cdot x_\ell(t)^2 > 0        
    \end{equation}
    and
    \begin{equation}\label{EQ:APP:computationmillsratio3:eq2}
        p_t(x_\ell(t)) > 0.
    \end{equation}
\end{fact}
\begin{proof}
    To prove this fact, we additionally define the following function $x_\ell'(t) = r_{27}(-t)$.
    We claim that $x_\ell(t) \geq x_\ell'(t) > 0$ (which will also show $x_\ell(t) > 0$) and thus, since $3(t^2+8) > 0$, it is sufficient to show 
    \begin{equation}\label{EQ:APP:computationmillsratio3:eq3}
        t^4+14t^2-5+4(t^3+10t)\cdot x_u(t)+3(t^2+8)\cdot x_\ell'(t)^2.
    \end{equation}
    in order to show \eqref{EQ:APP:computationmillsratio3:eq1}.
    We can verify $x_\ell(t) \geq x_\ell'(t) > 0$, \eqref{EQ:APP:computationmillsratio3:eq3} and \eqref{EQ:APP:computationmillsratio3:eq2} with the following Mathematica code:
\begin{lstlisting}
r[k_Integer, x_] := Module[{result = x, i}, For[i = k - 1, i >= 1, i--, result = x + i/result]; result]
Reduce[r[81, -t] >= r[27, -t] > 0]
Reduce[t^4 + 14 t^2 - 5 + 4 (t^3 + 10 t) r[20, -t] + 3 (t^2 + 8) r[27, -t]^2 > 0]
N[%]
Reduce[2 t^3 - 6 t + (t^4 + 14 t^2 - 5) r[81, -t] + (2 t^3 + 20 t) r[81, -t]^2 + (t^2 + 8) r[81, -t]^3 > 0]
N[%]
\end{lstlisting}
This shows that the three statements we need to verify are true for $t < 0$, $t < -0.271834$ and $t < -0.370331$ respectively and thus in particular for $t < -\frac{4}{10}$.
\end{proof}
\begin{fact}\label{FACT:computationmillsratio4}
    Define $\kappa_3(t)=\frac{\phi(t)}{\Phi(t)} \cdot \left( 1 - t^2 - 3 t \cdot \frac{\phi(t)}{\Phi(t)} - 2 \cdot \left(\frac{\phi(t)}{\Phi(t)}\right)^2 \right)$.\\
    Then for $-2<t<2$ we have $\abs{\kappa_3(t)}\geq\frac{1}{20}$.
\end{fact}
\begin{proof}
    This can be verified with the following Mathematica code:
\begin{lstlisting}
phi[t_] = PDF[NormalDistribution[0, 1], t];
pphi[t_] = CDF[NormalDistribution[0, 1], t];
kappa3[t_] = 
  phi[t]/pphi[t]*(1 - t^2 - 3 t*phi[t]/pphi[t] - 2 (phi[t]/pphi[t])^2);
Minimize[{-kappa3[t], -2 <= t <= 2}, t];
Reduce[First[%] > 1/20]
\end{lstlisting}
This shows that the statement is true for $-2<t<2$.
\end{proof}
Now, we give a proof of \Cref{claim:skew-atmost2,CLAIM:ddtskewoverddtkappa2,CLAIM:ddtskewoverddtpsi,CLAIM:ddtkappa2overddtkappa1} that we restate below.

\begin{claim*}[Restatement of \Cref{claim:skew-atmost2}]\claimskewatmosttwotext
\end{claim*}
\begin{proof}
    We can compute these limits with the following Mathematica code:
\begin{lstlisting}
phi[t_] = PDF[NormalDistribution[0, 1], t];
Phi[t_] = CDF[NormalDistribution[0, 1], t];
kappa2[t_] = 1 - t phi[t]/Phi[t] - (phi[t]/Phi[t])^2;
kappa3[t_] = 
  phi[t]/Phi[t] (1 - t^2 - 3 t phi[t]/Phi[t] - 2 (phi[t]/Phi[t])^2);
skew[t_] = kappa3[t]/kappa2[t]^(3/2);
Limit[skew[t], t -> -\[Infinity]]
Limit[skew[t], t -> \[Infinity]]
\end{lstlisting}
    This shows that $\lim_{t\to-\infty} \sk(t) = -2$ and $\lim_{t\to+\infty} \sk(t) = 0$.
    Together with the monotonicity of $\sk$ (cf. \Cref{lemma:monotonicity-of-skewness}), this implies that $\sk(t) \in [-2,0]$ and in particular $|\sk(t)| \leq 2$ for all $t \in \R$.
\end{proof}

\begin{claim*}[Restatement of \Cref{CLAIM:ddtskewoverddtkappa2}]\claimddtskewoverddtkappatwotext
\end{claim*}
\begin{proof}
    We define the function $g(t) \coloneqq \frac{\ddt\sk(t)}{|\ddt\kappa_2(t)|}$.
    Note that by \Cref{lemma:monotonicity-of-skewness}, we have $g(t) > 0$ point-wise for every $t \in \R$. 
    Our goal is to show that also $\inf_{t \in \R} g(t) > 0$.
    The Mathematica code below can be used to compute the two limits $\lim_{t \to \infty} g(t)$ and $\lim_{t \to -\infty} g(t)$.
\begin{lstlisting}
phi[t_] = PDF[NormalDistribution[0, 1], t];
Phi[t_] = CDF[NormalDistribution[0, 1], t];
kappa2[t_] = 1 - t phi[t]/Phi[t] - (phi[t]/Phi[t])^2;
kappa3[t_] = phi[t]/Phi[t] (1 - t^2 - 3 t phi[t]/Phi[t] - 2 (phi[t]/Phi[t])^2);
skew[t_] = kappa3[t]/kappa2[t]^(3/2);
dskew[t_] = D[skew[t], t];
dkappa2[t_] = D[kappa2[t], t];
Limit[dskew[t]/Abs[dkappa2[t]], t -> -\[Infinity]]
Limit[dskew[t]/Abs[dkappa2[t]], t -> \[Infinity]]
\end{lstlisting}
    This shows that $\lim_{t \to \infty} g(t) = 6$ and $\lim_{t \to -\infty} g(t) = \infty$.
    Thus, there exist $t_-$ and $t_+$ such that for all $t \not\in [t_-,t_+]$ we have $g(t) \geq 5$.
    Finally, since the function is continuous, the minimum of $g$ over the interval $[t_-,t_+]$ is achieved and thus there exists a constant $C' > 0$ such that also $g(t) \geq C'$ for all $t \in [t_-,t_+]$.
    Hence, we have that $\inf_{t \in \R} g(t) \geq \min(C',5)$. Choosing $C = \min(C',5)$ gives the first part of the claim.
    For the second part, note that for any $a, b \in \R$ we have
    \begin{align*}
        |\kappa_2(a)  - \kappa_2(b)| &= \Abs{\int_{\min(a,b)}^{\max(a,b)}\ddt\kappa_2(s) \mathrm{d}s}\\
        &\leq \int_{\min(a,b)}^{\max(a,b)}\Abs{\ddt\kappa_2(s)} \mathrm{d}s\\
        &\leq \frac{1}{C} \cdot \int_{\min(a,b)}^{\max(a,b)}\ddt\sk(s) \mathrm{d}s\\
        &= \frac{1}{C} \cdot \sk(\max(a,b)) - \sk(\min(a,b))\\
        &= \frac{1}{C} \cdot \Abs{\sk(a) - \sk(b)}.
    \end{align*}
    In the third step, we used the first part of the claim.
    And in the final step, we used that $\sk$ is monotonic by \Cref{lemma:monotonicity-of-skewness}. 
    This shows the second part of the claim.
\end{proof}

\begin{claim*}[Restatement of \Cref{CLAIM:ddtskewoverddtpsi}]\claimddtskewoverddtpsitext
\end{claim*}
\begin{proof}
    We define the function $g(t) = \frac{\ddt \sk(t)}{\abs{\ddt \psi(t)} \cdot \kappa_2(t)^2}$.
    As in the proof of \Cref{CLAIM:ddtskewoverddtkappa2}, we have by \Cref{lemma:monotonicity-of-skewness} that $g(t) > 0$ point-wise and we want to show that also the limits $\lim_{t \to \infty} g(t)$ and $\lim_{t \to -\infty} g(t)$ are strictly positive. The following Mathematica code shows that $\lim_{t \to \infty} g(t) = \infty$ and $\lim_{t \to -\infty} g(t) = 6$.
\begin{lstlisting}
phi[t_] = PDF[NormalDistribution[0, 1], t];
Phi[t_] = CDF[NormalDistribution[0, 1], t];
kappa1[t_] = -(phi[t]/Phi[t]);
kappa2[t_] = 1 - t phi[t]/Phi[t] - (phi[t]/Phi[t])^2;
kappa3[t_] = 
  phi[t]/Phi[t] (1 - t^2 - 3 t phi[t]/Phi[t] - 2 (phi[t]/Phi[t])^2);
skew[t_] = kappa3[t]/kappa2[t]^(3/2);
dskew[t_] = D[skew[t], t];
psi[t_] = kappa1[t]/kappa2[t]^(1/2);
dpsi[t_] = D[psi[t], t];
Limit[dskew[t]/Abs[dpsi[t]*kappa2[t]^2], t -> -\[Infinity]]
Limit[dskew[t]/Abs[dpsi[t]*kappa2[t]^2], t -> \[Infinity]]
\end{lstlisting}
    Thus, as in the proof of \Cref{CLAIM:ddtskewoverddtkappa2}, since $g$ is continuous, we get that there exists a constant $C > 0$ such that $g(t) \geq C$ for all $t \in \R$.
    For the second part, we argue analogous to the proof of \Cref{CLAIM:ddtskewoverddtkappa2}. Since $\kappa_2$ is monotonically increasing by \Cref{claim:monotonicity-of-kappa2}, we have $\kappa_2(s)^{-2} \leq \max\{\kappa_2(a)^{-2}, \kappa_2(b)^{-2}\}$ for every $s$ between $a$ and $b$.
    Thus, we get that
    \begin{align*}
        |\psi(a)  - \psi(b)| &\leq \int_{\min(a,b)}^{\max(a,b)} \Abs{\ddt\psi(s)} \mathrm{d}s\\
        &\leq \frac{1}{C} \cdot \int_{\min(a,b)}^{\max(a,b)} \ddt\sk(s) \cdot \kappa_2(s)^{-2} \mathrm{d}s\\
        &\leq \frac{1}{C} \cdot \max\{\kappa_2(a)^{-2}, \kappa_2(b)^{-2}\}\int_{\min(a,b)}^{\max(a,b)} \ddt\sk(s) \cdot \mathrm{d}s\\
        &\leq \frac{1}{C} \cdot \max\{\kappa_2(a)^{-2}, \kappa_2(b)^{-2}\} \cdot \Abs{\sk(a)-\sk(b)},
    \end{align*}
    which completes the proof.    
\end{proof}

\begin{claim*}[Restatement of \Cref{CLAIM:ddtkappa2overddtkappa1}]\claimddtkappatwooverddtkappaonetext
\end{claim*}
\begin{proof}
    We define the function $g(t) = \frac{\ddt \kappa_2(t)}{\abs{\ddt \kappa_1(t)} \cdot \kappa_2(t)^2}$. Analogous as for \Cref{CLAIM:ddtskewoverddtkappa2,CLAIM:ddtskewoverddtpsi}, we have by \Cref{claim:monotonicity-of-kappa2} that $g(t) > 0$ point-wise and we want to show that also the two limits $\lim_{t \to \infty} g(t)$ and $\lim_{t \to -\infty} g(t)$ are strictly positive. With the following Mathematica code we get that $\lim_{t \to \pm \infty} g(t) = \infty$.
\begin{lstlisting}
phi[t_] = PDF[NormalDistribution[0, 1], t];
Phi[t_] = CDF[NormalDistribution[0, 1], t];
kappa1[t_] = -(phi[t]/Phi[t]);
kappa2[t_] = 1 - t phi[t]/Phi[t] - (phi[t]/Phi[t])^2;
dkappa1[t_] = D[kappa1[t], t];
dkappa2[t_] = D[kappa2[t], t];
Limit[dkappa2[t]/Abs[dkappa1[t]*kappa2[t]^ 2], t -> -\[Infinity]]
Limit[dkappa2[t]/Abs[dkappa1[t]*kappa2[t]^2], t -> \[Infinity]]
\end{lstlisting}
    As for the proof of \Cref{CLAIM:ddtskewoverddtpsi}, this implies that $g(t) \geq C$ for all $t \in \R$ for some constant $C > 0$.
    Furthermore, also as in \Cref{CLAIM:ddtskewoverddtpsi}, we get that
    \begin{align*}
        \abs{\kappa_1(a) - \kappa_1(b)} &\leq \frac{1}{C} \int_{\min(a,b)}^{\max(a,b)} \ddt\kappa_2(s) \cdot \kappa_2(s)^{-2} \mathrm{d}s\\
        &\leq \frac{1}{C} \cdot \max\{\kappa_2(a)^{-2}, \kappa_2(b)^{-2}\} \cdot \Abs{\kappa_2(a)-\kappa_2(b)},
    \end{align*}
    which completes the proof.
\end{proof}

Finally, we prove the following fact that we used in \Cref{APP:facts}.
\begin{fact}\label{FACT:computationkappa2monotonic}
    The function
    \[
        g(t) = t^2-1+\frac{\exp\Paren{-t^2/2}}{\sqrt{\pi/2}+t}\cdot\Paren{3t+2\frac{\exp\Paren{-t^2/2}}{\sqrt{\pi/2}+t}}
    \]
    is monotonically increasing on the interval $(0,1)$.
\end{fact}
\begin{proof}
    This can be verified by the following Mathematica code, which confirms that $\ddt g(t) > 0$ for all~$t \in (0,1)$.
\begin{lstlisting}
g[t_] = t^2 - 1 + 
Exp[-t^2/2]/(Sqrt[Pi/2] + t)*(3 t + 2 Exp[-t^2/2]/(Sqrt[Pi/2] + t));
dg[t_] = D[g[t], t];
Reduce[{dg[t] > 0, 0 < t < 1}]
\end{lstlisting}
\end{proof}

\section{\texorpdfstring{Determining $w$}{Determining w}}\label{APP:howtogetw}
In this appendix, we show that for a vector $u \propto \Sigma w$, we can compute a vector proportional to $w$ by computing $M_2^{-1}u$.
This is not needed for our algorithm, but we include it for completeness (instead, this would be needed if one wants to estimate the survival set as they do in \cite{lee2024}).
Note that if $\Sigma$ is \textit{isotropic} (or more generally, $w$ is an \textit{eigenvector} of $\Sigma$), then we can learn the direction of $w$ from the third central moment.
For arbitrary $w$ and $\Sigma$, we get the following (where we assume $\Sigma \succ 0$; otherwise we cannot invert $M_2$).

\begin{claim}
    Assume $\Sigma \succ 0$.
    Then the second central moment $M_2$ of the truncated distribution is invertible and for $u\propto \Sigma w$ we have $\widetilde{w} \coloneqq M_2^{-1}u \propto w$.
\end{claim}
\begin{proof}
    From \Cref{lem:first3moments} we know
    \[
        M_2=\Sigma-(1-\kappa_2(\gamma)) \Paren{\frac{\Sigma w}{\vnorm{\Sigma^{1/2}w}}}\Paren{\frac{\Sigma w}{\vnorm{\Sigma^{1/2}w}}}^\top,
    \]
    so we get
    \[
        M_2 w=\Sigma w-(1-\kappa_2(\gamma))\Sigma w=\kappa_2(\gamma) \Sigma w,
    \]
    where we used that $w^\top \Sigma w = \vnorm{\Sigma^{1/2} w}^2$.
    
    We now show that $M_2$ is invertible.
    Let $\Set{v_i}_{i=1}^{d-1}$ be an orthonormal basis for the orthogonal complement of $\Sigma^{1/2}w$ with respect to $\R^d$.
    Recall from \Cref{claim:bound-on-kappa2} that we have $\kappa_2(t)\in (0,1]$.
    It follows that
    \begin{align*}
        &I_d-(1-\kappa_2(\gamma))\Paren{\frac{\Sigma^{1/2} w}{\vnorm{\Sigma^{1/2}w}}}\Paren{\frac{\Sigma^{1/2} w}{\vnorm{\Sigma^{1/2}w}}}^\top\\
        & \quad =\sum_{j=1}^{d-1}v_jv_j^T+\Paren{\frac{\Sigma^{1/2} w}{\vnorm{\Sigma^{1/2}w}}}\Paren{\frac{\Sigma^{1/2} w}{\vnorm{\Sigma^{1/2}w}}}^\top - (1-\kappa_2(\gamma))\Paren{\frac{\Sigma^{1/2} w}{\vnorm{\Sigma^{1/2}w}}}\Paren{\frac{\Sigma^{1/2} w}{\vnorm{\Sigma^{1/2}w}}}^\top\\
        & \quad =\kappa_2(\gamma)\cdot\Paren{\frac{\Sigma^{1/2} w}{\vnorm{\Sigma^{1/2}w}}}\Paren{\frac{\Sigma^{1/2} w}{\vnorm{\Sigma^{1/2}w}}}^\top+\sum_{j=1}^{d-1}v_jv_j^T\\
        & \quad \succeq \kappa_2(\gamma)\Paren{\Paren{\frac{\Sigma^{1/2} w}{\vnorm{\Sigma^{1/2}w}}}\Paren{\frac{\Sigma^{1/2} w}{\vnorm{\Sigma^{1/2}w}}}^\top+\sum_{j=1}^{d-1}v_jv_j^T}=\kappa_2(\gamma)\cdot I_d
    \end{align*}
    where the last inequality follows since $\kappa_2(\gamma) \leq 1$.
    This implies that, using $\kappa_2(\gamma) > 0$,
    \begin{align*}
        M_2 &= \Sigma^{1/2}\Paren{I_d-(1-\kappa_2(\gamma))\Paren{\frac{\Sigma^{1/2} w}{\vnorm{\Sigma^{1/2}w}}}\Paren{\frac{\Sigma^{1/2} w}{\vnorm{\Sigma^{1/2}w}}}^\top}\Sigma^{1/2}\\
        &\succeq\Sigma^{1/2}\kappa_2(\gamma)I_d\Sigma^{1/2} =\kappa_2(\gamma)\Sigma \succ 0.
    \end{align*}
    Thus $M_2$ is invertible and we obtain
    \[
        w=\kappa_2(\gamma)M_2^{-1}\Sigma w.
    \]
    Hence, if $u \propto \Sigma w$, then we get
    \[
        \widetilde{w}=M_2^{-1} u\propto M_2^{-1}\Sigma w\propto w,
    \]
    which completes the proof.
\end{proof}